\documentclass{CSML}

\def\dOi{11(3:19)2015}
\lmcsheading%
{\dOi}
{1--57}
{}
{}
{Aug.~18, 2014}
{Sep.~22, 2015}
{}

\ACMCCS{[{\bf Theory of computation}]: Logic---Logic and
  verification\,/\,Programming logic;  Semantics and reasoning---Program reasoning}

\usepackage{enumerate}
\usepackage{hyperref}
\usepackage{amssymb}
\usepackage{stmaryrd}
\usepackage{mathpartir}
\usepackage{color}
\usepackage{endnotes}
\usepackage{trfrac}

\makeatletter
\renewcommand*\makeenmark{\hbox{\textsuperscript{(\@alph{\theenmark})}}}
\makeatother

\renewcommand{\ss}{\hat{s}}
\newcommand{\sh}{\hat{h}}
\newcommand{\sell}{\hat{\ell}}
\newcommand{\sv}{\hat{v}}

\newcommand{\display}[1]{$$#1$$}

\DeclareMathAlphabet{\mathsfsl}{OT1}{cmss}{m}{sl}

\newcommand{\symi}{\mathsfsl{i}}
\newcommand{\symn}{\mathsfsl{n}}
\newcommand{\symr}{\mathsfsl{r}}
\newcommand{\symd}{\mathsfsl{d}}
\newcommand{\syml}{\mathsfsl{l}}
\newcommand{\symv}{\mathsfsl{v}}

\definecolor{Green}{rgb}{0.0,0.5,0.0}
\definecolor{White}{rgb}{1,1,1}
\definecolor{Blue}{rgb}{0.0,0.0,1}
\definecolor{Red}{rgb}{1,0,0}
\definecolor{Purple}{rgb}{0.4,0,0.6}
\definecolor{Orange}{rgb}{1,0.65,0}
\definecolor{Gray}{rgb}{0.5,0.5,0.5}
\newcommand{\annot}[1]{{\color{Blue}#1}}

\newcommand{\comment}[1]{{\color{Green}#1}}

\newcommand{\branching}[1]{{\color{Orange}#1}}

\newcommand{\llbrace}{\{\hspace{-3pt}[}
\newcommand{\rrbrace}{]\hspace{-3pt}\}}

\newenvironment{changemargin}[2]{}{}
\newenvironment{beamerframe}[1]{}{}

\theoremstyle{definition}\newtheorem{notation}[thm]{Notation}
\theoremstyle{plain}

\def\ie{{\em i.e.}}

\begin{document}

\title[Featherweight VeriFast]{Featherweight VeriFast}

\author[F.~Vogels]{Fr\'ed\'eric Vogels}	
\address{iMinds-DistriNet, Dept. C.S., KU Leuven, Celestijnenlaan 200A, 3001 Leuven, Belgium}	
\email{\{frederic.vogels, bart.jacobs, frank.piessens\}@gmail.com}  

\author[B.~Jacobs]{Bart Jacobs}	

\author[F.~Piessens]{Frank Piessens}	



\begin{abstract}
  \noindent VeriFast is a leading research prototype tool for the sound modular verification of 
  safety and correctness properties of
  single-threaded and multithreaded C and Java programs. It has been used 
  as a vehicle for exploration and validation of novel program 
  verification techniques and for industrial case studies;
  it has served well at a number of program 
  verification competitions; and it has been used for teaching by multiple teachers independent of the authors.
  
  However, until now, while VeriFast's operation has been described 
  informally in a number of publications, and specific verification techniques have been formalized,
  a clear and precise exposition of how VeriFast works has not 
  yet appeared.
  
  In this article we present for the first time a formal definition and 
  soundness proof of a core subset of the VeriFast program verification 
  approach. The exposition aims to be both accessible and rigorous: the text is 
  based on lecture notes for a graduate course on program verification, 
  and it is backed by an executable machine-readable definition and 
  machine-checked soundness proof in Coq.
\end{abstract}

\maketitle

\section*{Introduction}

For many classes of safety-critical or security-critical programs, such as 
operating system components, internet infrastructure, or embedded 
software, conventional quality assurance approaches such as testing, code 
review, or even model checking are insufficient to detect all bugs and 
achieve good confidence in their safety and security; for these programs, 
the newer technique of modular formal verification may be the most 
promising approach. 

VeriFast is a sound modular formal verification approach for 
single-threaded and multithreaded imperative programs being 
developed at KU~Leuven. The prototype tool that implements this 
approach\footnote{Download it from \url{http://distrinet.cs.kuleuven.be/software/VeriFast/}.} takes as input a C or Java program annotated with 
preconditions, postconditions, loop invariants, data structure 
definitions, and proof hints written in a variant of separation 
logic \cite{seplogic-ohearn,seplogic}, and symbolically executes each function/method. It 
either reports ``0 errors found'' or the source location of a 
potential error. If it reports ``0 errors found'', it is 
guaranteed (modulo bugs in the tool) that no execution of the program will a) perform an 
illegal memory access such as a null pointer dereference, 
an access of unallocated memory, or an access of an array outside of 
its bounds; b) perform a data race, where two threads access 
the same variable concurrently without synchronization, and at 
least one access is a write operation; c) violate the 
user-specified function/method contracts or the contracts of 
the library or API functions/methods used by the program. If it 
reports an error, it shows a symbolic execution trace that 
leads to the error, including the symbolic state (store, heap, 
and path condition) at each step. 

VeriFast has served as a vehicle for exploration and validation of a 
number of novel program verification techniques 
\cite{DBLP:conf/popl/JacobsP11,DBLP:conf/fm/JacobsSP11,DBLP:conf/forte/VogelsJPS11,DBLP:conf/sefm/VanspauwenJ13} 
and for a number of industrial case studies \cite{Philippaerts2013}; it 
has served well at a number of program verification 
competitions\footnote{See Endnote~(a).}\endnotetext[1]{Results of the VeriFast team:
\begin{center}
\begin{tabular}{| c | c | c |}
\hline
\textbf{Competition} & \textbf{Conference} & \textbf{Result}\\
\hline
1st Verified Software Competition \cite{DBLP:conf/fm/KlebanovMSLWAABCCHJLMPPRSTTUW11} & VSTTE 2010 & roughly tied with all other teams\\
2nd Verified Software Competition \cite{filliatre12compare} & VSTTE 2012 & score 570/600, rank 8\\
VerifyThis \cite{verifythis-tr} & FM 2012 & sole winner\\
\hline
\end{tabular} 
\end{center}
}; and it has been used 
for teaching program verification by the authors as well as by independent 
instructors at other institutions\footnote{See Endnote~(b).}\endnotetext[2]{by Jesper Bengtson (at ITU 
Copenhagen), Alexey Gotsman (at ENS Lyon), Dilian Gurov (at KTH 
Stockholm), and Stephan van Staden (at ETH Zurich)}. 

Until now, while VeriFast's operation has been described informally in a 
number of publications 
\cite{verifast,DBLP:conf/nfm/JacobsSPVPP11,DBLP:series/lncs/SmansJP13}, 
and specific verification techniques have been formalized 
\cite{DBLP:conf/popl/JacobsP11,DBLP:conf/fm/JacobsSP11,DBLP:conf/forte/VogelsJPS11,DBLP:conf/sefm/VanspauwenJ13}, 
a clear and precise exposition of how VeriFast works has not yet appeared.

In this article, we present a formal definition of a simplified version of 
the VeriFast program verification approach, called Featherweight VeriFast, 
as well as an outline for a proof of the soundness of this approach, 
\ie~that if verification of a program succeeds, then no execution of the 
program accesses unallocated memory. Featherweight VeriFast targets a 
simple toy programming language with routines, loops, and dynamic memory allocation 
and deallocation, and supports routine contracts, loop invariants, 
separation logic predicates, and symbolic execution. It captures some of the core aspects of the C programming language, but leaves out many complexities, including advanced concepts such as function pointers and concurrency, even though these are supported by VeriFast \cite{DBLP:conf/popl/JacobsP11,DBLP:conf/fm/JacobsSP11}. The running example (introduced on p.~\pageref{fig:example-program})
builds a large linked list in one routine and tears it down in another 
one; another example that appears (on p.~\pageref{listing:reverse}) is the in-place reversal of a linked list.
We use Featherweight VeriFast to verify the safety of both examples.

We hope that the definitions in this article are clear and the proofs are convincing;
however, to address any shortcomings in this regard,
we developed a machine-readable
executable definition and machine-checked soundness proof
of a slight variant of Featherweight VeriFast, called
Mechanised Featherweight VeriFast, in the Coq proof assistant. It is available at \url{http://www.cs.kuleuven.be/~bartj/fvf/}.
Furthermore, the executable nature of the definitions allowed us to test for errors in our programming language semantics and to verify that the formalized verification algorithm succeeds in verifying the example programs.

The structure of the article is as follows. In 
Section~\ref{sec:language}, we define the syntax of the 
input programming language, 
and we illustrate it with a few example programs. In 
Section~\ref{sec:exec}, we illustrate and define the 
\emph{concrete execution} of programs in this programming 
language. In Sections~\ref{sec:scexec} and \ref{sec:symexec}, 
we gradually introduce the VeriFast verification approach: in 
Section~\ref{sec:symexec}, we present the approach itself, 
called \emph{symbolic execution}\footnote{We use this term in this article to denote the specific algorithm implemented by VeriFast. It is an instance of the general approach known in the literature as symbolic execution.}; in Section~\ref{sec:scexec}, 
we present an intermediate type of execution, called \emph{semiconcrete execution}, that sits between 
concrete execution and symbolic execution, which introduces 
some but not all features of the VeriFast approach.
In Section~\ref{sec:mech}, we discuss Mechanised Featherweight VeriFast. We end the article
with an overview of related work in Section~\ref{sec:related-work} and
a conclusion in Section~\ref{sec:conclusion}.

This article is based on a slide deck and lecture notes for a graduate 
course on program verification, and aims to be usable as an introduction 
to program verification. In the course, theory lectures based on this 
material are interleaved with hands-on lab sessions based on the VeriFast 
Tutorial \cite{vftutorial}. Acknowledgements of related work are deferred 
to Section~\ref{sec:related-work}.

\section{The Programming Language}\label{sec:language}

In this section, we define the syntax of programs and then show an example program.

\subsection{Syntax of Programs}

The programming language is as follows. \label{cmd-syntax} An 
integer expression $e$ is either an integer literal $z$, a 
variable $x$, an addition $e + e$, or a subtraction $e - e$. A 
boolean expression $b$ is either an equality comparison $e = 
e$, a less-than comparison $e < e$, or a negation $\lnot b$ of 
another boolean expression. A command $c$ is either an 
assignment $x := e$ of an integer expression $e$ to a variable 
$x$, a sequential composition $(c; c)$ of two commands (whose 
execution proceeds by first executing the first command and 
then the second command), a conditional command $\mathbf{if}\ 
b\ \mathbf{then}\ c\ \mathbf{else}\ c$, a while loop 
$\mathbf{while}\ b\ \mathbf{do}\ c$, a routine call 
$r(\overline{e})$ (which calls routine $r$ with argument list 
$\overline{e}$ (a line over a letter means a list of the things 
denoted by the letter))\footnote{In this simple language, routines have no return value. A routine can pass a result to its caller by taking an address where the result should be stored as an argument.}, a heap memory block allocation $x := 
\mathbf{malloc}(n)$ (which allocates a block of heap memory of size $n$ and stores the 
address of the new block in variable $x$), a memory 
read $x := [e]$ (which reads the value of the memory cell whose 
address is given by $e$ and stores it in variable $x$), a 
memory write $[e] := e$ (which writes the value of the second 
expression into the memory cell whose address is given by the 
first expression), or a deallocation command $\mathbf{free}(e)$ 
which releases the memory block allocated by $\mathbf{malloc}$ 
whose address is given by $e$. A routine definition 
$\mathit{rdef}$ is of the form $\mathbf{routine}\ 
r(\overline{x}) = c$ which declares $\overline{x}$ as the 
parameter list and $c$ as the body of routine $r$. 

\begin{defi}{Syntax of Programs}\label{defi:cmd-syntax}

$$\begin{array}{r l}
& z \in \mathbb{Z}, n \in \mathbb{N}\\
& x \in \mathit{Vars}\\
e ::= & z\ |\ x\ |\ e + e\ |\ e - e\\
b ::= & e = e\ |\ e < e\ |\ \lnot b\\
c ::= & x := e\ |\ (c; c)\ |\ \mathbf{if}\ b\ \mathbf{then}\ c\ \mathbf{else}\ c\ |\ \mathbf{while}\ b\ \mathbf{do}\ c\\
& |\ r(\overline{e})\ |\ x := \mathbf{malloc}(n)\ |\ x := [e]\ |\ [e] := e\ |\ \mathbf{free}(e)\\
\mathit{rdef} ::= & \mathbf{routine}\ r(\overline{x}) = c
\end{array}$$

\end{defi}

\subsection{Example Program}

\begin{figure}
$$
\begin{array}{c}
\begin{array}{c c}
\begin{array}{l}
\mathbf{routine}\ \mathsf{range}(\mathsf{i}, \mathsf{n}, \mathsf{result})\;=\\
\quad \mathbf{if}\ \mathsf{i} = \mathsf{n}\ \mathbf{then}\\
\quad\quad [\mathsf{result}] := 0\\
\quad \mathbf{else}\ (\\
\quad\quad \mathsf{head} := \mathbf{malloc}(2);\\
\quad\quad [\mathsf{result}] := \mathsf{head};\\
\quad\quad [\mathsf{head}] := \mathsf{i};\\
\quad\quad \mathsf{range}(\mathsf{i} + 1, \mathsf{n}, \mathsf{head} + 1)\\
\quad )
\end{array}
&
\begin{array}{l}
\mathbf{routine}\ \mathsf{dispose}(\mathsf{list})\;=\\
\quad \mathbf{if}\ \mathsf{list} = 0\ \mathbf{then}\\
\quad\quad \mathsf{dummy} := \mathsf{dummy}\\
\quad \mathbf{else}\ (\\
\quad\quad \mathsf{tail} := [\mathsf{list} + 1];\\
\quad\quad \mathbf{free}(\mathsf{list});\\
\quad\quad \mathsf{dispose}(\mathsf{tail})\\
\quad )
\end{array}
\end{array}\\
\\
\begin{array}{l}
\mathsf{cell} := \mathbf{malloc}(1); \mathsf{range}(0, 100000000, \mathsf{cell});\\
\mathsf{list} := [\mathsf{cell}]; \mathbf{free}(\mathsf{cell}); \mathsf{dispose}(\mathsf{list})
\end{array}
\end{array}$$
\caption{Example Program}\label{fig:example-program}
\end{figure}

The example program of Figure~\ref{fig:example-program} consists of routine definitions for 
routines $\mathsf{range}$ and $\mathsf{dispose}$ and a main 
command. Routine $\mathsf{range}$ has parameters $\mathsf{i}$, 
$\mathsf{n}$, and $\mathsf{result}$; it builds a linked list 
that stores the integers from $\mathsf{i}$, inclusive, to 
$\mathsf{n}$, exclusive, and writes the address of the new 
linked list into the memory cell whose address is given by 
$\mathsf{result}$. If $\mathsf{i}$ equals $\mathsf{n}$, the 
value 0 is written to address $\mathsf{result}$, denoting the 
empty linked list. Otherwise, a new linked list node is 
allocated with two fields; the first field holds the value of 
the node, and the second field holds the address of the next 
node. A recursive call of routine $\mathsf{range}$ is used to 
build the remaining nodes of the linked list. 

Routine $\mathsf{dispose}$ has the single parameter 
$\mathsf{list}$. It frees the nodes of the linked list pointed 
to by $\mathsf{list}$. If $\mathsf{list}$ is 0, this means the 
linked list is empty and nothing needs to be done. (Since in this programming language, each $\mathbf{if}$ command must specify a command for the $\mathbf{then}$ branch and for the $\mathbf{else}$ branch, we specify the command $\mathsf{dummy} := \mathsf{dummy}$ for the $\mathbf{then}$ branch, which has no effect.) Otherwise, 
the first node is freed and then a recursive call of 
$\mathsf{dispose}$ is used to free the remaining nodes.

The main program calls $\mathsf{range}$ to build a linked list 
holding the numbers 0 through 99999999. Before doing so, 
however, it allocates a memory cell to hold the address of the 
new list. After the $\mathsf{range}$ call, the address of the 
list is read from the cell, the cell is freed, and finally the 
list nodes are freed using a call of routine 
$\mathsf{dispose}$.

The purpose of Featherweight VeriFast is to verify that 
programs, like this one, never \emph{fail} (i.e.~access 
unallocated memory), i.e.~that no execution of the program 
fails. The example program has an infinite number of 
executions: for each possible address of each linked list node, 
there is a separate execution. In one execution, the first node 
is allocated at address 1000, the second node at address 2000, 
etc. In another execution, the first node is allocated at 
address 123, the second node at address 234, etc. Featherweight 
VeriFast must check that none of these infinitely many 
executions fail. 

Note: in a language like Java, the precise address at which an 
object is allocated cannot influence program execution, since 
the program can only compare two object references for 
equality; it cannot compare an object reference with an 
integer, check if one reference is less than another one, use literal addresses as object references, etc. 
However, in C, as well as in the programming language which we 
defined above, this is possible, so it is possible to write 
programs that fail or not depending on the address picked by 
$\mathbf{malloc}$. Here is such a program:
$$x := \mathbf{malloc}(1); [42] := 0$$
If, in a given execution of this program, the address picked by $\mathbf{malloc}$ happens to be 42,
the execution completes normally; otherwise, it fails.

Note also: While this aspect of memory allocation is peculiar 
to C, the fact that the language contains 
\emph{nondeterministic} constructs, i.e.~constructs whose 
observable behavior is not uniquely determined by the language 
specification, is universal to all programming languages: any 
language construct that accepts user input or otherwise 
interacts with the environment is nondeterministic from a 
verification point of view, since it leads to multiple possible 
executions, all of which need to be checked.

The main point illustrated by the example is that a program may 
have infinitely many executions, each of which may be very long 
(or even infinitely long), and all of these need to be checked 
for failure. This is true in all programming languages. 
Clearly, it is inefficient or impossible to naively check each 
execution separately. VeriFast (and Featherweight VeriFast) 
perform \emph{modular symbolic execution} to achieve 
efficiency. After we define concrete execution precisely in 
Section~\ref{sec:exec}, we introduce the Featherweight VeriFast 
constructs for modularity in Section~\ref{sec:scexec} and the 
symbolic execution in Section~\ref{sec:symexec}. 

\section{Concrete Execution}\label{sec:exec}

In this section, we provide a formal definition of the behavior of programs of our programming language. We first introduce the notion of \emph{concrete execution states} by means of two examples of concrete execution \emph{traces} (sequences of states reached during an execution). We then introduce the notion of \emph{outcomes}, which we use to express failure, nontermination, and nondeterminism. Finally, we use these concepts to define concrete execution of commands and safety of a program, and we discuss the verification problem.

\subsection{Small Example Concrete Execution Trace}

\begin{figure}
$$\begin{array}{l}
\comment{\textrm{// $s = \mathbf{0}, h = \mathbf{0}$}}\\
\mathsf{pair} := \mathbf{malloc}(2);\\
\comment{\textrm{// $s = \mathbf{0}[\mathsf{pair} := \branching{100}], h = \llbrace\mathsf{mb}(\branching{100}, 2), \branching{100} \mapsto \branching{42}, \branching{101} \mapsto \branching{24}\rrbrace$}}\\{}
[\mathsf{pair}] := 0;\\
\comment{\textrm{// $s = \mathbf{0}[\mathsf{pair} := 100], h = \llbrace\mathsf{mb}(100, 2), 100 \mapsto 0, 101 \mapsto 24\rrbrace$}}\\
\mathbf{free}(\mathsf{pair})\\
\comment{\textrm{// $s = \mathbf{0}[\mathsf{pair} := 100], h = \mathbf{0}$}}\\
\\
\textrm{where}\\
\quad f[x := y] = \textrm{function update} = \lambda z.\;\left\{\begin{array}{l l}
y & \textrm{if $z = x$}\\
f(z) & \textrm{otherwise}
\end{array}\right.\\
\quad \mathbf{0} = \lambda x.\;0 = \textrm{empty store} = \textrm{empty heap} = \llbrace \rrbrace = \textrm{empty multiset}\\
\quad \llbrace e_1, \dots, e_n\rrbrace = \mathbf{0} + \llbrace e_1\rrbrace + \cdots + \llbrace e_n\rrbrace\\
\quad M + \llbrace e \rrbrace = M[e := M(e) + 1]
\end{array}$$
\caption{Example concrete execution trace}\label{fig:ex-ctrace}
\end{figure}

The small example program in Figure~\ref{fig:ex-ctrace} allocates a memory block of 
size 2, initializes the first element of the block to 0, and 
then frees the block. An example \emph{execution trace} of this 
program is shown in comments before and after the code lines. 
An execution trace is a sequence of \emph{execution states} (or 
\emph{states} for short). In our simple programming language, a 
state consists of a \emph{store} $s$ and a \emph{heap} $h$. A 
store is a function that maps variables to their current 
values; a heap is a \emph{multiset} (or \emph{bag}) of 
\emph{heap chunks}. A multiset is like a set, except that it 
may contain elements more than once. Mathematically, it is a 
function that maps each potential element to the number of times it 
occurs in the multiset. A heap chunk (in concrete executions) 
is either a \emph{points-to chunk} $\ell \mapsto v$ denoting 
that there is an allocated memory cell at address $\ell$ whose 
current value is $v$, or a \emph{malloc block chunk} 
$\mathsf{mb}(\ell, n)$ denoting that a memory block of size $n$ 
was allocated at address $\ell$ by $\mathbf{malloc}$, i.e.~that 
the memory cells at addresses $\ell$ through $\ell + n - 1$ are 
part of a single block, which will be freed as one unit when 
$\mathbf{free}$ is called with argument $\ell$. 

The example programming language makes a few simplifications 
compared to real machine states: memory cells may store 
arbitrary integers, rather than just bytes, and memory 
addresses may be arbitrary positive integers, rather than being 
bounded by the size of installed memory (or the size of the 
address space).

The initial store maps all variables to zero; the initial heap 
contains no heap chunks, i.e.~it contains all heap chunks zero 
times, so it is a function that maps all heap chunks to zero. 
We denote a function that maps all arguments to zero by 
$\mathbf{0}$.

In the example execution trace, the $\mathbf{malloc}$ operation 
allocates the new block at address 100. Therefore, in the 
execution state after the $\mathbf{malloc}$ operation, the 
store maps the target variable $\mathsf{pair}$ of the 
$\mathbf{malloc}$ operation to 100, and the heap contains three 
heap chunks: the two points-to chunks that correspond to the 
two memory cells that constitute the newly allocated block, and 
the malloc block chunk that records that these two memory cells 
are part of the same block. As in C, the initial contents of 
the newly allocated memory cells are arbitrary; in the example 
trace, the contents are 42 and 24. (All numbers that were 
picked arbitrarily are shown in orange, to highlight that the 
program has infinitely many other executions, that pick these 
numbers differently.)

The notation $f[a:=b]$ denotes the function that is like $f$ 
except that it maps argument $a$ to value $b$. The notation 
$\llbrace e_1, e_2\rrbrace$ denotes the multiset with elements 
$e_1$ and $e_2$ (where possibly $e_1 = e_2$). Formally, 
$\llbrace e_1, \dots, e_n\rrbrace = \mathbf{0} + \llbrace 
e_1\rrbrace + \cdots + \llbrace e_n\rrbrace$, where $M + 
\llbrace e\rrbrace = M[e:=M(e) + 1]$; i.e.~the multiset $M + 
\llbrace e\rrbrace$ is like $M$ except that element $e$ occurs 
once more than in $M$.

The second command, which initializes the memory cell at 
address $\mathsf{pair}$ to zero, causes the state to change in 
just one place: the value of the points-to chunk with address 
100 changes from 42 to 0.

Finally, the $\mathbf{free}$ command removes the three heap 
chunks from the heap and leaves it empty; it does not modify 
any variables so the store remains unchanged. 

\subsection{Large Example Concrete Execution Trace}

\begin{figure}

\begin{changemargin}{-1cm}{-1cm}
$$\begin{array}{l}
\mathbf{routine}\ \mathsf{range}(\mathsf{i}, \mathsf{n}, \mathsf{r})\:=\\
\comment{s{:} \mathbf{0}[\mathsf{i}{:}5,\mathsf{n}{:}8,\mathsf{r}{:}41], h {:} h_0{\uplus}\llbrace 41 {\mapsto} 77\rrbrace}\\
\mathbf{if}\ \mathsf{i} = \mathsf{n}\ \mathbf{then}\ \mathsf{l} := 0\ \mathbf{else}\ (\\
\comment{s{:} \mathbf{0}[\mathsf{i}{:}5,\mathsf{n}{:}8,\mathsf{r}{:}41], h{:} h_0{\uplus}\llbrace 41 {\mapsto} 77\rrbrace}\\
\mathsf{l} := \mathbf{malloc}(2);\\
\comment{s{:} \mathbf{0}[\mathsf{i}{:}5,\mathsf{n}{:}8,\mathsf{r}{:}41,\mathsf{l}{:}\branching{50}], h{:} h_0{\uplus}\llbrace 41 {\mapsto}77{,} \mathsf{mb}(\branching{50}, 2){,} \branching{50}{\mapsto}\branching{88}{,} \branching{51} {\mapsto} \branching{99}\rrbrace}\\{}
[\mathsf{l}] := \mathsf{i}; \mathsf{range}(\mathsf{i} + 1, \mathsf{n}, \mathsf{l} + 1)\\
\quad \vdots \quad {\color{Purple}\textit{(Execution of 3 nested $\mathsf{range}$ calls)}}\\
\comment{\begin{array}{@{} l @{}}s{:}\mathbf{0}[\mathsf{i}{:}5,\mathsf{n}{:}8,\mathsf{r}{:}41,\mathsf{l}{:}50], h{:} h_0{\uplus}\llbrace41 {\mapsto} 77{,} \mathsf{mb}(50, 2){,}50{\mapsto}5{,} 51 {\mapsto} \branching{60}{,}\\
\quad \mathsf{mb}(\branching{60}, 2){,} \branching{60}\mapsto 6{,} \branching{61}\mapsto \branching{70}{,}\mathsf{mb}(\branching{70},2){,}\branching{70}\mapsto 7{,}\branching{71}\mapsto 0\rrbrace
\end{array}}\\
);\\
{}[\mathsf{r}] := \mathsf{l}\\
\comment{\begin{array}{@{} l @{}}
s{:}\mathbf{0}[\mathsf{i}{:}5,\mathsf{n}{:}8,\mathsf{r}{:}41,\mathsf{l}{:}50], h{:} h_0{\uplus}\llbrace 41\mapsto 50{,} \mathsf{mb}(50, 2){,} 50{\mapsto}5{,} 51 {\mapsto} 60{,}\\
\quad \mathsf{mb}(60, 2){,} 60{\mapsto} 6{,} 61{\mapsto} 70{,}\mathsf{mb}(70,2){,}70{\mapsto} 7{,}71{\mapsto} 0\rrbrace
\end{array}}\\
\textrm{where $h_0:\llbrace\mathsf{mb}(30){,}30{\mapsto} 3{,}31{\mapsto} 40{,}\mathsf{mb}(40,2){,}40{\mapsto} 4\rrbrace$}
\end{array}$$
\end{changemargin}

\caption{Larger Example Concrete Execution Trace}\label{fig:larger-cexec-trace}
\end{figure}

\begin{exa}{Larger Example Concrete Execution Trace}
See Figure~\ref{fig:larger-cexec-trace}.
\end{exa}

Now, let's look at an execution trace of routine 
$\mathsf{range}$ from the example program introduced earlier. 
This trace is part of a larger program execution trace. We look 
at a particular call of $\mathsf{range}$. As shown in the first 
state of the trace, the values of parameters $\mathsf{i}$, 
$\mathsf{n}$, and $\mathsf{r}$ are 5, 8, and 41. That is, the 
caller is asking $\mathsf{range}$ to build a linked list with 
three nodes holding the values 5, 6, and 7, respectively, and 
to store the address of the newly built linked list in the 
previously allocated memory cell at address 41. At the time of 
the call, the heap consists of some chunks $h_0$ plus a 
points-to chunk with address 41 and value 77. (We use both $M + 
M'$ and $M \uplus M'$ for multiset union, defined as $M + M' = 
M \uplus M' = \lambda e.\;M(e) + M'(e)$.)

The first statement is the $\mathbf{if}$ statement. It checks 
if $\mathsf{i} = \mathsf{n}$. Since this is not the case, we 
skip the $\mathbf{then}$ branch and execute the $\mathbf{else}$ 
branch. The state upon arrival in the $\mathbf{else}$ branch is 
unchanged.

The execution of the $\mathbf{malloc}$ block is as before, in 
the simple example. In this trace, the block is allocated at 
address 50, and the initial values of the memory cells are 88 
and 99.

Then, the first cell of the new block is initialized to 
$\mathsf{i}$. We do not show the resulting state; only the 
value of the points-to chunk with address 50 changes.

Then, we get the recursive call of $\mathsf{range}$ to build 
the rest of the linked list. We do not show the execution 
states reached during the execution of this recursive call 
(which itself contains two more calls, one nested within the 
other); we skip directly to the state reached upon return from 
the call.

At this point, two more linked list nodes have been allocated, 
at addresses 60 and 70 (in this trace). Also, the linked list 
is well-formed: the second cell (which serves as the 
$\mathsf{next}$ field) of the node at address 50 points to the 
node at address 60, the $\mathsf{next}$ field of the node at 
address 60 points to the node at address 70, and the 
$\mathsf{next}$ field of the node at address 70 is a null 
pointer, indicating the end of the linked list. 

The final command writes the address 50 of the newly built 
linked list to the address 41 provided by the caller; this 
modifies only the points-to chunk with address 41 in the heap.

The points to remember about this example trace are that it is 
long (since it contains three nested routine executions, which 
we did not show), that its states have large heaps with many 
chunks (here, up to 15 chunks, if we include $h_0$), and that 
routine $\mathsf{range}$ has infinitely many more execution 
traces like this one, that pick the numbers shown in orange 
differently. In subsequent sections, we will define alternative ways of executing programs
where programs have fewer and shorter executions, and execution states have fewer heap chunks.

\subsection{Concrete Execution States}

We define the set $\mathit{CStates}$ of concrete execution 
states. The concrete stores $\mathit{CStores}$ 
are the functions from variables to integers. The concrete 
predicates (i.e.~the concrete chunk names) are the points-to 
predicate $\mapsto$ and the malloc block predicate 
$\mathsf{mb}$. The set of concrete chunks is the set of 
expressions of the form $p(\ell, v)$, where $p$ is a concrete 
predicate, and $\ell$ and $v$ are integers. We call $p$ the 
\emph{name} of the chunk, and $\ell$ and $v$ the 
\emph{arguments} of the chunk. The concrete heaps are the 
multisets of concrete chunks. The concrete states are the pairs 
of concrete stores and concrete heaps.

We often use the alternative syntax $\ell \mapsto v$ for the 
points-to chunk ${\mapsto}(\ell, v)$. 

\begin{defi}{Concrete Execution States}

$$\begin{array}{r l}
\mathit{CStores} = & \mathit{Vars} \rightarrow \mathbb{Z}\\
\mathit{CPredicates} = & \{\mapsto, \mathsf{mb}\}\\
\mathit{CChunks} = & \{p(\ell, v)\ |\ p \in \mathit{CPredicates}, \ell, v \in \mathbb{Z}\}\\
\mathit{CHeaps} = & \mathit{CChunks} \rightarrow \mathbb{N}\\
\mathit{CStates} = & \mathit{CStores} \times \mathit{CHeaps}
\end{array}$$

\begin{center}
$\ell \mapsto v$ is alternative syntax for ${\mapsto}(\ell, v)$
\end{center}

\end{defi}

\subsection{Outcomes}

In this subsection, we introduce the notion of \emph{outcomes}, which we use to express failure, nontermination, and nondeterminism. We first introduce the various types of outcomes by example. We then provide formal definitions of outcomes, without or with \emph{answers}. Finally, we define the concepts of satisfaction of a postcondition by an outcome, coverage of an outcome by another outcome, and sequential composition of outcomes; we state some properties; and we introduce some notations.

\subsubsection{Outcomes by Example}

To define mathematically what the concrete executions of a 
given program are, we define the function $\mathsf{exec}$, 
which takes as arguments a command and an input state, and 
returns the \emph{outcome} of executing the command starting in 
the given input state. In simple cases, such as in the case of 
the assignment command $\mathsf{p} := 42$, the outcome is a 
single output state: executing this assignment in the state 
$(\mathbf{0}, \mathbf{0})$ with an empty store (i.e.~one that 
maps all variables to zero) and an empty heap results in the 
single output state where the store maps $\mathsf{p}$ to value 
42 and all other variables to zero, and where the heap is still 
empty. We call such an outcome a \emph{singleton outcome}, and 
we denote the singleton outcome with output state $\sigma$ 
using angle brackets: $\langle \sigma\rangle$.

\begin{exa}{Concrete Execution: Singleton Outcomes}
$$\mathsf{exec}(\mathsf{p} := 42)((\mathbf{0}, \mathbf{0})) = \langle(\mathbf{0}[\mathsf{p}:=42], \mathbf{0})\rangle$$
\end{exa}

Note: we define and use function $\mathsf{exec}$ in a 
\emph{curried} form: instead of defining it as a function of 
two parameters (a command and an input state), we define it as 
a function of one parameter (a command) that returns another 
function of one parameter (an input state) which itself returns 
an outcome. We call the latter kind of function (a function 
that takes an input state and returns an outcome) a 
\emph{mutator}. Therefore, $\mathsf{exec}$ is a function that 
maps commands to mutators.

\begin{exa}{Concrete Execution: Demonic Choice}

$$\begin{array}{l}
\mathsf{exec}(\mathsf{p} := \mathbf{malloc}(0))((\mathbf{0}, \mathbf{0}))\;=\\
\quad \begin{array}{l l}
& \langle(\mathbf{0}[\mathsf{p} := 1], \llbrace \mathsf{mb}(1, 0)\rrbrace)\rangle\\
\otimes & \langle(\mathbf{0}[\mathsf{p} := 2], \llbrace \mathsf{mb}(2, 0)\rrbrace)\rangle\\
\otimes & \langle(\mathbf{0}[\mathsf{p} := 3], \llbrace \mathsf{mb}(3, 0)\rrbrace)\rangle\\
\otimes & \cdots
\end{array}
\end{array}$$

$$\begin{array}{l}
\mathsf{exec}(\mathsf{p} := \mathbf{malloc}(1))((\mathbf{0}, \mathbf{0}))\;=\\
\quad \langle(\mathbf{0}[\mathsf{p} := 1], \llbrace \mathsf{mb}(1, 1), 1 \mapsto 0\rrbrace)\rangle\\
\quad\quad \otimes\;\langle(\mathbf{0}[\mathsf{p} := 1], \llbrace \mathsf{mb}(1, 1), 1 \mapsto 1\rrbrace)\rangle \otimes \cdots\\
\quad \otimes\;\langle(\mathbf{0}[\mathsf{p} := 2], \llbrace \mathsf{mb}(2, 1), 2 \mapsto 0\rrbrace)\rangle\\
\quad\quad \otimes\;\langle(\mathbf{0}[\mathsf{p} := 2], \llbrace \mathsf{mb}(2, 1), 2 \mapsto 1\rrbrace)\rangle \otimes \cdots\\
\quad \otimes\;\langle(\mathbf{0}[\mathsf{p} := 3], \llbrace \mathsf{mb}(3, 1), 3 \mapsto 0\rrbrace)\rangle\\
\quad\quad \otimes\;\langle(\mathbf{0}[\mathsf{p} := 3], \llbrace \mathsf{mb}(3, 1), 3 \mapsto 1\rrbrace)\rangle \otimes \cdots\\
\quad \otimes \cdots
\end{array}$$

\end{exa}

The outcome of executing a command is not always a single 
state. Specifically, consider $\mathbf{malloc}$ commands: the 
command $\mathsf{p} := \mathbf{malloc}(0)$ allocates a new 
memory block of size zero. This means that it does not allocate 
any memory cells, but it does create an $\mathsf{mb}$ chunk at 
an address that is different from the address of existing 
$\mathsf{mb}$ chunks. When starting from an empty heap, this 
address may be any positive integer. For every distinct address 
chosen, there is a different output state. Notice that this 
choice can be considered \emph{demonic}: the program should not 
fail even if an attacker who tries to make the program fail 
makes this choice.  Therefore, the outcome returned by 
$\mathsf{exec}$ is a \emph{demonic choice} over the integers, 
where the chosen number is used as the address of the new block 
in the output state of a singleton outcome. That is, the 
\emph{operands} of the demonic choice outcome in this example 
are singleton outcomes. The demonic choice between outcomes 
$o_1$ and $o_2$ is denoted as $o_1 \otimes o_2$. 

In the case of the command $\mathsf{p} := \mathbf{malloc}(1)$, 
the outcome is a demonic choice over both the address of the 
new block and the initial value of the new memory cell. 

\begin{exa}{Concrete Execution: Failure, Nontermination, Angelic Choice}

$$\begin{array}{l}
\mathsf{exec}([0] := 33)((\mathbf{0}, \mathbf{0})) = \bot\\
\\
\mathsf{exec}(\mathsf{recurse}())((\mathbf{0}, \mathbf{0})) = \top\\
\quad \textrm{where $\mathbf{routine}\ \mathsf{recurse}() = \mathsf{recurse}()$}\\
\\
\mathsf{exec}(\mathbf{backtrack}(c_1, c_2))(\sigma) = \mathsf{exec}(c_1)(\sigma) \oplus \mathsf{exec}(c_2)(\sigma)
\end{array}$$

\end{exa}

Singleton outcomes and demonic choices are not the only kinds 
of outcomes; there are three more kinds.

Consider the command $[0] := 33$ when executed in the empty 
heap. This is an access of an unallocated memory cell, i.e.~it 
is a failure. We denote the failure outcome by the symbol for 
``bottom'': $\bot$. 

Consider the routine call $\mathsf{recurse}()$ and assume that 
routine $\mathsf{recurse}$ is defined such that its body is 
simply a recursive call of itself. This command performs an 
infinite recursion; it does not terminate.\footnote{In the case 
of a real programming language, this would lead to a stack 
overflow error at run time, except if the compiler performs 
tail recursion optimization. Neither VeriFast nor Featherweight 
VeriFast verify the absence of stack overflows, so we ignore 
this issue in this formalization.} Nontermination is often 
considered undesirable; however, in many other cases, it is 
intentional: for example, a web server or a database server is 
not supposed to terminate unless and until the user instructs 
it to do so. In any case, VeriFast and Featherweight VeriFast 
do not verify termination. Featherweight VeriFast verifies only 
the absence of accesses of unallocated memory, so from this 
point of view a nonterminating command is a good thing, since 
it prevents the remainder of the program from executing, 
including any commands that might fail. Therefore, we represent 
nontermination with the symbol for ``top'', $\top$, the 
opposite of $\bot$.

Finally, to round out the ``algebra of outcomes'', we 
introduce also \emph{angelic choice}. True angelic choice does 
not occur in concrete executions of our programming 
language\footnote{A degenerate form of angelic choice does 
occur: failure is equivalent to angelic choice over zero 
alternatives, as we will see later.}; however, some real 
programming languages do have a form of angelic choice. For 
example, the logic programming language Prolog allows the user 
to specify multiple alternative ways to solve a problem. At run 
time, Prolog will try first the first alternative; if it fails, 
it restores the program state and then tries the second 
alternative. The program as a whole succeeds if either 
alternative succeeds: it is as if an angel chooses the right 
alternative. Another example is a transactional database: if a 
schedule fails, the state is rolled back and another schedule 
is attempted.

We introduce angelic choice here to obtain a nice, complete 
algebra, but also because we will use angelic choice in our 
definition of the Featherweight VeriFast verification 
algorithm.

In summary, (concrete) mutators are functions from (concrete) 
input states to outcomes over (concrete) output states. (Later 
we will also use outcomes over other state spaces.) The 
concrete execution function $\mathsf{exec}$ maps commands to 
concrete mutators. 

\begin{defi}{Type of Concrete Execution}

$$\begin{array}{r c l}
\mathit{CMutators} & = & \mathit{CStates} \rightarrow \mathit{Outcomes}(\mathit{CStates})\\
\mathsf{exec} & \in & \mathit{Commands} \rightarrow \mathit{CMutators}
\end{array}$$

\end{defi}

\subsubsection{Outcomes: Definition}

An outcome $\phi$ over a state space 
$\mathcal{S}$ is either a singleton outcome $\langle \sigma\rangle$, 
with $\sigma \in \mathcal{S}$, or a demonic choice $\bigotimes \Phi$
over the outcomes in $\Phi$,
or an angelic choice $\bigoplus \Phi$
over the outcomes in $\Phi$,
where $\Phi$ is a set of outcomes over $\mathcal{S}$.
We denote the set of outcomes over state 
space $\mathcal{S}$ as $\mathit{Outcomes}(\mathcal{S})$.
$\bigotimes \Phi$ and $\bigoplus \Phi$ are called \emph{infinitary} demonic and angelic choice, since
the set $\Phi$ is potentially infinite.

Binary demonic choice, binary angelic choice, nontermination, 
and failure can be defined as special cases of infinitary demonic 
choice and infinitary angelic choice: binary choices are
choices over the set collecting the two alternatives; nontermination is a demonic 
choice over zero alternatives (the attacker is stuck with no 
alternatives, which is a good thing); failure is an angelic 
choice over zero alternatives (the angel is stuck with no 
alternatives, which is a bad thing). 

\begin{defi}{Outcomes}

$$\begin{array}{r c l l l}
\phi & ::= & & \langle \sigma\rangle & \textrm{singleton outcome}\\
& & | & \bigotimes \Phi & \textrm{demonic choice}\\
& & | & \bigoplus \Phi & \textrm{angelic choice}
\end{array}$$

$$\begin{array}{r c l}
\sigma \in \mathcal{S} & \Rightarrow & \langle \sigma \rangle \in \mathit{Outcomes}(\mathcal{S})\\
\Phi \subseteq \mathit{Outcomes}(\mathcal{S}) & \Rightarrow & \bigotimes \Phi \in \mathit{Outcomes}(\mathcal{S})\\
\Phi \subseteq \mathit{Outcomes}(\mathcal{S}) & \Rightarrow & \bigoplus \Phi \in \mathit{Outcomes}(\mathcal{S})
\end{array}$$

$$\begin{array}{r c l l}
\phi_1 \otimes \phi_2 & = & \bigotimes \{\phi_1, \phi_2\} & \textrm{binary demonic choice}\\
\phi_1 \oplus \phi_2 & = & \bigoplus \{\phi_1, \phi_2\} & \textrm{binary angelic choice}\\
\top & = & \bigotimes \emptyset & \textrm{nontermination}\\
\bot & = & \bigoplus \emptyset & \textrm{failure}
\end{array}$$

\end{defi}

\subsubsection{Outcomes with Answers}

Often, it is useful to consider mutators that have not just an output state but also an \emph{answer}.
We denote a singleton outcome with output state $\sigma \in \mathcal{S}$ and answer $a \in \mathcal{A}$ by $\langle \sigma, a\rangle$.
For uniformity, we treat outcomes without answers like outcomes whose answer is the \emph{unit value} $\mathsf{tt}$, the sole element of the
\emph{unit set} $\mathsf{unit}$. That is, we consider $\mathit{Outcomes}(\mathcal{S})$ a shorthand for $\mathit{Outcomes}(\mathcal{S}, \mathsf{unit})$,
and $\langle \sigma\rangle$ a shorthand for $\langle \sigma, \mathsf{tt}\rangle$.

\begin{defi}{Outcomes with Answers}

$$\begin{array}{r c l l l}
\phi & ::= & & \langle \sigma, a\rangle & \textrm{singleton outcome}\\
& & | & \bigotimes \Phi & \textrm{demonic choice}\\
& & | & \bigoplus \Phi & \textrm{angelic choice}
\end{array}$$

$$\begin{array}{r c l}
\sigma \in \mathcal{S} \land a \in \mathcal{A} & \Rightarrow & \langle \sigma, a\rangle \in \mathit{Outcomes}(\mathcal{S}, \mathcal{A})\\
\Phi \subseteq \mathit{Outcomes}(\mathcal{S}, \mathcal{A}) & \Rightarrow & \bigotimes \Phi \in \mathit{Outcomes}(\mathcal{S}, \mathcal{A})\\
\Phi \subseteq \mathit{Outcomes}(\mathcal{S}, \mathcal{A}) & \Rightarrow & \bigoplus \Phi \in \mathit{Outcomes}(\mathcal{S}, \mathcal{A})
\end{array}$$

$$\langle \sigma\rangle = \langle \sigma, \mathsf{tt}\rangle \in \mathit{Outcomes}(\mathcal{S}) = \mathit{Outcomes}(\mathcal{S}, \mathsf{unit})$$

\end{defi}

\subsubsection{Outcomes: Satisfaction, Coverage}

A useful question to ask is whether an outcome $\phi \in 
\mathit{Outcomes}(\mathcal{S}, \mathcal{A})$ satisfies a given postcondition 
$Q$, where a postcondition can be modelled mathematically as 
the set of pairs of output states and answers that satisfy it, i.e.~$Q \subseteq 
\mathcal{S} \times \mathcal{A}$. We denote this by $\phi\;\{Q\}$.

We define this recursively as follows:
\begin{itemize}
\item A singleton outcome $\langle \sigma, a\rangle$ satisfies 
    postcondition $Q$ if the output state $\sigma$ and answer $a$ satisfy
    $Q$, i.e.~$(\sigma, a) \in Q$.
\item A demonic choice $\bigotimes \Phi$ 
    satisfies $Q$ if all alternatives satisfy $Q$
\item An angelic choice $\bigoplus \Phi$ 
    satisfies $Q$ if some alternative satisfies $Q$.
\end{itemize}

Notice that it follows from this definition that nontermination 
satisfies all postconditions (even the postcondition that does 
not accept any output state), and failure satisfies no 
postcondition (not even the postcondition that accepts all output 
states).

We also define \emph{coverage} between outcomes: we say outcome 
$\phi$ \emph{covers} outcome $\phi'$, denoted $\phi \Rrightarrow 
\phi'$, if for any postcondition $Q$, if $\phi$ satisfies $Q$, 
then $\phi'$ satisfies $Q$. Intuitively, this means $\phi$ is a 
``worse'' outcome than $\phi'$; if $\phi'$ is failure, then 
$\phi$ must be failure, but the converse does not hold: it is 
possible that $\phi$ is failure but $\phi'$ is not. Another way
to look at this is to say that $\phi$ is a safe approximation of $\phi'$ for verification:
if we prove that $\phi$ satisfies some postcondition,
then it follows that $\phi'$ also satisfies it.

We lift outcome coverage pointwise to mutators: a mutator $C$ 
\emph{covers} a mutator $C'$ if for each input state $\sigma$, the outcome 
of $C$ started in state $\sigma$ covers the outcome of $C'$ started in 
state $\sigma$. 

\begin{defi}{Outcomes: Satisfaction, Coverage}

$$\phi \in \mathit{Outcomes}(\mathcal{S}, \mathcal{A})\quad Q \subseteq \mathcal{S} \times \mathcal{A}$$
$$\phi\;\{Q\}\quad \textrm{(``outcome $\phi$ satisfies postcondition $Q$'')}$$

$$\begin{array}{r c l}
\langle \sigma, a \rangle\; \{Q\} & \Leftrightarrow & (\sigma, a) \in Q\\
\bigotimes \Phi\; \{Q\} & \Leftrightarrow & \forall \phi \in \Phi.\; \phi\; \{Q\}\\
\bigoplus \Phi\; \{Q\} & \Leftrightarrow & \exists \phi \in \Phi.\; \phi\; \{Q\}\\
\\
\phi \Rrightarrow \phi' & \Leftrightarrow & \forall Q.\; \phi\; \{Q\} \Rightarrow \phi'\; \{Q\}\\
C \Rrightarrow C' & \Leftrightarrow & \forall \sigma.\; C(\sigma) \Rrightarrow C'(\sigma)
\end{array}$$

\end{defi}

\subsubsection{Outcomes: Sequential Composition}

An important concept is the sequential composition $\phi; C$ of 
an outcome $\phi \in \mathit{Outcomes}(\mathcal{S})$ and a mutator $C \in \mathcal{S} \rightarrow 
\mathit{Outcomes}(\mathcal{S'})$. \label{mutator-seqcomp} The 
intuition is straightforward: the output states of $\phi$ are 
passed as input states to $C$. The result is again an outcome. 
It is defined as follows: 
\begin{itemize}
\item if $\phi$ is a singleton outcome $\langle \sigma\rangle$, 
    the sequential composition is the outcome of passing 
    $\sigma$ as input to $C$
\item if $\phi$ is a demonic or angelic choice, the 
    sequential composition is the distribution of the 
    sequential composition over the alternatives.
\end{itemize} 

We also define sequential composition $C; C'$ of two mutators 
$C$ and $C'$: it is simply the mutator that, for a given input 
state $\sigma$, passes $\sigma$ to $C$ and composes the outcome 
sequentially with $C'$. 

\begin{defi}{Outcomes: Sequential composition}
$$- ; -\ :\ \mathit{Outcomes}(\mathcal{S}) \rightarrow (\mathcal{S} \rightarrow \mathit{Outcomes}(\mathcal{S}')) \rightarrow \mathit{Outcomes}(\mathcal{S}')$$
$$\begin{array}{r c l}
\langle \sigma\rangle; C & = & C(\sigma)\\
(\bigotimes \Phi); C & = & \bigotimes \{\phi \in \Phi.\; (\phi; C)\}\\
(\bigoplus \Phi); C & = & \bigoplus \{\phi \in \Phi.\; (\phi; C)\}
\end{array}$$
$$C; C' = \lambda \sigma.\; C(\sigma); C'$$
\end{defi}

We have the following important properties of sequential 
composition:
\begin{itemize}
\item Associativity: given three mutators $C$, $C'$, and 
    $C''$, first composing $C$ and $C'$ and then composing the 
    resulting mutator with $C''$ is equivalent to first 
    composing $C'$ and $C''$ and then composing $C$ with the 
    resulting mutator.
\item Monotonicity: if mutator $C_1$ is worse than mutator 
    $C_1'$, and mutator $C_2$ is worse than mutator $C_2'$, 
    then $C_1; C_2$ is worse than $C_1'; C_2'$.
\item Satisfaction: the sequential composition $\phi; C$ 
    satisfies the postcondition $Q$ if and only if $\phi$ 
    satisfies the postcondition that accepts the state $\sigma$ 
    if $C(\sigma)$ satisfies $Q$.
\end{itemize}

\begin{lem}[Associativity of Sequential Composition of Mutators]
$$(C; C'); C'' = C; (C'; C'')$$
\end{lem}

\begin{lem}[Monotonicity of Sequential Composition of Mutators]
If $C_1 \Rrightarrow C_1'$ and $C_2 \Rrightarrow C_2'$ then $C_1; C_2\ \Rrightarrow C_1'; C_2'$.
\end{lem}

\begin{lem}[Satisfaction of Sequential Composition of Mutators]
$$\phi; C\; \{Q\} \Leftrightarrow \phi\; \{\sigma\;|\;C(\sigma)\;\{Q\}\}$$
\end{lem}

\subsubsection{Outcomes: Sequential Composition (with Answers)}

We can generalize these concepts to the case of outcomes with answers.
If $\phi$ is an outcome and $C({-})$ is a function from answers to mutators, i.e.~a mutator parameterized by an answer, then
we write the sequential composition of $\phi$ and $C({-})$ as $x \leftarrow \phi; C(x)$; that is,
the answer of $\phi$, bound to the variable $x$, is passed as an input argument to $C({-})$.
The definition and the properties are a straightforward adaptation of the ones given above for outcomes without answers.

\begin{defi}{Outcomes: Sequential composition (with Answers)}

$$x \leftarrow {-} ; {-}(x)\ :\ 
\mathit{O}(\mathcal{S}, \mathcal{A}) \rightarrow
(\mathcal{A} \rightarrow \mathcal{S} \rightarrow \mathit{O}(\mathcal{S}', \mathcal{B})) \rightarrow \mathit{O}(\mathcal{S}', \mathcal{B})$$
$$\begin{array}{r c l}
x \leftarrow \langle \sigma, a\rangle; C(x) & = & C(a)(\sigma)\\
x \leftarrow (\bigotimes \Phi); C(x) & = & \bigotimes \{\phi \in \Phi.\; (x \leftarrow \phi; C(x))\}\\
x \leftarrow (\bigoplus \Phi); C(x) & = & \bigoplus \{\phi \in \Phi.\; (x \leftarrow \phi; C(x))\}
\end{array}$$

$$x \leftarrow C; C'(x) = \lambda \sigma.\; x \leftarrow C(\sigma); C'(x)$$

\end{defi}

\begin{lem}[Associativity of Sequential Composition of Mutators with Answers]
$$y \leftarrow (x \leftarrow C; C'(x)); C''(y) = x \leftarrow C; (y \leftarrow C'(x); C''(y))$$
\end{lem}

\begin{lem}[Monotonicity of Sequential Composition of Mutators with Answers]
If $C_1 \Rrightarrow C_1'$ and $\forall a.\;C_2(a) \Rrightarrow C_2'(a)$ then $x \leftarrow C_1; C_2(x) \Rrightarrow x \leftarrow C_1'; C_2'(x)$.
\end{lem}

\begin{lem}[Satisfaction of Sequential Composition of Mutators with Answers]
$$x \leftarrow \phi; C(x)\; \{Q\} \Leftrightarrow \phi\; \{(\sigma, a)\;|\;C(a)(\sigma)\;\{Q\}\}$$
\end{lem}

\subsubsection{Outcomes: Notations}

We introduce some additional notations and concepts that will 
be useful in the definition of the executions.

We lift demonic and angelic choice to mutators: if $\tilde{C}$ is a set of mutators,
then $\bigotimes \tilde{C}$ is the demonic choice over these mutators. It is the
mutator that, for a given input state $\sigma$, demonically chooses between the outcomes obtained by
passing $\sigma$ to the elements of $\tilde{C}$. Angelic choice over mutators is defined analogously.

We use the ``variable binding'' notation $\bigotimes i \in I.\;\phi_i$ to denote the demonic
choice over the outcomes obtained by letting $i$ range over $I$ in $\phi_i$. We also use this
notation for angelic choice and for choices over mutators.

As an extension of the variable binding notation, we also allow boolean 
propositions to the left of the dot in demonic and angelic 
choices. If the proposition is true, this has no effect; 
otherwise, in the case of angelic choice, this means failure, 
and in the case of demonic choice, this means nontermination. 

We define the primitive mutator $\mathsf{yield}\ a$ as the mutator that does not modify the state
and answers $a$. We define $\mathsf{noop}$ as the mutator that does nothing; it merely answers the unit
element $\mathsf{tt}$.

We define \emph{side-effect-only} sequential composition $C{;}, C'$ of two mutators $C$ and $C'$ as the mutator
that first executes $C$, and then executes $C'$, and whose answer is the answer of $C$. The answer of $C'$ is ignored.

\begin{notation}{Outcomes: Notations}

$$\begin{array}{r @{\ } l r}
\bigotimes \tilde{C} = & \lambda \sigma.\;\bigotimes \{C(\sigma)\ |\ C \in \tilde{C}\}\\
\bigoplus \tilde{C} = & \lambda \sigma.\;\bigoplus \{C(\sigma)\ |\ C \in \tilde{C}\}\\
\end{array}$$

$$\begin{array}{r @{\ } l}
\bigotimes i \in I.\;\phi_i = & \bigotimes \{\phi_i\;|\;i \in I\}\\
\bigoplus i \in I.\;\phi_i = & \bigoplus \{\phi_i\;|\;i \in I\}\\
\end{array}$$

$$\begin{array}{c c}
\begin{array}{r @{\ } l}
\bigotimes \mathsf{true}.\;\phi = & \phi\\
\bigoplus \mathsf{true}.\;\phi = & \phi\\
\end{array}
&
\begin{array}{r @{\ } l}
\bigotimes \mathsf{false}.\;\phi = & \top\\
\bigoplus \mathsf{false}.\;\phi = & \bot\\
\end{array}
\end{array}$$

$$\begin{array}{r @{\ } l}
\mathsf{yield}\ a = & \lambda \sigma.\;\langle \sigma, a\rangle\\
\mathsf{noop} = & \mathsf{yield}\;\mathsf{tt}
\end{array}$$

$$C{;}, C' = x \leftarrow C; C'; \mathsf{yield}\;x$$

\end{notation}

\subsection{Some Auxiliary Definitions}

We introduce some further auxiliary notions that will be useful in the definition of concrete execution of commands.

The domain of a heap $h$ is the set of \emph{domain elements} of the form 
$p(\ell)$ where a value $v$ exists such that $p(\ell, v)$ 
occurs in $h$.

The mutator $\mathsf{assume}(b)$, \label{def:assume} where $b$ 
is a boolean expression, evaluates $b$ in the given input 
store; if $b$ evaluates to $\mathsf{true}$, the mutator does 
nothing; otherwise, it does not terminate.
We define evaluation $\llbracket b\rrbracket_s$
of a boolean expression $b$ or arithmetic expression $e$ under a store $s$ as follows:
$\llbracket e = e'\rrbracket_s = (\llbracket e\rrbracket_s = \llbracket e'\rrbracket_s)$,
$\llbracket e < e'\rrbracket_s = (\llbracket e\rrbracket_s < \llbracket e'\rrbracket_s)$,
$\llbracket \lnot b\rrbracket_s = \lnot \llbracket b\rrbracket_s$,
$\llbracket z\rrbracket_s = z$,
$\llbracket x\rrbracket_s = s(x)$,
$\llbracket e + e'\rrbracket_s = \llbracket e\rrbracket_s + \llbracket e'\rrbracket_s$,
and
$\llbracket e - e'\rrbracket_s = \llbracket e\rrbracket_s - \llbracket e'\rrbracket_s$.

The mutator $\mathsf{store}$ simply returns the current store.
The mutator $\mathsf{store} := s$ sets the current store to $s$.
The mutator $\mathsf{with}(s, C)$ executes the mutator $C$ under store $s$ and
then restores the original store. Its answer is the answer of $C$.
The mutator $\mathsf{eval}(e)$ answers the value of $e$ under the current store.
The mutator $x := v$ updates the store, assigning value $v$ to variable $x$.

\begin{defi}{Some Auxiliary Definitions}

$$\begin{array}{r @{\ } l}
\mathsf{dom}(h) = & \{p(\ell)\ |\ \exists v.\;p(\ell, v) \in h\}\\
\mathsf{assume}(b) = & \lambda (s, h). \bigotimes \llbracket b\rrbracket_s = \mathsf{true}.\;\langle(s, h)\rangle\\
\mathsf{store} = & \lambda (s, h).\;\langle (s, h), s\rangle\\
\mathsf{store} := s' = & \lambda (s, h).\;\langle (s', h)\rangle\\
\mathsf{with}(s', C) = & s \leftarrow \mathsf{store}; \mathsf{store} := s'; \mathsf{C}{;}, \mathsf{store} := s\\
\mathsf{eval}(e) = & \lambda (s, h).\;\langle (s, h), \llbracket e\rrbracket_s\rangle\\
x := v = & \lambda (s, h).\;\langle (s[x:=v], h)\rangle
\end{array}$$

\end{defi}

We denote mutator $C$ iterated $n$ times as $C^n$. $C$ iterated zero times does nothing; $C$ iterated $n + 1$ times is the sequential composition of $C$ and
$C$ iterated $n$ times. The demonic iteration $C^*$ of $C$ is $C$ iterated a demonically chosen number of times.

Concrete consumption of a multiset $h$ of chunks fails if the heap does not contain these chunks; otherwise, it removes them.
Concrete production of a multiset $h$ of chunks blocks if the heap already contains chunks with the same address, i.e.~if the addresses of the chunks in $h$ are not
all pairwise distinct from the addresses of the chunks that are already in the heap. Otherwise, it adds the chunks to the heap.
Concrete consumption and production of a single chunk $\alpha$ are defined in the obvious way.

\begin{defi}{Some Auxiliary Definitions}

$$\begin{array}{r @{\ } l}
C^0 = & \mathsf{noop}\\
C^{n + 1} = & C; C^n\\
C^* = & \bigotimes n \in \mathbb{N}.\;C^n\\
\mathsf{cconsume\_chunks}(h') = & \lambda (s, h).\;\bigoplus h' \le h.\;\langle (s, h - h')\rangle\\
\mathsf{cconsume\_chunk}(\alpha) = & \mathsf{cconsume\_chunks}(\llbrace\alpha\rrbrace)\\
\mathsf{cproduce\_chunks}(h') = & \lambda (s, h).\;\bigotimes \mathsf{dom}(h) \cap \mathsf{dom}(h') = \emptyset.\;\langle (s, h \uplus h')\rangle\\
\mathsf{cproduce\_chunk}(\alpha) = & \mathsf{cproduce\_chunks}(\llbrace\alpha\rrbrace)
\end{array}$$

\end{defi}

\subsection{Concrete Execution of Commands}

\begin{figure}
$$\begin{array}{l}
\mathsf{exec}_0(c) = \top\\
\\
\mathsf{exec}_{n+1}(x:=e) = v \leftarrow \mathsf{eval}(e); x := v\\
\\
\mathsf{exec}_{n+1}(c; c') = \mathsf{exec}_n(c); \mathsf{exec}_n(c')\\
\\
\mathsf{exec}_{n+1}(\mathbf{if}\ b\ \mathbf{then}\ c\ \mathbf{else}\ c') =\\
\quad \mathsf{assume}(b); \mathsf{exec}_n(c) \otimes \mathsf{assume}(\lnot b); \mathsf{exec}_n(c')\\
\\
\mathsf{exec}_{n+1}(\mathbf{while}\ b\ \mathbf{do}\ c) =\\
\quad (\mathsf{assume}(b); \mathsf{exec}_n(c))^*; \mathsf{assume}(\lnot b)\\
\\
\mathsf{exec}_{n+1}(r(\overline{e})) = \overline{v} \leftarrow \mathsf{eval}(\overline{e}); \mathsf{with}(\mathbf{0}[\overline{x}:=\overline{v}], \mathsf{exec}_n(c))\\
\quad\textrm{where $\mathbf{routine}\ r(\overline{x}) = c$}\\
\\
\mathsf{exec}_{n+1}(x := \mathbf{malloc}(n)) =\\
\quad\bigotimes \ell, v_1,\dots,v_n \in \mathbb{Z}.\\
\quad\quad \mathsf{cproduce\_chunks}(\llbrace\mathsf{mb}(\ell, n), \ell \mapsto v_1,\dots,\ell + n - 1\mapsto v_n\rrbrace); x := \ell\\
\\
\mathsf{exec}_{n+1}(x := [e]) = \ell \leftarrow \mathsf{eval}(e);\\
\quad \bigoplus v.\;\mathsf{cconsume\_chunk}(\ell \mapsto v); \mathsf{cproduce\_chunk}(\ell \mapsto v); x := v\\
\\
\mathsf{exec}_{n+1}([e] := e') = \ell \leftarrow \mathsf{eval}(e); v \leftarrow \mathsf{eval}(e');\\
\quad \bigoplus v_0.\;\mathsf{cconsume\_chunk}(\ell \mapsto v_0); \mathsf{cproduce\_chunk}(\ell \mapsto v)\\
\\
\mathsf{exec}_{n+1}(\mathbf{free}(e)) = \ell \leftarrow \mathsf{eval}(e);\\
\quad \bigoplus N \in \mathbb{N}, v_1, \dots, v_N \in \mathbb{Z}.\\
\quad\quad \mathsf{cconsume\_chunks}(\llbrace\mathsf{mb}(\ell, N), \ell \mapsto v_1, \dots, \ell + N - 1 \mapsto v_N\rrbrace)\\
\\
\mathsf{exec}(c) = \bigotimes n \in \mathbb{N}.\ \mathsf{exec}_n(c)
\end{array}$$

\caption{Concrete Execution of Commands}\label{fig:cexec}
\end{figure}

\begin{defi}{Concrete Execution of Commands}
See Figure~\ref{fig:cexec}.
\end{defi}

To define the concrete execution function $\mathsf{exec}$, we 
first define a helper function $\mathsf{exec}_n$, which is 
indexed by the maximum depth of the execution. If an execution 
exceeds the maximum depth, $\mathsf{exec}_n$ returns $\top$, 
i.e.~the execution does not terminate.

Therefore, for any command $c$, $\mathsf{exec}_0(c)$ returns 
the mutator $\top$ (which is the mutator that for any input 
state returns the outcome $\top$).

Execution of an assignment $x := e$ evaluates $e$ and binds variable $x$ to its value.

Execution of a sequential composition $c; c'$ is the sequential 
composition of the execution of $c$ and the execution of $c'$. 
(Notice that the two semicolons in this rule have different 
meanings: the former is part of the syntax of commands defined 
on Page~\pageref{cmd-syntax}; the latter is the function 
defined on Page~\pageref{mutator-seqcomp} that takes two 
mutators and returns a mutator.) 

Execution of an if-then-else command $\mathbf{if}\ b\ 
\mathbf{then}\ c\ \mathbf{else}\ c'$ demonically chooses 
between two branches: in the first branch, it is assumed that 
the condition $b$ evaluates to true, and then command $c$ is 
executed; in the second branch, it is assumed that $b$ 
evaluates to false, and then $c'$ is executed. Notice that this 
is equivalent to evaluating the condition and then, depending 
on whether it evaluates to true or false, executing $c$ or 
$c'$, respectively.

Execution of a loop $\mathbf{while}\ b\ \mathbf{do}\ c$ first executes the body some demonically chosen number of times, after assuming
that the loop condition holds, and then assumes that the condition does not hold.

Execution of a call $r(\overline{e})$ of routine $r$ with 
argument list $\overline{e}$ first evaluates $\overline{e}$ to obtain values $\overline{v}$
and then executes the body $c$ of $r$ in a 
store which binds the parameters $\overline{x}$ of $r$ to $\overline{v}$.

Execution of a memory block allocation command $x := 
\mathbf{malloc}(n)$  demonically picks an address $\ell$ and values 
$v_1,\dots,v_n$ and produces the malloc block chunk and the $n$
points-to chunks that constitute the newly allocated memory block.
Finally, the execution binds variable $x$ to the address $\ell$.

Execution of a memory read command $x := [e]$ angelically
picks a value $v$ and tries to consume a points-to chunk at the
address given by $e$ and with value $v$. If it succeeds, it puts the chunk back and binds $x$ to $v$.

Execution of a memory write command $[e] := e'$ angelically
picks an old value $v_0$ and tries to consume a points-to chunk at the
address given by $e$ and with value $v_0$. If it succeeds, it puts the chunk back with an updated value.

Execution of a memory block deallocation command 
$\mathbf{free}(e)$ first evaluates expression $e$ to an address $\ell$
and then tries to consume a malloc block chunk and a corresponding number of points-to chunks at address $\ell$.

Execution of a command demonically chooses a maximum depth and 
then executes the command up to that depth. Notice that this is 
equivalent to executing the command without a depth bound. 

\subsection{Safety of a Program}

We say that a program is \emph{safe} if no execution of the 
program accesses unallocated memory, i.e.~no execution fails, 
when started from the empty state $\sigma_0$. (Notice that the 
failure outcome is the only outcome that does not satisfy 
postcondition ``true''.)

The verification problem addressed by Featherweight VeriFast is 
to check whether a command $c$ is a safe program.

\begin{defi}{Safety of a Program}
$$\begin{array}{r c l}
\sigma_0 & = & (\mathbf{0}, \mathbf{0})\\
a \triangleright f & = & f(a)\\
\mathsf{safe\_program}(c) & = & \sigma_0 \triangleright \mathsf{exec}(c)\;\{\mathrm{true}\}
\end{array}$$
\end{defi}

\begin{defi}[The Verification Problem]
$$\mathsf{safe\_program}(c)$$
\end{defi}

\subsection{Solving the Verification Problem}

\begin{figure}

\begin{changemargin}{-1cm}{-1cm}
\begin{center}
\begin{tabular}{| r | c | c | c |}
\hline
& $\mathsf{exec}$ & $\mathsf{scexec}$ & $\mathsf{symexec}$\\
\hline
Recursion & Yes & No & No\\
Looping & Yes & No & No\\
Branching & Infinite & Infinite & Finite\\
\hline
Is Algorithm & No & No & Yes\\
\hline
\end{tabular} 
\end{center}

$$\mathsf{exec}\quad\xrightarrow{\begin{array}{c}
\textrm{Assertions}\\
\textrm{Predicates}\\
\textrm{Routine contracts}\\
\textrm{Loop invariants}
\end{array}}\quad
\mathsf{scexec}\quad\xrightarrow{\begin{array}{c}
\textrm{Symbols}\\
\textrm{Path Condition}\\
\textrm{Fresh Symbols}\\
\textrm{Theorem Prover}
\end{array}}\quad
\mathsf{symexec}$$
\end{changemargin}

\caption{Solving the Verification Problem}\label{fig:solving}
\end{figure}

See Figure~\ref{fig:solving}.

How to solve the verification problem? Naively computing the 
full traces of all executions of a program is impossible, since 
traces may be very large or even infinite (due to recursion and 
loops), and there may be infinitely many executions (due to the 
nondeterminism of memory allocation, causing execution to split 
into infinitely many branches, one for each choice of address). 
Therefore, concrete execution itself cannot serve as an 
algorithm for checking program safety. 

To obtain an algorithm, we define new kinds of executions that 
do not exhibit infinitely long traces and/or infinite 
branching. Specifically, in Section~\ref{sec:scexec} we define 
\emph{semiconcrete execution} ($\mathsf{scexec}$), where we use routine contracts 
and loop invariants, expressed as \emph{assertions} that use 
\emph{predicates} to denote data structures of potentially 
unbounded size, to limit the length of execution traces. 
Specifically, semiconcrete execution executes each routine 
separately, starting from an arbitrary initial state that 
satisfies the precondition, and checking that each final state 
satisfies the postcondition. Correspondingly, a routine call is 
executed using the callee's contract instead of its body. 
Similarly, a loop body is executed separately, starting from an 
arbitrary state that satisfies the loop invariant, and checking 
that each final state again satisfies the loop invariant. 
Execution of a loop first checks that the loop invariant holds 
on entry to the loop, and then updates the state to an 
arbitrary final state that satisfies the loop invariant. Since 
routine body and loop body executions are no longer inlined 
into the executions of their callers or loops, all executions 
have finite length.

However, semiconcrete execution still exhibits infinite 
branching; therefore, in Section~\ref{sec:symexec} we define 
the actual verification algorithm of Featherweight VeriFast, 
which we call \emph{symbolic execution} ($\mathsf{symexec}$). It builds on semiconcrete execution but eliminates infinite 
branching through the use of \emph{symbols} and a \emph{path 
condition}, such that a single symbol can be used to represent 
an infinite number of concrete values. Infinite branching is 
thus replaced by picking a \emph{fresh symbol}. A \emph{theorem 
prover} is used to decide equalities between terms and other conditions involving 
symbols under a given path condition.

Our solution to the verification problem is then to execute the 
program symbolically. Crucially, the executions are designed 
such that if symbolic execution of a program succeeds ($\mathsf{sym}\textsf{-}\mathsf{safe\_program}(c)$), then 
semiconcrete execution succeeds ($\mathsf{sc}\textsf{-}\mathsf{safe\_program}(c)$), and if semiconcrete execution 
succeeds, then concrete execution succeeds ($\mathsf{safe\_program}(c)$). These properties 
are called the soundness of symbolic execution and the 
soundness of semiconcrete execution, respectively. In the 
next two sections, we define these executions and sketch a 
proof of their soundness. 

\begin{defi}[Soundness]
$$\mathsf{safe\_program}(c) \Leftarrow \mathsf{sc}\textsf{-}\mathsf{safe\_program}(c) \Leftarrow \mathsf{sym}\textsf{-}\mathsf{safe\_program}(c)$$
\end{defi}

\section{Semiconcrete Execution}\label{sec:scexec}

In this section, we define semiconcrete execution, which introduces 
routine contracts and loop invariants to limit the length of execution 
traces. Routine contracts and loop invariants are specified using a 
language of \emph{assertions}, which specify both the \emph{facts} 
(boolean expressions) and the \emph{resources} (heap chunks) that are 
required or provided by a routine or loop body. To specify potentially 
unbounded data structures, \emph{predicates} are used, which are named, 
parameterized assertions which may be recursive, i.e.~mention themselves 
in their definition. 

The structure of this section is as follows. First, we 
introduce the new concepts involved in semiconcrete execution 
using a number of example programs and execution traces. Then, 
we formally define semiconcrete execution. Finally, we sketch 
an approach for proving that if a program is safe under 
semiconcrete execution, then it is safe under concrete 
execution, i.e.~semiconcrete execution is a sound approximation 
for checking the safety of a program under concrete execution.

\subsection{Annotations by Example}

In this subsection, we introduce the kinds of program annotations required by Featherweight VeriFast by means of some examples.

\begin{exa}{Annotations: Simple Example}

$$\begin{array}{l}
\mathbf{routine}\ \mathsf{swap}(\mathsf{cell1}, \mathsf{cell2})\\
\quad \annot{\mathbf{req}\ \mathsf{cell1} \mapsto {?}\mathsf{v1} * \mathsf{cell2} \mapsto {?}\mathsf{v2}}\\
\quad \annot{\mathbf{ens}\ \mathsf{cell1} \mapsto \mathsf{v2} * \mathsf{cell2} \mapsto \mathsf{v1}}\\
\quad =\\
\quad\quad \mathsf{value1} := [\mathsf{cell1}];\\
\quad\quad \mathsf{value2} := [\mathsf{cell2}];\\
\quad\quad [\mathsf{cell1}] := \mathsf{value2};\\
\quad\quad [\mathsf{cell2}] := \mathsf{value1}
\end{array}$$

\end{exa}

The example above shows a simple routine $\mathsf{swap}$ that 
swaps the values of two memory cells whose addresses are given 
by arguments $\mathsf{cell1}$ and $\mathsf{cell2}$. The body 
first reads the cells' original values into variables and then 
writes each cell's original value into the other cell. The 
routine has been annotated with a \emph{routine contract} 
consisting of a \emph{precondition} (also known as a 
\emph{requires clause}, denoted using keyword $\mathbf{req}$) 
and a \emph{postcondition} (also known as an \emph{ensures 
clause}, denoted using keyword $\mathbf{ens}$). The 
precondition describes the set of initial states accepted by 
the routine; the postcondition describes the set of final 
states generated by the routine when started from an initial 
state that satisfies the precondition.

The precondition of routine $\mathsf{swap}$ states that the 
routine requires two distinct memory cells to be present in the 
heap, one at address $\mathsf{cell1}$ and the other at address 
$\mathsf{cell2}$. Furthermore, it introduces two \emph{ghost 
variables} $\mathsf{v1}$ and $\mathsf{v2}$: it binds 
$\mathsf{v1}$ to the original value of the cell at address 
$\mathsf{cell1}$ and $\mathsf{v2}$ to the original value of the 
cell at address $\mathsf{cell2}$. In general, when a variable 
appears in an assertion immediately preceded by a question 
mark, this is called a \emph{variable pattern}. A variable 
pattern $?x$ introduces the variable $x$ and binds it to the 
value found in the heap corresponding to the position where the 
variable pattern appears.

In the example, the purpose of introducing the variables 
$\mathsf{v1}$ and $\mathsf{v2}$ in the precondition is so that 
they can be used in the postcondition to specify the 
relationship between the initial state and the final state of 
the routine. Specifically, the postcondition specifies that in 
the final state, the same memory cells are still present in the 
heap, and their value has changed such that the new value of 
the cell at address $\mathsf{cell1}$ equals the original value 
of the cell at address $\mathsf{cell2}$ and vice versa.

Notice that the assertions that serve as the precondition and 
the postcondition of routine $\mathsf{swap}$ specify only resources (heap 
chunks). In general, assertions may also specify facts (boolean expressions).
Correspondingly, there are two kinds of elementary 
assertions: boolean expressions and \emph{predicate 
assertions}. Elementary assertions can be composed using the 
\emph{separating conjunction} $*$. Its meaning is that the 
facts on the left and the facts on the right are both true, and 
that furthermore the resources on the left and the resources on 
the right are both present \emph{separately}, i.e.~the heap can 
be split into two parts such that the resources specified by 
the left-hand side of the assertion are in one part and the 
resources specified by the right-hand side of the assertion are 
in the other part. Notice how, in this respect, separating 
conjunction differs from ordinary logical conjunction (AND): we 
have that $a$ is equivalent to $a \land a$, but we do not have 
that $a$ is equivalent to $a * a$. In particular, $a * a$ 
specifies that the heap contains \emph{two occurrences} of each 
resource specified by $a$, which generally is not possible, and 
therefore $a * a$ is generally unsatisfiable. This also means 
that the precondition of routine $\mathsf{swap}$ implies that 
$\mathsf{cell1}$ and $\mathsf{cell2}$ denote distinct addresses. 

\begin{figure}
$$\begin{array}{l}
\annot{\begin{array}{@{} l @{}}
\mathbf{predicate}\ \mathsf{list}(\mathsf{l})\;=\\
\quad \mathbf{if}\ \mathsf{l} = 0\ \mathbf{then}\ 0 = 0\ \mathbf{else}\\
\quad\quad \mathsf{mb}(\mathsf{l}, 2) * \mathsf{l} \mapsto {?}\mathsf{v} * \mathsf{l} + 1 \mapsto {?}\mathsf{n} * \mathsf{list}(\mathsf{n})
\end{array}}\\
\\
\mathbf{routine}\ \mathsf{range}(\mathsf{i}, \mathsf{n}, \mathsf{result})\\
\quad \annot{\mathbf{req}\ \mathsf{result} \mapsto {?}\mathsf{dummy}}\\
\quad \annot{\mathbf{ens}\ \mathsf{result} \mapsto {?}\mathsf{list} * \mathsf{list}(\mathsf{list})}\\
\quad =\\
\quad\quad \mathbf{if}\ \mathsf{i} = \mathsf{n}\ \mathbf{then}\ \mathsf{head} := 0\ \mathbf{else}\ (\\
\quad\quad\quad \mathsf{head} := \mathbf{malloc}(2);\\
\quad\quad\quad [\mathsf{head}] := \mathsf{i};\\
\quad\quad\quad \mathsf{range}(\mathsf{i} + 1, \mathsf{n}, \mathsf{head} + 1)\\
\quad\quad );\\
\quad\quad \annot{\mathbf{close}\ \mathsf{list}(\mathsf{head})}; [\mathsf{result}] := \mathsf{head}
\end{array}$$
\caption{Annotations: Predicates}\label{fig:ex-preds}
\end{figure}

Now, consider again the example routine $\mathsf{range}$ that 
we introduced earlier. Recall that this routine builds a linked 
list holding the values between argument $\mathsf{i}$, 
inclusive, and argument $\mathsf{n}$, exclusive, and writes the 
address of the first node into the previously allocated memory 
cell whose address is given by argument $\mathsf{result}$. 

We show a contract for this routine in Figure~\ref{fig:ex-preds}. The precondition 
specifies that a memory cell must exist at the address given by 
argument $\mathsf{result}$. The postcondition specifies that 
this memory cell still exists, and that it now points to a 
linked list. The latter guarantee is specified using the 
\emph{predicate} $\mathsf{list}$, defined above. The definition 
of the predicate declares one parameter, $\mathsf{l}$, and a body, which is an assertion. The body 
performs a case analysis on whether $\mathsf{l}$ equals 
0. If so, it specifies only the trivial fact that 0 
equals 0, i.e.~it does not specify anything. Otherwise, it 
specifies that the heap contains a malloc block chunk of size 2 
at address $\mathsf{l}$, as well as two memory cells, at 
addresses $\mathsf{l}$ and $\mathsf{l} + 1$, as well as another 
linked list pointed to by the memory cell at address 
$\mathsf{l} + 1$.

Technically, what happens is that the predicate definition 
introduces a new kind of chunk, or more specifically, a new 
chunk name, and allows this chunk name to be used in predicate 
assertions. As a result, in semiconcrete execution, there are 
two kinds of predicates: the built-in predicates $\mathbf{mb}$ 
and $\mapsto$, and the user-defined predicates. 
Correspondingly, the heap contains two kinds of chunks: those 
whose name is a built-in predicate, and those whose name is a 
user-defined predicate. The purpose of chunks corresponding to 
user-defined predicates is to ``bundle up'' zero or more malloc 
block chunks and points-to chunks, along with some facts. Such 
``bundling up'' is necessary for writing contracts for routines 
that manipulate data structures of unbounded size. For example, 
it is impossible to write a postcondition for routine 
$\mathsf{range}$ without using user-defined predicates: a postcondition
that contains $m$ points-to assertions cannot describe a linked list
of length greater than $m$, so such a postcondition does not hold for a call of $\mathsf{range}$ where $\mathsf{n} - \mathsf{i} > m$.

The built-in chunks are created by the $\mathbf{malloc}$ 
statement. How are the user-defined chunks created? To enable 
the creation of user-defined chunks, semiconcrete execution 
introduces a new form of commands into the programming 
language, called $\mathbf{close}$ commands. The command 
$\mathbf{close}\ p(\overline{e})$ requires that $p$ is a user-defined 
predicate; it removes from the heap the chunks described by the 
body of the predicate, and checks the facts required by the body of the predicate,
and then adds a user-defined chunk whose name 
is $p$ and whose arguments are the values of $\overline{e}$. That 
is, the command bundles up the resources and facts described by 
the body of predicate $p$ into a chunk named $p$.

In the example, the body of routine $\mathsf{range}$, after 
allocating the first node and performing the recursive call to 
build the rest of the linked list, performs a $\mathbf{close}$ 
operation to bundle the three chunks of the first node and the 
$\mathsf{list}$ chunk that represents the rest of the linked 
list together into a single $\mathsf{list}$ chunk.

\subsection{Syntax of Annotations}

In summary, the programming language syntax extensions 
introduced by semiconcrete execution are as follows.

A program may now declare a number of \emph{routine 
specifications} $\mathit{rspec}$ of the form $\mathbf{routine}\ 
r(\overline{x})\ \mathbf{req}\ a\ \mathbf{ens}\ a'$, which 
associate with the routine name $r$ and parameter list 
$\overline{x}$ the precondition $a$ and postcondition $a'$, 
which are assertions. Furthermore, the syntax of loops is 
extended to include a loop invariant clause $\mathbf{inv}\ a$, 
where $a$ is an assertion. Furthermore, a program may declare a 
number of predicate definitions, which associate a predicate 
name and a list of parameters with a body, which is an 
assertion. An assertion $a$ is a boolean expression $b$, a 
predicate assertion $p(\overline{e}, \overline{{?}x})$ (where $p$ is either a 
built-in predicate or a user-defined predicate), a separating 
conjunction $a 
* a'$, or a conditional assertion $\mathbf{if}\ b\ \mathbf{then}\ a\ 
\mathbf{else}\ a'$. Two new commands are introduced: the 
$\mathbf{open}$ command and the $\mathbf{close}$ command. The 
$\mathbf{open}$ command performs the inverse operation of the 
$\mathbf{close}$ command: it unbundles a user-defined chunk, 
i.e.~it removes the user-defined chunk from the heap and adds 
the chunks described by the body of the predicate.

\begin{defi}{Annotations}\label{defi:annotations}

$$\begin{array}{r l}
& q \in \mathit{UserDefinedPredicates}\\
p ::= & \mapsto\ |\ \mathsf{mb}\ |\ q\\
a ::= & b\ |\ p(\overline{e}, \overline{{?}x})\ |\ a * a\ |\ \mathbf{if}\ b\ \mathbf{then}\ a\ \mathbf{else}\ a\\
\mathit{preddef} ::= & \mathbf{predicate}\ q(\overline{x}) = a\\
c ::= & \cdots\ |\ \mathbf{while}\ b\ \mathbf{inv}\ a\ \mathbf{do}\ c\ |\ \mathbf{open}\ q(\overline{e})\ |\ \mathbf{close}\ q(\overline{e})\\
\mathit{rspec} ::= & \mathbf{routine}\ r(\overline{x})\ \mathbf{req}\ a\ \mathbf{ens}\ a
\end{array}$$

\begin{center}
$e \mapsto {?}x$ is alternative syntax for ${\mapsto}(e, {?}x)$
\end{center}

\end{defi}

\subsection{Semiconcrete Execution: Example Trace}

Recall the example concrete execution trace of routine 
$\mathsf{range}$ in Figure~\ref{fig:larger-cexec-trace}. 
Recall that the notable features of this trace are that the 
trace is long, since it contains three nested executions of 
routine $\mathsf{range}$; that the heap is large, since it 
includes the entire heap that existed on entry to the routine, 
as well as all of the chunks produced by all of the nested 
calls; and that there is infinite branching. 

\begin{figure}
$$\begin{array}{l}
\mathbf{routine}\ \mathsf{range}(\mathsf{i}, \mathsf{n}, \mathsf{r})\\
\quad \annot{\mathbf{req}\ \mathsf{r} \mapsto {?}\mathsf{dummy}\ \mathbf{ens}\ \mathsf{r} \mapsto {?}\mathsf{list} * \mathsf{list}(\mathsf{list})}\\
\comment{s{:} \mathbf{0}[\mathsf{i}{:}\branching{5},\mathsf{n}{:}\branching{8},\mathsf{r}{:}\branching{41}], h {:} \mathbf{0}}\\
{\color{Gray}\mathit{produce}(}\annot{\mathsf{r} \mapsto {?}\mathsf{dummy}}{\color{Gray})}\\
\comment{s{:} \mathbf{0}[\mathsf{i}{:}5,\mathsf{n}{:}8,\mathsf{r}{:}41], h {:} \llbrace 41 {\mapsto} \branching{77}\rrbrace}\\
\mathbf{if}\ \mathsf{i} = \mathsf{n}\ \mathbf{then}\ \mathsf{l} := 0\ \mathbf{else}\ (\\
\mathsf{l} := \mathbf{malloc}(2);\\
\comment{s{:} \mathbf{0}[\mathsf{i}{:}5,\mathsf{n}{:}8,\mathsf{r}{:}41,\mathsf{l}{:}\branching{50}], h{:} \llbrace 41 {\mapsto}77{,} \mathsf{mb}(\branching{50}, 2){,} \branching{50}{\mapsto}\branching{88}{,} \branching{51} {\mapsto} \branching{99}\rrbrace}\\{}
[\mathsf{l}] := \mathsf{i}; \mathsf{range}(\mathsf{i} + 1, \mathsf{n}, \mathsf{l} + 1)\\
{\color{Gray}\mathit{consume}(}\annot{\mathsf{l} {+} 1 {\mapsto} {?}\mathsf{dummy}}{\color{Gray});\mathit{produce}(}\annot{\mathsf{l}{+}1 {\mapsto} {?}\mathsf{list} * \mathsf{list}(\mathsf{list})}{\color{Gray})}\\
\comment{s{:}\mathbf{0}[\mathsf{i}{:}5,\mathsf{n}{:}8,\mathsf{r}{:}41,\mathsf{l}{:}50], h{:} \llbrace 41 {\mapsto} 77{,} \mathsf{mb}(50, 2){,} 50{\mapsto}5{,} 51 {\mapsto} \branching{60}{,}\mathsf{list}(\branching{60})\rrbrace}\\
);\ \annot{\mathbf{close}\ \mathsf{list}(\mathsf{l})}; [\mathsf{r}] := \mathsf{l}\\
\comment{s{:}\mathbf{0}[\mathsf{i}{:}5,\mathsf{n}{:}8,\mathsf{r}{:}41,\mathsf{l}{:}50],h{:}\llbrace 41{\mapsto} 50{,}\mathsf{list}(50)\rrbrace}\\
{\color{Gray}\mathit{consume}(}\annot{\mathsf{r} \mapsto {?}\mathsf{list} * \mathsf{list}(\mathsf{list})}{\color{Gray})}\\
\comment{s{:}\mathbf{0}[\mathsf{i}{:}5,\mathsf{n}{:}8,\mathsf{r}{:}41,\mathsf{l}{:}50],h{:}\mathbf{0}}
\end{array}$$
\caption{Semiconcrete Execution: Example Trace}\label{fig:scexec-trace}
\end{figure}

We show an example semiconcrete execution trace for routine 
$\mathsf{range}$ in Figure~\ref{fig:scexec-trace}. Recall that semiconcrete execution 
executes each routine separately. Therefore, the above trace is 
not an excerpt from a larger program trace; rather, it is a 
complete trace of the execution of routine $\mathsf{range}$.

Execution starts in a state where the store binds each 
parameter to an arbitrary argument value and the heap is empty. 
It then \emph{produces} the precondition: it adds the resources 
and assumes the facts specified by the precondition. When 
producing a predicate assertion $p(\overline{e}, \overline{?x})$, the values of the arguments corresponding to the variable patterns $\overline{?x}$
are arbitrary. In the example, a 
points-to chunk at the address given by parameter 
$\mathsf{result}$ is added to the heap.

The execution of the $\mathbf{malloc}$ command and the memory 
write command are the same as in the concrete execution. 

The routine call is executed not by inlining a nested execution 
of the body of the routine, but by using the contract: the 
precondition is \emph{consumed}, and then the postcondition is 
\emph{produced}. Consuming an assertion means removing the heap 
chunks and checking the facts specified by the assertion. If a 
fact specified by the assertion is false, execution fails. The 
net effect is that the points-to chunk at address 51 gets some 
arbitrary value (60 in this trace) and a $\mathsf{list}$ chunk 
is added whose argument is 60.

The $\mathbf{close}$ command collapses the four chunks 
representing the linked list into a single chunk 
$\mathsf{list}(50)$.

Finally, after execution of the routine body is complete, the 
postcondition is consumed. It removes all of the heap chunks 
and leaves the heap empty.

Generally, in semiconcrete execution, if the heap is left 
nonempty after a routine execution, this indicates a memory 
leak, since the memory described by the remaining chunks can no 
longer be accessed by any subsequent operation in the program 
execution. Indeed, of the heap chunks that exist at the end of 
a routine body execution, only the ones described by the 
postcondition become available to the caller; the others can no 
longer be retrieved in any way. Therefore, as the final step of 
a routine execution, semiconcrete execution checks that the 
heap is empty; if not, routine execution fails.

\subsection{Semiconcrete Execution: Types}

The set $\mathit{SCStates}$ of semiconcrete states is defined 
above; the only difference with the concrete states is that the 
predicates now include the user-defined predicates, and 
consequently the chunks now include the user-defined chunks.

To formally define semiconcrete command execution, we will 
define a function $\mathsf{scexec}$ from commands to mutators, 
similar to function $\mathsf{exec}$ for concrete execution. 
Additionally, we define functions $\mathsf{consume}$ and 
$\mathsf{produce}$ that formalize what it means to consume and 
produce an assertion, respectively.

\begin{defi}{Semiconcrete Execution: Types}

$$\begin{array}{r @{\ } l}
\mathit{SCStores} = & \mathit{Vars} \rightarrow \mathbb{Z}\\
\mathit{SCPredicates} = & \{\mapsto, \mathsf{mb}\} \cup \mathit{UserDefinedPredicates}\\
\mathit{SCChunks} = & \{p(\overline{v})\ |\ p \in \mathit{SCPredicates}, \overline{v} \in \mathbb{Z}\}\\
\mathit{SCHeaps} = & \mathit{SCChunks} \rightarrow \mathbb{N}\\
\mathit{SCStates} = & \mathit{SCStores} \times \mathit{SCHeaps}\\
\mathit{SCMutators} = & \mathit{SCStates} \rightarrow \mathit{Outcomes}(\mathit{SCStates})
\end{array}$$

$$\begin{array}{l l}
\mathsf{scexec} & \in \mathit{Commands} \rightarrow \mathit{SCMutators}\\
\mathsf{consume} & \in \mathit{Assertions} \rightarrow \mathit{SCMutators}\\
\mathsf{produce} & \in \mathit{Assertions} \rightarrow \mathit{SCMutators}
\end{array}$$

\end{defi}

\subsection{Some Auxiliary Definitions}

The definition of semiconcrete execution uses the following auxiliary 
mutators, in addition to the ones used by the definition of concrete 
execution. Semiconcrete consumption $\mathsf{consume\_chunks}(h)$ of a multiset of 
chunks $h$ fails if the heap does not contain these chunks, and otherwise 
removes them from the heap. It is identical to concrete consumption of 
chunks. Semiconcrete production $\mathsf{produce\_chunks}(h)$ of a multiset of 
chunks $h$ adds the chunks to the heap. It differs from concrete 
production in that it does not check that the added chunks do not clash 
with existing chunks in the heap. 
Semiconcrete consumption and production of a single chunk $\alpha$ are defined in the obvious way.
The mutator $\mathsf{assert}(b)$ 
asserting a boolean expression $b$ fails if $b$, evaluated in the current 
store, is false, and otherwise does nothing. 

\begin{defi}{Some Auxiliary Definitions}

$$\begin{array}{r @{\ } l}
\mathsf{consume\_chunks}(h') = & \lambda (s, h).\;\bigoplus h' \le h.\;\langle (s, h - h')\rangle\\
\mathsf{consume\_chunk}(\alpha) = & \mathsf{consume\_chunks}(\llbrace\alpha\rrbrace)\\
\mathsf{produce\_chunks}(h') = & \lambda (s, h).\;\langle (s, h \uplus h')\rangle\\
\mathsf{produce\_chunk}(\alpha) = & \mathsf{produce\_chunks}(\llbrace \alpha\rrbrace)\\
\mathsf{assert}(b) = & \lambda (s, h).\;\bigoplus \llbracket b\rrbracket_s = \mathsf{true}.\;\langle (s, h)\rangle
\end{array}$$

\end{defi}

\subsection{Producing Assertions}

Production of an assertion is defined as follows.

Production of a boolean expression means assuming it. Recall 
from the definition of $\mathsf{assume}$ on 
Page~\pageref{def:assume} that assuming a boolean expression is 
equivalent to a no-op if it evaluates to true, and equivalent 
to nontermination if it evaluates to false. The effect is that 
all final states generated by production satisfy the 
expression.

Production of a predicate assertion means demonically choosing 
a value for each variable pattern, binding the pattern variable 
to it, and adding the specified chunk to the heap.

Producing a separating conjunction means first producing the 
left-hand side and then producing the right-hand side. Notice 
that the variable bindings introduced by the left-hand side 
are active when producing the right-hand side. Notice also that 
this definition correctly captures the \emph{separating} aspect 
of the separating conjunction: if a chunk is specified by both 
the left-hand side and the right-hand side, two occurrences of 
it end up in the heap.

Producing a conditional assertion is defined analogously to
executing a conditional statement.

\begin{defi}{Producing Assertions}

$$\begin{array}{l}
\mathsf{produce}(b) = \mathsf{assume}(b)\\[.5em]

\mathsf{produce}(p(\overline{e}, \overline{{?}x})) =
  \overline{v} \leftarrow \mathsf{eval}(\overline{e});
  \bigotimes \overline{v}'.\;
  \mathsf{produce\_chunk}(p(\overline{v},\overline{v}'));
  \overline{x} := \overline{v}'\\[.5em]

\mathsf{produce}(a * a') = \mathsf{produce}(a); \mathsf{produce}(a')\\[.5em]

\mathsf{produce}(\mathbf{if}\ b\ \mathbf{then}\ a\ \mathbf{else}\ a') =
 \mathsf{assume}(b); \mathsf{produce}(a) \otimes \mathsf{assume}(\lnot b); \mathsf{produce}(a')
\end{array}$$

\end{defi}

\subsection{Consuming Assertions}

Consumption of an assertion is defined as follows.

Consuming a boolean expression is equivalent to a no-op if the 
expression evaluates to true under the current store; 
otherwise, it is equivalent to failure.

Consuming a predicate assertion fails unless there exists a value for
each variable pattern such that the specified chunk can be consumed.
If so, each pattern variable is bound to the corresponding value.

Consuming a separating conjunction first consumes the left-hand 
side and then consumes the right-hand side. Notice that this 
correctly reflects the \emph{separating} aspect of the 
separating conjunction: if the left-hand side and the 
right-hand side specify the same chunk, consumption fails 
unless the heap contains two occurrences of the chunk, which is 
generally impossible.

Consuming a conditional assertion is defined analogously to 
executing a conditional statement. 

\begin{defi}{Consuming Assertions}

$$\begin{array}{l}
\mathsf{consume}(b) = \mathsf{assert}(b)\\[.5em]

\mathsf{consume}(p(\overline{e}, \overline{{?}x})) =\\
\quad
  \overline{v} \leftarrow \mathsf{eval}(\overline{e});
  \bigoplus \overline{v}'.\;
  \mathsf{consume\_chunk}(p(\overline{v},\overline{v}'));
  \overline{x} := \overline{v}'\\[.5em]

\mathsf{consume}(a * a') = \mathsf{consume}(a); \mathsf{consume}(a')\\[.5em]

\mathsf{consume}(\mathbf{if}\ b\ \mathbf{then}\ a\ \mathbf{else}\ a') =\\
\quad \mathsf{assume}(b); \mathsf{consume}(a) \otimes \mathsf{assume}(\lnot b); \mathsf{consume}(a')
\end{array}$$

\end{defi}

\subsection{Semiconcrete Execution of Commands}

\begin{figure}
$$\begin{array}{l}
\mathsf{scexec}(x := e) = v \leftarrow \mathsf{eval}(e); x := v\\
\\
\mathsf{scexec}(c; c') = \mathsf{scexec}(c); \mathsf{scexec}(c')\\
\\
\mathsf{scexec}(\mathbf{if}\ b\ \mathbf{then}\ a\ \mathbf{else}\ a') =\\
\quad \mathsf{assume}(b); \mathsf{scexec}(c) \otimes \mathsf{assume}(\lnot b); \mathsf{scexec}(c')\\
\\
\mathsf{scexec}(\mathbf{while}\ e\ \mathbf{inv}\ a\ \mathbf{do}\ c) = \textrm{See Figure~\ref{fig:scexec-while}}\\
\\
\mathsf{scexec}(r(\overline{e})) =
  \overline{v} \leftarrow \mathsf{eval}(\overline{e});
  \mathsf{with}(\mathbf{0}[\overline{x}:=\overline{v}], \mathsf{consume}(a); \mathsf{produce}(a'))\\
\quad \textrm{where $\mathbf{routine}\ r(\overline{x})\ \mathbf{req}\ a\ \mathbf{ens}\ a'$}\\
\\
\mathsf{scexec}(x := \mathbf{malloc}(n)) =\\
\quad\bigotimes \ell, v_1,\dots,v_n \in \mathbb{Z}.\\
\quad\quad \mathsf{produce\_chunks}(\llbrace\mathsf{mb}(\ell, n), \ell \mapsto v_1,\dots,\ell + n - 1\mapsto v_n\rrbrace); x := \ell\\
\\
\mathsf{scexec}(x := [e]) = \ell \leftarrow \mathsf{eval}(e);\\
\quad \bigoplus v.\;\mathsf{consume\_chunk}(\ell \mapsto v); \mathsf{produce\_chunk}(\ell \mapsto v); x := v\\
\\
\mathsf{scexec}([e] := e') = \ell \leftarrow \mathsf{eval}(e); v \leftarrow \mathsf{eval}(e');\\
\quad \bigoplus v_0.\;\mathsf{consume\_chunk}(\ell \mapsto v_0); \mathsf{produce\_chunk}(\ell \mapsto v)\\
\\
\mathsf{scexec}(\mathbf{free}(e)) = \ell \leftarrow \mathsf{eval}(e);\\
\quad \bigoplus N \in \mathbb{N}, v_1, \dots, v_N \in \mathbb{Z}.\\
\quad\quad \mathsf{consume\_chunks}(\llbrace\mathsf{mb}(\ell, N), \ell \mapsto v_1, \dots, \ell + N - 1 \mapsto v_N\rrbrace)\\
\\
\mathsf{scexec}(\mathbf{open}\ p(\overline{e})) = \overline{v} \leftarrow \mathsf{eval}(\overline{e});\\
\quad \mathsf{consume\_chunk}(p(\overline{v})); \mathsf{with}(\mathbf{0}[\overline{x}:=\overline{v}], \mathsf{produce}(a))\\
\quad\textrm{where $\mathbf{predicate}\ p(\overline{x}) = a$}\\
\\
\mathsf{scexec}(\mathbf{close}\ p(\overline{e})) = \overline{v} \leftarrow \mathsf{eval}(\overline{e});\\
\quad \mathsf{with}(\mathbf{0}[\overline{x}:=\overline{v}], \mathsf{consume}(a)); \mathsf{produce\_chunk}(p(\overline{v}))\\
\quad\textrm{where $\mathbf{predicate}\ p(\overline{x}) = a$}
\end{array}$$
\caption{Semiconcrete Execution of Commands}\label{fig:scexec}
\end{figure}

\begin{figure}
$$\begin{array}{l}
\begin{array}{@{} l c l}
\mathsf{targets}(x := e) & = & \{x\}\\
\mathsf{targets}(c_1; c_2) & = & \mathsf{targets}(c_1) \cup \mathsf{targets}(c_2)\\
\mathsf{targets}(\mathbf{if}\ b\ \mathbf{then}\ c_1\ \mathbf{else}\ c_2) & = & \mathsf{targets}(c_1) \cup \mathsf{targets}(c_2)\\
\mathsf{targets}(r(\overline{e})) & = & \emptyset\\
\mathsf{targets}(\mathbf{while}\ b\ \mathbf{inv}\ a\ \mathbf{do}\ c_0) & = & \mathsf{targets}(c_0)\\
\mathsf{targets}(x := \mathbf{malloc}(n)) & = & \{x\}\\
\mathsf{targets}(x := [e]) & = & \{x\}\\
\mathsf{targets}([e] := e') & = & \emptyset\\
\mathsf{targets}(\mathbf{free}(e)) & = & \emptyset\\
\end{array}\\
\\
\mathsf{havoc}(\overline{x}) = \lambda (s, h).\; \bigotimes \overline{v} \in \mathbb{Z}.\;\langle(s[\overline{x}:=\overline{v}], h)\rangle\\
\mathsf{leakcheck} = \lambda (s, h).\; \bigoplus h = \mathbf{0}.\;\top\\
\\
\mathsf{scexec}(\mathbf{while}\ b\ \mathbf{inv}\ a\ \mathbf{do}\ c) =\\
\quad s \leftarrow \mathsf{store}; \mathsf{with}(s, \mathsf{consume}(a));\\
\quad \mathsf{havoc}(\mathsf{targets}(c));\\
\quad (\\
\quad\quad \mathsf{heap} := \mathbf{0};\\
\quad\quad s \leftarrow \mathsf{store}; \mathsf{with}(s, \mathsf{produce}(a));\\
\quad\quad \mathsf{assume}(b); \mathsf{scexec}(c);\\
\quad\quad s \leftarrow \mathsf{store}; \mathsf{with}(s, \mathsf{consume}(a));\\
\quad\quad \mathsf{leakcheck}\\
\quad \otimes\\
\quad\quad s \leftarrow \mathsf{store}; \mathsf{with}(s, \mathsf{produce}(a));\\
\quad\quad \mathsf{assume}(\lnot b)\\
\quad )\\
\end{array}$$
\caption{Semiconcrete Execution of Loops}\label{fig:scexec-while}
\end{figure}

Recall that the concrete execution function $\mathsf{exec}$ is 
defined in terms of the helper function $\mathsf{exec}_n$. The 
latter function is defined by recursion on $n$.

In contrast, the semiconcrete execution function 
$\mathsf{scexec}$ is defined directly, by recursion on the 
structure of the command. Doing so for the concrete execution 
function would not have been possible, since the execution of a 
routine call involves the execution of the callee's body, which 
obviously is not part of the structure of the call command 
itself. However, since in semiconcrete execution routine call 
involves only production and consumption of assertions, this 
simple approach is possible here.

\begin{defi}{Semiconcrete Execution of Commands}
See Figures~\ref{fig:scexec} and \ref{fig:scexec-while}.
\end{defi}

Execution of assignments, sequential compositions, and 
conditional assertions is the same as in concrete execution.

Execution of a routine call $r(\overline{e})$ looks up routine 
$r$'s precondition $a$ and postcondition $a'$, to be 
interpreted under a parameter list $\overline{x}$, and sets up 
a store that binds the parameters to the values of the 
arguments. In this store, it first consumes the precondition 
and then produces the postcondition. Notice that the variable 
bindings generated during consumption of the precondition are 
active during production of the postcondition, since the output 
store of the consumption operation serves as the input store of 
the production operation. 

Execution of a while loop is relatively complex. It proceeds as 
follows:
\begin{itemize}
\item The loop invariant is consumed.

\item An arbitrary new value is assigned to each variable 
    modified by the loop body. 
    
\item Execution chooses demonically between two branches:

\begin{itemize}
\item In the first branch, execution proceeds as 
    follows:
    
\begin{itemize}
\item The heap is emptied, so that heap chunks not 
    described by the loop invariant are not 
    available to the loop body. 
    
\item The loop invariant is produced, but the 
    resulting variable bindings are discarded. 
    
\item It is assumed that the loop condition holds. 

\item The loop body is executed.

\item The loop invariant is consumed.

\item A leak check is performed, i.e.~execution 
    fails if the heap is not empty; otherwise, 
    execution blocks.

\end{itemize}

\item In the second branch, the loop invariant is 
    produced (but the resulting variable bindings are 
    discarded), and it is assumed that the loop 
    condition does not hold.
    
\end{itemize}

\end{itemize}

The definition uses the auxiliary functions $\mathsf{targets}$, 
$\mathsf{havoc}$, and $\mathsf{leakcheck}$.

Function $\mathsf{targets}$ maps a command to the set of 
variables modified by the command.

Function $\mathsf{havoc}(\overline{x})$ demonically chooses a 
value for each variable in $\overline{x}$ and assigns it to the 
corresponding variable.

Function $\mathsf{leakcheck}$ fails if the heap is nonempty, 
and otherwise blocks, i.e.~does not terminate. 

Semiconcrete execution of memory block allocation, memory read, 
memory write, and memory block deallocation are the same as in 
concrete execution.

Execution of an $\mathbf{open}$ command $\mathbf{open}\ 
p(\overline{e})$ first consumes the chunk whose name is $p$ and 
whose arguments are the values of $\overline{e}$ and then 
produces the body of predicate $p$, the latter in a store that 
binds the predicate parameters $\overline{x}$ to the values of 
the argument expressions $\overline{e}$. 

Conversely, execution of a $\mathbf{close}$ command 
$\mathbf{close}\ p(\overline{e})$ first consumes the body of 
predicate $p$ in a store that binds the predicate parameters 
$\overline{x}$ to the values of the argument expressions 
$\overline{e}$. Then it produces the chunk whose name is $p$ 
and whose arguments are the values of $\overline{e}$. 

\subsection{Validity of Routines}

In concrete execution, safety of a program simply means that 
execution of the main command starting from an empty state does 
not fail. In semiconcrete execution, safety of a program means 
that two things are true: 1) execution of the main command 
starting from the empty state does not fail; and 2) all 
routines are \emph{valid}.

Validity of a routine means that its body satisfies its 
contract. More specifically, it means that the \emph{routine 
validity mutator} does not fail, when started from an empty 
state. The routine validity mutator for a given routine $r$ 
proceeds as follows: 1) it sets up a store that binds each of the routine's parameters to a demonically chosen value;
2) it produces the routine precondition; 
3) it semiconcretely executes the routine body; 4) it consumes 
the routine postcondition; 5) it checks for leaks.

\begin{defi}{Validity of Routines}

$$\begin{array}{l}
\mathsf{valid}(r) =\\
\quad (\mathbf{0}, \mathbf{0})\; \triangleright\\
\quad \bigotimes \overline{v}.\\
\quad \mathsf{with}(\mathbf{0}[\overline{x}:=\overline{v}],\\
\quad\quad s' \leftarrow \mathsf{with}(\mathbf{0}[\overline{x}:=\overline{v}], \mathsf{produce}(a); \mathsf{store});\\
\quad\quad \mathsf{scexec}(c);\\
\quad\quad \mathsf{with}(s', \mathsf{consume}(a'))\\
\quad );\\
\quad \mathsf{leakcheck}\\
\quad \{\mathrm{true}\}\\
\quad \textrm{where $\mathbf{routine}\ r(\overline{x})\ \mathbf{req}\ a\ \mathbf{ens}\ a' = c$}\\
\end{array}$$

\end{defi}

Notice that the postcondition is consumed starting from the store saved after producing the precondition.
This ensures that the variable bindings generated by producing the precondition are
visible when consuming the postcondition.

\subsection{Semiconcrete Execution: Program Safety}

As stated before, safety of a program in semiconcrete execution 
means that execution of the main command succeeds when started 
from the empty state, \emph{and} that all routines are valid. 

\begin{defi}{Semiconcrete Execution: Program Safety}
\[
\mathsf{sc}\textsf{-}\mathsf{safe\_program}(c)\quad=\quad (\forall r.\;\mathsf{valid}(r)) \land \sigma_0 \triangleright \mathsf{scexec}(c)\;\{\mathrm{true}\}
\]
where $r$ ranges over the declared routines of the program.
\end{defi}

\subsection{Soundness}

Now that we have defined safety of a program in semiconcrete 
execution, we discuss its relationship with safety of the 
program in concrete execution. The intended relationship is that if a program is safe in 
semiconcrete execution (i.e.~all routines are valid and the 
main command does not fail when executed semiconcretely 
starting from the empty state), then it is safe in concrete 
execution (i.e.~the main command does not fail when executed 
concretely starting from the empty state). We call this 
property the \emph{soundness} of semiconcrete execution. In the 
remainder of this section, we sketch a proof of this property.

\subsubsection{Properties of Assertion Consumption and Production}

First, we discuss some properties of assertion consumption and 
production. To gain more insight into consumption and production, we here 
offer an alternative definition of them, in terms of 
\emph{consumption and production arrows} ${\xrightarrow{a}_\mathsf{c}}, {\xrightarrow{a}_\mathsf{p}} \subseteq \mathit{SCStates} \times \mathit{SCStates}$,
defined inductively using the inference rules shown below.
$\sigma \xrightarrow{a}_\mathsf{c} \sigma'$ means that consumption of assertion $a$ starting from state $\sigma$ succeeds and results in state $\sigma'$.
Similarly, $\sigma \xrightarrow{a}_\mathsf{p} \sigma'$ means that production of assertion $a$ starting from state $\sigma$ results in
state $\sigma'$.

\begin{defi}{The Consumption Arrow}
\begin{mathpar}
\inferrule{
\llbracket b\rrbracket_s = \mathsf{true}
}{
(s, h) \xrightarrow{b}_\mathsf{c} (s, h)
}
\and
\inferrule{
h = \llbrace p(\llbracket \overline{e}\rrbracket_s,\overline{v})\rrbrace \uplus h'
}{
(s, h) \xrightarrow{p(\overline{e},\overline{?x})}_\mathsf{c} (s[\overline{x}:=\overline{v}], h')
}
\and
\inferrule{
(s, h) \xrightarrow{a}_\mathsf{c} (s', h')\\
(s', h') \xrightarrow{a'}_\mathsf{c} (s'', h'')
}{
(s, h) \xrightarrow{a * a'}_\mathsf{c} (s'', h'')
}
\\
\inferrule{
\llbracket b\rrbracket_s = \mathsf{true}\\
(s, h) \xrightarrow{a}_\mathsf{c} (s', h')
}{
(s, h) \xrightarrow{\mathbf{if}\ b\ \mathbf{then}\ a\ \mathbf{else}\ a'}_\mathsf{c} (s', h')
}
\and
\inferrule{
\llbracket b\rrbracket_s = \mathsf{false}\\
(s, h) \xrightarrow{a'}_\mathsf{c} (s', h')
}{
(s, h) \xrightarrow{\mathbf{if}\ b\ \mathbf{then}\ a\ \mathbf{else}\ a'}_\mathsf{c} (s', h')
}
\end{mathpar}

\end{defi}

\begin{defi}{The Production Arrow}
\begin{mathpar}
\inferrule{
\llbracket b\rrbracket_s = \mathsf{true}
}{
(s, h) \xrightarrow{b}_\mathsf{p} (s, h)
}
\and
\inferrule{
h' = \llbrace p(\llbracket \overline{e}\rrbracket_s,\overline{v})\rrbrace \uplus h
}{
(s, h) \xrightarrow{p(\overline{e},\overline{?x})}_\mathsf{p} (s[\overline{x}{:=}\overline{v}], h')
}
\and
\inferrule{
(s, h) \xrightarrow{a}_\mathsf{p} (s', h')\\
(s', h') \xrightarrow{a'}_\mathsf{p} (s'', h'')
}{
(s, h) \xrightarrow{a * a'}_\mathsf{p} (s'', h'')
}
\\
\inferrule{
\llbracket b\rrbracket_s = \mathsf{true}\\
(s, h) \xrightarrow{a}_\mathsf{p} (s', h')
}{
(s, h) \xrightarrow{\mathbf{if}\ b\ \mathbf{then}\ a\ \mathbf{else}\ a'}_\mathsf{p} (s', h')
}
\and
\inferrule{
\llbracket b\rrbracket_s = \mathsf{false}\\
(s, h) \xrightarrow{a'}_\mathsf{p} (s', h')
}{
(s, h) \xrightarrow{\mathbf{if}\ b\ \mathbf{then}\ a\ \mathbf{else}\ a'}_\mathsf{p} (s', h')
}
\end{mathpar}

\end{defi}

Notice that the only difference between the two definitions is 
the different positions of $h$ and $h'$ in the rule for 
predicate assertions. Consumption of predicate assertions removes matching chunks, whereas production adds matching chunks.

Notice that in both cases, there are generally multiple output states for any given input state: in both cases,
there is a distinct output state for each distinct binding of values to pattern variables in predicate assertions. However, this is much
more common in the case of production than in the case of consumption, since in the case of consumption multiple bindings are possible only
if the heap contains multiple chunks that match the predicate assertion.

Given the consumption and production arrows, we can give an alternative definition of the consumption and production mutators, as shown below.

\begin{lem}[Consumption and Production and the Arrows]
$$\begin{array}{r l}
\mathsf{consume}(a) = \lambda \sigma.\;\bigoplus \sigma', \sigma \xrightarrow{a}_\mathsf{c} \sigma'.\;\langle \sigma'\rangle\\
\mathsf{produce}(a) = \lambda \sigma.\;\bigotimes \sigma', \sigma \xrightarrow{a}_\mathsf{p} \sigma'.\;\langle \sigma'\rangle
\end{array}$$
\end{lem}

Notice that the consumption mutator chooses angelically among the output states, and fails if there are none; production chooses demonically
among the output states, and blocks if there are none.

We can easily prove some important properties of the consumption and production arrows.
Firstly, consumption is local: if consumption succeeds, then it also succeeds if more chunks are available, and those additional chunks
remain untouched.
\begin{lem}[Consumption Locality]\label{lemma:conlocal}
\display{(s, h) \xrightarrow{a}_\mathsf{c} (s', h') \Rightarrow (s, h \uplus h'') \xrightarrow{a}_\mathsf{c} (s', h' \uplus h'')}
\end{lem}

Secondly, consumption is monotonic: if it succeeds, then the resulting heap is a sub-multiset of the original heap, and
consumption also succeeds if only the consumed chunks are available, and then it yields the empty heap.
\begin{lem}[Consumption Monotonicity]\label{lemma:conadd}
\display{(s, h) \xrightarrow{a}_\mathsf{c} (s', h') \Rightarrow \exists h''.\;h = h' \uplus h'' \land (s, h'') \xrightarrow{a}_\mathsf{c} (s', \mathbf{0})}
\end{lem}

Thirdly,
production is the converse of consumption: production adds back the chunks removed by consumption.
\begin{lem}[Production after Consumption (Arrows)]\label{lemma:conprod}
\display{(s, h) \xrightarrow{a}_\mathsf{c} (s', \mathbf{0}) \Rightarrow (s, h'') \xrightarrow{a}_\mathsf{p} (s', h'' \uplus h)}
\end{lem}
All of these properties are proved easily
by induction on the assertion.

From these properties of the consumption and production arrows, we can 
easily derive corresponding properties of the consumption and production 
mutators:
\begin{lem}[Consumption and Production (with Post-stores)]
\display{\begin{array}{l}
s_1 \leftarrow \mathsf{with}(s, \mathsf{consume}(a); \mathsf{store});\\
s_2 \leftarrow \mathsf{with}(s, \mathsf{produce}(a); \mathsf{store});\\
C(s_1, s_2)\\
\Rrightarrow\\
\bigoplus s'.\;C(s', s')
\end{array}}
\end{lem}

\begin{lem}[Consumption and Production]
\display{\mathsf{with}(s, \mathsf{consume}(a)); \mathsf{with}(s, \mathsf{produce}(a)) \Rrightarrow \mathsf{noop}}
\end{lem}

The first lemma states that consuming an assertion and then 
producing the same assertion starting from the same store, and then 
performing some mutator $C({-}, {-})$ parameterized by the output 
stores of the consumption and production, safely approximates doing 
nothing to the heap and angelically picking a store and performing $C$ 
using this store for both parameters. The second lemma is a simplified 
version that ignores the output stores: consuming an assertion and then 
producing the same assertion starting from the same store safely 
approximates doing nothing.

\subsubsection{Locality and Modifies}

Two important but simple properties of semiconcrete execution are that it 
is local and that it modifies only the command's targets. Locality means 
that execution under some initial heap and then adding more chunks safely 
approximates first adding those chunks and then executing.

\begin{defi}{Locality}
$$\mathsf{local}\;C\quad\Leftrightarrow\quad \forall h.\;C{;}, \mathsf{produce\_chunks}(h) \Rrightarrow \mathsf{produce\_chunks}(h); C$$
\end{defi}

\begin{lem}[Locality of Semiconcrete Execution]
$$\mathsf{local}\;\mathsf{scexec}(c)$$
\end{lem}

\begin{defi}[Modifies]
$$s \stackrel{\overline{x}}{\sim} s' \Leftrightarrow s[\overline{x}:=0] = s'[\overline{x}:=0]$$

$$\mathsf{modified}_{\overline{x}}(s') = \lambda (s, h).\;\bigoplus s \stackrel{\overline{x}}{\sim} s'.\;\mathsf{noop}$$

$$\mathsf{modifies}_{\overline{x}}\;C\quad \Leftrightarrow\quad
\forall s.\;\mathsf{modified}_{\overline{x}}(s); C \Rrightarrow C{;},\mathsf{modified}_{\overline{x}}(s)$$
\end{defi}

\begin{lem}[Semiconcrete Execution Modifies Targets]
\display{\mathsf{modifies}_{\mathsf{targets}(c)}\;\mathsf{scexec}(c)}
\end{lem}

\subsubsection{Heap Refinement}

Having discussed the properties of assertion consumption and production, 
we now discuss the relationship between semiconcrete command execution and 
concrete command execution. For this purpose, we need to characterize the 
relationship between semiconcrete states and concrete states. 

We say that a concrete heap $h_\mathrm{c}$ \emph{refines} a semiconcrete 
heap $h$, denoted $h_\mathrm{c} \triangleleft h$, if $h$ can be obtained 
from $h_\mathrm{c}$ by \emph{closing} some finite number of user-defined 
predicate chunks. This is expressed formally using the three inference 
rules shown below. 
\begin{defi}[Heap refinement]
\begin{mathpar}
\inferrule{
h_\mathrm{c} \triangleleft h\\
\mathbf{predicate}\ p(\overline{x}) = a\\
(\mathbf{0}[\overline{x}:=\overline{v}], h) \xrightarrow{a}_\mathsf{c} (s', \mathbf{0})
}{
h_\mathrm{c} \triangleleft \llbrace p(\overline{v})\rrbrace
}
\and
\inferrule{
\ %
}{
h_\mathrm{c} \triangleleft h_\mathrm{c}
}
\and
\inferrule{
h_\mathrm{c} \triangleleft h\\
h_\mathrm{c}' \triangleleft h'
}{
h_\mathrm{c} \uplus h_\mathrm{c}' \triangleleft h \uplus h'
}
\end{mathpar}
\end{defi}
The first rule states that if a concrete heap 
$h_\mathrm{c}$ refines a heap $h$ that satisfies the body $a$ of some 
predicate $p$, with no chunks left, when consumed under a store that binds 
the predicate parameters $\overline{x}$ to some argument list 
$\overline{v}$, then it refines the singleton heap containing just the 
chunk $p(\overline{v})$. The second rule states that any heap refines 
itself. The third rule states that heap refinement is compatible with heap 
union.

Notice that there are typically many concrete heaps that refine a given 
semiconcrete heap. Consider for example the semiconcrete heap $\llbrace 
\mathsf{list}(50)\rrbrace$, where predicate $\mathsf{list}$ is defined as 
in the example earlier. Any concrete heap that contains exactly a linked 
list starting at address 50 refines this semiconcrete heap. There are 
infinitely many such concrete heaps, corresponding to different list 
lengths, different addresses of nodes, and different values stored in the 
nodes.

The following property of heap refinement allows us to fold and unfold predicate definitions:
\begin{lem}[Open, Close]
$$h_\mathrm{c} \triangleleft h \uplus \llbrace p(\overline{v})\rrbrace
\Leftrightarrow \exists s', h'.\; (\mathbf{0}[\overline{x}:=\overline{v}], h') \xrightarrow{a}_\mathsf{c} (s', \mathbf{0}) \land
h_\mathrm{c} \triangleleft h \uplus h'$$
where $\mathbf{predicate}\ p(\overline{x}) = a$
\end{lem}

\subsubsection{Soundness of Semiconcrete Execution of Commands}

Given the refinement relation, we can define a \emph{refinement mutator} 
$\kappa$ that takes a semiconcrete state as input and outputs a 
demonically chosen concrete state such that the output heap refines the 
input heap. 

\begin{defi}{Refinement Mutator}
$$\kappa = \lambda (s, h).\;\bigotimes h_\mathrm{c} \in \mathit{CHeaps}, h_\mathrm{c} \triangleleft h.\;\langle(s, h_\mathrm{c})\rangle$$
\end{defi}

Given the refinement mutator, we can state the main lemma for the 
soundness of semiconcrete execution.
\begin{lem}[Soundness of Semiconcrete Execution of Commands]
If $\forall r.\;\mathsf{valid}(r)$, then
$$\mathsf{scexec}(c); \kappa \Rrightarrow \kappa; \mathsf{exec}(c)$$
\end{lem}
\begin{proof}
It is sufficient to prove
$$\forall n, c.\;\mathsf{scexec}(c); \kappa \Rrightarrow \kappa; \mathsf{exec}_n(c)$$
By induction on $n$. The base case is trivial. Assume $\forall c.\;\mathsf{scexec}(c); \kappa \Rrightarrow \kappa; \mathsf{exec}_n(c)$. The goal is
$\forall c.\;\mathsf{scexec}(c); \kappa \Rrightarrow \kappa; \mathsf{exec}_{n+1}(c)$. By case analysis on $c$.
\end{proof}
It roughly states that, assuming that 
all routines are valid, executing a command semiconcretely starting from 
some semiconcrete state is worse than executing it concretely starting 
from a demonically chosen corresponding concrete state.

The lemma can be proven by induction on the depth of concrete execution 
and a case analysis on the command. Most cases are trivial; the nontrivial 
cases are routine calls, while loops, and $\mathbf{open}$ and 
$\mathbf{close}$ commands. The proofs of the latter cases use the 
properties of consumption and production.

Below, we sketch the proof in some more detail for the cases of routine 
calls and while loops.

\begin{proof}[Proof (Routine Calls)]
Assume routine definition $$\mathbf{routine}\ r(\overline{x})\ \mathbf{req}\ a\ \mathbf{ens}\ a' = c$$
The goal is $\mathsf{scexec}(r(\overline{e}));\kappa \Rrightarrow \kappa; \mathsf{cexec}_{n+1}(r(\overline{e}))$.
This expands to
\begin{multline*}
\overline{v} \leftarrow \mathsf{eval}(\overline{e});
\mathsf{with}(\mathbf{0}[\overline{x}:=\overline{v}], \mathsf{consume}(a); \mathsf{produce}(a'));
\kappa\\
\Rrightarrow
\kappa;
\overline{v} \leftarrow \mathsf{eval}(\overline{e});
\mathsf{with}(\mathbf{0}[\overline{x}:=\overline{e}], \mathsf{cexec}_n(c))
\end{multline*}
We have $\mathsf{eval}(e){;}, \kappa \Rrightarrow \kappa; \mathsf{eval}(e)$.
Furthermore, we have monotonicity of sequential composition of mutators with respect to coverage.
Therefore, it is sufficient to fix values $\overline{v}$ and prove
$$\mathsf{with}(\mathbf{0}[\overline{x}:=\overline{v}], \mathsf{consume}(a); \mathsf{produce}(a'));
\kappa
\Rrightarrow
\kappa;
\mathsf{with}(\mathbf{0}[\overline{x}:=\overline{v}], \mathsf{cexec}_n(c))$$
Let $s = \mathbf{0}[\overline{x}:=\overline{v}]$.
Furthermore, we abbreviate $\mathsf{with}$, $\mathsf{consume}$, $\mathsf{produce}$,
$\mathsf{scexec}$, and $\mathsf{exec}$ as $\mathsf{w}$, $\mathsf{c}$, $\mathsf{p}$, $\mathsf{sce}$, and $\mathsf{e}$, respectively.
The goal then becomes $$\mathsf{w}(s, \mathsf{c}(a); \mathsf{p}(a')); \kappa \Rrightarrow \kappa; \mathsf{w}(s, \mathsf{e}_n(c))$$

By the induction hypothesis, we have $\mathsf{sce}(c); \kappa \Rrightarrow \kappa; \mathsf{e}_n(c)$ and therefore
$\mathsf{w}(s, \mathsf{sce}(c)); \kappa \Rrightarrow \mathsf{w}(s, \mathsf{e}_n(c))$. By transitivity and monotonicity of coverage, it is sufficient
to prove\\
$\mathsf{w}(s, \mathsf{c}(a); \mathsf{p}(a')) \Rrightarrow \mathsf{w}(s, \mathsf{sce}(c))$.

By validity of $r$, we have
\begin{multline*}
(\mathbf{0}, \mathbf{0}) \triangleright
\bigotimes \overline{v}.\;
\mathsf{w}(\mathbf{0}[\overline{x}:=\overline{v}],
s' \leftarrow \mathsf{w}(\mathbf{0}[\overline{x}:=\overline{v}], \mathsf{p}(a); \mathsf{store});\\
\mathsf{sce}(c);
\mathsf{w}(s', \mathsf{c}(a'))
);
\mathsf{leakcheck}
\;\{\mathrm{true}\}
\end{multline*}
We abbreviate $\mathsf{w}(s, C; \mathsf{store})$ by $\mathsf{ws}(s, C)$.
Furthermore, we instantiate the demonic choice using our fixed $\overline{v}$, and we use $s = \mathbf{0}[\overline{x}:=\overline{v}]$, obtaining
$$(\mathbf{0}, \mathbf{0}) \triangleright
\mathsf{w}(s,
s' \leftarrow \mathsf{ws}(s, \mathsf{p}(a));
\mathsf{sce}(c);
\mathsf{w}(s', \mathsf{c}(a'))
);
\mathsf{leakcheck}
\;\{\mathrm{true}\}$$
It is easy to see that it follows that for any store $s_0$ and heap $h_0$, we have
$(s_0, \mathbf{0}) \triangleright
\mathsf{w}(s,
s' \leftarrow \mathsf{ws}(s, \mathsf{p}(a));
\mathsf{sce}(c);
\mathsf{w}(s', \mathsf{c}(a'))
);
\mathsf{produce}(h_0)
\;\{s_1, h_1.\;s_1 = s_0 \land h_1 = h_0\}$.
By locality of assertion consumption, semiconcrete execution, and assertion production, we can shift $\mathsf{produce}(h_0)$ to the front,
obtaining $$(s_0, h_0) \triangleright
\mathsf{w}(s,
s' \leftarrow \mathsf{ws}(s, \mathsf{p}(a));
\mathsf{sce}(c);
\mathsf{w}(s', \mathsf{c}(a'))
)\;\{s_1, h_1.\;s_1 = s_0 \land h_1 = h_0\}$$
Hence,
$\mathsf{noop} \Rrightarrow 
\mathsf{w}(s,
s' \leftarrow \mathsf{ws}(s, \mathsf{p}(a));
\mathsf{sce}(c);
\mathsf{w}(s', \mathsf{c}(a'))$.

The goal now follows by simple rewriting, using the rewriting lemmas seen above for consumption followed by production:
$$\begin{array}{@{} l @{\ } l @{}}
& \mathsf{w}(s, \mathsf{c}(a); \mathsf{p}(a'))\\
\Rrightarrow & s_1{\leftarrow}\mathsf{ws}(s, \mathsf{c}(a)); \mathsf{w}(s_1, \mathsf{p}(a'))\\
\Rrightarrow & s_1 {\leftarrow} \mathsf{ws}(s, \mathsf{c}(a)); \mathsf{noop}; \mathsf{w}(s_1, \mathsf{p}(a'))\\
\Rrightarrow & s_1 {\leftarrow} \mathsf{ws}(s, \mathsf{c}(a)); \mathsf{w}(s,
  s' {\leftarrow} \mathsf{ws}(s, \mathsf{p}(a));
  \mathsf{sce}(c);
  \mathsf{w}(s', \mathsf{c}(a'))
); \mathsf{w}(s_1, \mathsf{p}(a'))\\
\Rrightarrow &
s_1 {\leftarrow} \mathsf{ws}(s, \mathsf{c}(a));
s_2 {\leftarrow} \mathsf{ws}(s, \mathsf{p}(a));
\mathsf{w}(s, \mathsf{sce}(c));
\mathsf{w}(s_2, \mathsf{c}(a'));
\mathsf{w}(s_1, \mathsf{p}(a'))\\
\Rrightarrow & \mathsf{w}(s, \mathsf{sce}(c));
\mathsf{w}(s'', \mathsf{c}(a'));
\mathsf{w}(s'', \mathsf{p}(a'))\\
\Rrightarrow & \mathsf{w}(s, \mathsf{sce}(c))
\end{array}\vspace{-\baselineskip}$$
\end{proof}

\begin{proof}[Proof (Loops)]
The goal is
$$\mathsf{sce}(\mathbf{while}\ b\ \mathbf{inv}\ a\ \mathbf{do}\ c); \kappa \Rrightarrow \kappa; \mathsf{e}_{n+1}(\mathbf{while}\ b\ \mathbf{inv}\ a\ \mathbf{do}\ c)$$
Expanding the definitions, and further abbreviating
$\mathsf{modified}$, $\mathsf{havoc}$, $\mathsf{assume}$, $\mathsf{leakcheck}$, $\mathsf{heap} := \mathbf{0}$, $\mathsf{targets}(c)$,
$s \leftarrow \mathsf{store}; \mathsf{with}(s, \mathsf{consume}(a))$, and
$s \leftarrow \mathsf{store};$ $\mathsf{with}(s, \mathsf{produce}(a))$ as
$\mathsf{m}$, $\mathsf{h}$, $\mathsf{a}$, $\mathsf{lck}$, $\mathsf{clh}$, $\overline{x}$, $\mathsf{cc}(a)$, and $\mathsf{pc}(a)$,
our goal reduces to
$$
\mathsf{cc}(a); \mathsf{h}(\overline{x});
(\mathsf{clh}; \mathsf{pc}(a); \mathsf{a}(b); \mathsf{sce}(c); \mathsf{cc}(a); \mathsf{lck}
\otimes
\mathsf{pc}(a); \mathsf{a}(\lnot b));
\kappa
\Rrightarrow
\kappa;
(\mathsf{a}(b); \mathsf{e}_n(c))^*;\mathsf{a}(\lnot b)
$$
Using the property $(\forall s.\;\mathsf{m}_{\overline{x}}(s); C \Rrightarrow C') \Rightarrow C \Rrightarrow C'$, and fixing $s$,
it is sufficient to prove
\begin{multline*}
\mathsf{m}_{\overline{x}}(s); \mathsf{cc}(a); \mathsf{h}(\overline{x});
(\mathsf{clh}; \mathsf{pc}(a); \mathsf{a}(b); \mathsf{sce}(c); \mathsf{cc}(a); \mathsf{lck}
\otimes
\mathsf{pc}(a); \mathsf{a}(\lnot b));
\kappa\\
\Rrightarrow
\kappa;
(\mathsf{a}(b); \mathsf{e}_n(c))^*;\mathsf{a}(\lnot b)
\end{multline*}\smallskip

\noindent We now prove the following lemma.
\begin{lem}
Assume $\mathsf{local}\;C$ and $\mathsf{modifies}_{\overline{x}}\;C$. We have
$$\mathsf{m}_{\overline{x}}(s); \mathsf{h}(\overline{x}); \mathsf{clh}; C; \mathsf{lck} \Rrightarrow \bot
\quad\lor\quad
\mathsf{m}_{\overline{x}}(s); \mathsf{h}(\overline{x}) \Rrightarrow C$$
\end{lem}
\begin{proof}
We assume the left-hand disjunct is false and we prove the right-hand disjunct.
From this assumption it follows that there exists an initial state $(s_0, h_0)$ such that
$$(s_0, h_0) \triangleright \mathsf{m}_{\overline{x}}(s); \mathsf{h}(\overline{x}); \mathsf{clh}; C; \mathsf{lck}\;\{\mathrm{true}\}$$
It follows that $s_0 \stackrel{\overline{x}}{\sim} s$ and for any $s_1 \stackrel{\overline{x}}{\sim} s_0$ we have
$(s_1, \mathbf{0}) \triangleright C\;\{s', h'.\;h' = \mathbf{0}\}$.
Hence, by $\mathsf{modifies}_{\overline{x}}\;C$, we have
$(s_1, \mathbf{0}) \triangleright C\;\{s', h'.\;s' \stackrel{\overline{x}}{\sim} s_1 \land h' = \mathbf{0}\}$.
Hence, by $\mathsf{local}\;C$, we have, for any $h_1$,
$(s_1, h_1) \triangleright C\;\{s', h'.\;s' \stackrel{\overline{x}}{\sim} s_1 \land h' = h_1\}$.
From this our goal follows.
\end{proof}

We have $\mathsf{modifies}_{\overline{x}}\;\mathsf{pc}(a); \mathsf{a}(b); \mathsf{sce}(c); \mathsf{cc}(a)$ and
$\mathsf{local}\;\mathsf{pc}(a); \mathsf{a}(b); \mathsf{sce}(c); \mathsf{cc}(a)$; applying the lemma, we obtain
\begin{multline*}
\mathsf{m}_{\overline{x}}(s); \mathsf{h}(\overline{x}); \mathsf{clh}; \mathsf{pc}(a); \mathsf{a}(b); \mathsf{sce}(c); \mathsf{cc}(a); \mathsf{lck} \Rrightarrow \bot\\
\lor\;\mathsf{m}_{\overline{x}}(s); \mathsf{h}(\overline{x}) \Rrightarrow \mathsf{pc}(a); \mathsf{a}(b); \mathsf{sce}(c); \mathsf{cc}(a)
\end{multline*}
We consider both cases. In the first case, the goal follows trivially. In the remainder of the proof, we assume the second case.

Using the property $C_2 \Rrightarrow C_3 \Rightarrow C_1 \otimes C_2 \Rrightarrow C_3$,
we drop the left-hand side of the demonic choice in our goal. Our goal becomes
$$\mathsf{m}_{\overline{x}}(s); \mathsf{cc}(a); \mathsf{h}(\overline{x}); \mathsf{pc}(a); \mathsf{a}(\lnot b);
\kappa
\Rrightarrow
\kappa;
(\mathsf{a}(b); \mathsf{e}_n(c))^*;\mathsf{a}(\lnot b)$$
Applying the induction hypothesis, we have
$(\mathsf{a}(b); \mathsf{sce}(c))^*; \mathsf{a}(\lnot b); \kappa \Rrightarrow \kappa; (\mathsf{a}(b);$ $\mathsf{e}_n(c))^*;\mathsf{a}(\lnot b)$.
By transitivity of coverage and monotonicity of mutator sequential composition with respect to coverage, it is sufficient to prove
$$\mathsf{m}_{\overline{x}}(s); \mathsf{cc}(a); \mathsf{h}(\overline{x}); \mathsf{pc}(a) \Rrightarrow (\mathsf{a}(b); \mathsf{sce}(c))^*$$

Note that to prove $C' \Rrightarrow C^*$, it is sufficient to prove $C' \Rrightarrow \mathsf{noop}$ and $C' \Rrightarrow C; C'$. Applying this
rule to the goal, the first subgoal is easy to prove (using the properties of consumption followed by production). Our remaining goal is
$$\mathsf{m}_{\overline{x}}(s); \mathsf{cc}(a); \mathsf{h}(\overline{x}); \mathsf{pc}(a)
\Rrightarrow \mathsf{a}(b); \mathsf{sce}(c); \mathsf{m}_{\overline{x}}(s); \mathsf{cc}(a); \mathsf{h}(\overline{x}); \mathsf{pc}(a)$$

The goal now follows by simple rewriting, using the rewriting lemmas seen above for consumption followed by production,
as well as the properties $\mathsf{m}_{\overline{x}}(s) \Rrightarrow \mathsf{m}_{\overline{x}}(s); \mathsf{m}_{\overline{x}}(s)$,
$\mathsf{h}(\overline{x}) \Rrightarrow \mathsf{h}(\overline{x}); \mathsf{h}(\overline{x})$, and
$\mathsf{modifies}_{\overline{x}}\;\mathsf{sce}(c)$:
$$\begin{array}{@{} l l @{}}
& \mathsf{m}_{\overline{x}}(s); \mathsf{cc}(a); \mathsf{h}(\overline{x}); \mathsf{pc}(a)\\
\Rrightarrow & \mathsf{m}_{\overline{x}}(s); \mathsf{cc}(a); \mathsf{m}_{\overline{x}}(s_0); \mathsf{h}(\overline{x}); \mathsf{h}(\overline{x}); \mathsf{pc}(a)\\
\Rrightarrow & \mathsf{m}_{\overline{x}}(s); \mathsf{cc}(a); \mathsf{pc}(a); \mathsf{a}(b); \mathsf{sce}(c); \mathsf{cc}(a); \mathsf{h}(\overline{x}); \mathsf{pc}(a)\\
\Rrightarrow & \mathsf{m}_{\overline{x}}(s); \mathsf{a}(b); \mathsf{sce}(c); \mathsf{cc}(a); \mathsf{h}(\overline{x}); \mathsf{pc}(a)\\
\Rrightarrow & \mathsf{a}(b); \mathsf{sce}(c); \mathsf{m}_{\overline{x}}(s); \mathsf{cc}(a); \mathsf{h}(\overline{x}); \mathsf{pc}(a)
\end{array}\vspace{-\baselineskip}$$
\end{proof}

\subsection{Soundness of Semiconcrete Execution}

\begin{thm}[Soundness of Semiconcrete Execution]
$$\mathsf{sc}\textsf{\emph{-}}\mathsf{safe\_program}(c) \Rightarrow \mathsf{safe\_program}(c)$$
\end{thm}

The soundness of semiconcrete execution follows directly from the 
soundness of semiconcrete execution of commands.
Therefore, we are now halfway on our way towards a 
formalization and soundness proof of Featherweight VeriFast. 
Semiconcrete execution is not suitable as a verification 
algorithm since it performs infinite branching. In the next 
section, we formalize and sketch a soundness proof of 
Featherweight VeriFast's symbolic execution algorithm, which 
builds on semiconcrete execution but introduces \emph{symbols} 
to eliminate infinite branching.

\section{Symbolic Execution}\label{sec:symexec}

In this section, we introduce symbolic execution by example, and then provide formal definitions. Finally, we sketch a soundness proof.

\subsection{Symbolic Execution: Example Trace}

Recall the example semiconcrete execution trace for 
the example routine $\mathsf{range}$ in Figure~\ref{fig:scexec-trace}.
Notice that while the length 
of this trace is linear in the size of the body of routine 
$\mathsf{range}$, there are infinitely many such traces, since 
each number shown in orange is picked by demonic choice among 
all integers (potentially with some constraints).

We introduce \emph{symbolic execution} to arrive at an 
execution with a finite number of traces of limited length. 
Instead of demonically choosing among an infinite set of 
integers, symbolic execution uses a fresh \emph{symbol} to 
represent an arbitrary number. Symbolic execution states are 
like semiconcrete execution states, except that a \emph{term} 
may be used instead of a literal value in the store and the 
heap. A term is either a literal number, a symbol, or an 
operation (addition or subtraction) applied to two terms. In 
addition to replacing numbers by terms, symbolic execution adds 
a third component to the state: the \emph{path condition}. This 
is a set of \emph{formulae} that define the set of \emph{relevant
interpretations} of the symbols used in the store and the heap. 
A formula is either an equality between terms ($t = t'$), an 
inequality between terms ($t < t'$), or the negation of another 
formula.

\begin{figure}
$$\begin{array}{l}
\mathbf{routine}\ \mathsf{range}(\mathsf{i}, \mathsf{n}, \mathsf{r})\\
\quad \annot{\mathbf{req}\ \mathsf{r} \mapsto {?}\mathsf{dummy}\ \mathbf{ens}\ \mathsf{r} \mapsto {?}\mathsf{list} * \mathsf{list}(\mathsf{list})}\\
\comment{\Phi{:}\{\branching{\symi}{,}\branching{\symn}{,}\branching{\symr}\}, s{:} \mathbf{0}[\mathsf{i}{:}\branching{\symi},\mathsf{n}{:}\branching{\symn},\mathsf{r}{:}\branching{\symr}], h {:} \mathbf{0}}\quad\quad\quad\Phi{:}\{\dots,\varsigma,\dots\} = \Phi{:}\{\dots,\varsigma = \varsigma, \dots\}\\
{\color{Gray}\mathsf{sproduce}(}\annot{\mathsf{r} \mapsto {?}\mathsf{dummy}}{\color{Gray})}\\
\comment{\Phi{:}\{\symi{,}\symn{,}\symr{,}\branching{\symd}\}, s{:} \mathbf{0}[\mathsf{i}{:}\symi,\mathsf{n}{:}\symn,\mathsf{r}{:}\symr], h {:} \llbrace \symr {\mapsto} \branching{\symd}\rrbrace}\\
\mathbf{if}\ \mathsf{i} = \mathsf{n}\ \mathbf{then}\ \mathsf{l} := 0\ \mathbf{else}\ (\\
\comment{\Phi{:}\{\symi{,}\symn{,}\symr{,}\symd, \symi{\neq}\symn\}, s{:} \mathbf{0}[\mathsf{i}{:}\symi,\mathsf{n}{:}\symn,\mathsf{r}{:}\symr], h{:} \llbrace \symr {\mapsto}\symd\rrbrace}\\{}
\mathsf{l} := \mathbf{malloc}(2);\\
\comment{\Phi{:}\{\symi{,}\symn{,}\symr{,}\symd{,}\branching{\syml}{,}\branching{\symv}{,}\branching{\symv'}, \symi{\neq}\symn{,}0{<}\syml\}, s{:} \mathbf{0}[\mathsf{i}{:}\symi,\mathsf{n}{:}\symn,\mathsf{r}{:}\symr,\mathsf{l}{:}\branching{\syml}], h{:} \llbrace \symr {\mapsto}\symd{,} \mathsf{mb}(\branching{\syml}, 2){,} \branching{\syml}{\mapsto}\branching{\symv}{,} \branching{\syml}{+}1 {\mapsto} \branching{\symv'}\rrbrace}\\{}
[\mathsf{l}] := \mathsf{i}; \mathsf{range}(\mathsf{i} + 1, \mathsf{n}, \mathsf{l} + 1)\\
{\color{Gray}\mathsf{sconsume}(}\annot{\mathsf{l} {+} 1 {\mapsto} {?}\mathsf{dummy}}{\color{Gray});\mathsf{sproduce}(}\annot{\mathsf{l}{+}1 {\mapsto} {?}\mathsf{list} * \mathsf{list}(\mathsf{list})}{\color{Gray})}\\
\comment{\begin{array}{@{} l @{}}
\Phi{:}\{\symi{,}\symn{,}\symr{,}\symd{,}\syml{,}\symv{,}\symv'{,}\branching{\syml'}, \symi{\neq}\symn{,}0{<}\syml\}, s{:}\mathbf{0}[\mathsf{i}{:}\symi,\mathsf{n}{:}\symn,\mathsf{r}{:}\symr,\mathsf{l}{:}\syml],\\
h{:} \llbrace \symr {\mapsto} \symd{,} \mathsf{mb}(\syml, 2){,} \syml{\mapsto}\symi{,} \syml{+}1 {\mapsto} \branching{\syml'}{,}\mathsf{list}(\branching{\syml'})\rrbrace
\end{array}}\\
);\ \annot{\mathbf{close}\ \mathsf{list}(\mathsf{l})}; [\mathsf{r}] := \mathsf{l}\\
\comment{\Phi{:}\{\symi{,}\symn{,}\symr{,}\symd{,}\syml{,}\symv{,}\symv'{,}\syml', \symi{\neq}\symn{,}0{<}\syml\}, s{:}\mathbf{0}[\mathsf{i}{:}\symi,\mathsf{n}{:}\symn,\mathsf{r}{:}\symr,\mathsf{l}{:}\syml],h{:}\llbrace \symr{\mapsto} \syml{,}\mathsf{list}(\syml)\rrbrace}\\
{\color{Gray}\mathsf{sconsume}(}\annot{\mathsf{r} \mapsto {?}\mathsf{list} * \mathsf{list}(\mathsf{list})}{\color{Gray})}\\
\comment{\Phi{:}\{\symi{,}\symn{,}\symr{,}\symd{,}\syml{,}\symv{,}\symv'{,}\syml', \symi{\neq}\symn{,}0{<}\syml\}, s{:}\mathbf{0}[\mathsf{i}{:}\symi,\mathsf{n}{:}\symn,\mathsf{r}{:}\symr,\mathsf{l}{:}\syml],h{:}\mathbf{0}}
\end{array}$$
\caption{Symbolic Execution: Example Trace}\label{fig:symexec-trace}
\end{figure}

\begin{exa}{Symbolic Execution: Example Trace}
See Figure~\ref{fig:symexec-trace}
\end{exa}

In Figure~\ref{fig:symexec-trace} we show the symbolic execution trace for routine 
$\mathsf{range}$ corresponding to the semiconcrete execution 
trace shown before. Note: do not confuse the program variables 
and the symbols. The former are shown in an upright font; the 
latter are shown in a slanted font. In the symbolic execution 
trace, the letters shown in orange do not denote branching 
(i.e.~demonic choices); rather, they show freshly picked 
symbols.

Besides the use of symbols, notice the path condition $\Phi$: 
it starts out empty; in the $\mathbf{else}$ branch of the 
$\mathbf{if}$ statement, the formula $\symi \neq \symn$ is added; and the 
$\mathbf{malloc}$ statement adds the formula $0 < \syml$. 

\subsection{Symbolic Execution: Types}

The set $\mathsf{SStates}$ of symbolic execution states is defined below. Terms are like expressions, except that they may mention symbols, which represent a fixed value,
instead of program variables, whose value may change through assignments. Similarly, formulae correspond to boolean expressions.

Symbolic states are like semiconcrete states, except that terms are used instead of values in the store and as chunk arguments; furthermore, the state includes an
additional component, called the path condition, which is a set of formulae.

\begin{defi}{Symbolic Execution: Types}
$$\begin{array}{r l}
& \varsigma \in \mathit{Symbols}\\
t, \sell, \sv \in \mathit{Terms} & ::= z\ |\ \varsigma\ |\ t + t\ |\ t - t\\
\varphi \in \mathit{Formulae} & ::= t = t\ |\ t < t\ |\ \lnot \varphi\\
\\
\ss \in \mathit{SStores} = & \mathit{Vars} \rightarrow \mathit{Terms}\\
\mathit{SPredicates} = & \{\mapsto, \mathsf{mb}\} \cup \mathit{UserDefinedPredicates}\\
\mathit{SChunks} = & \{p(\overline{\sv})\ |\ p \in \mathit{SPredicates}, \overline{\sv} \in \mathit{Terms}\}\\
\sh \in \mathit{SHeaps} = & \mathit{SChunks} \rightarrow \mathbb{N}\\
\mathit{PathConditions} = & \mathcal{P}(\mathit{Formulae})\\
\mathit{SStates} = & \mathit{PathConditions} \times \mathit{SStores} \times \mathit{SHeaps}\\
\mathit{SMutators} = & \mathit{SStates} \rightarrow \mathit{Outcomes}(\mathit{SStates})\\
\\
\mathit{sconsume}(a) \in & \mathit{Assertions} \rightarrow \mathit{SMutators}\\
\mathit{sproduce}(a) \in & \mathit{Assertions} \rightarrow \mathit{SMutators}\\
\mathit{symexec}(c) \in & \mathit{Commands} \rightarrow \mathit{SMutators}\\
\end{array}$$
\end{defi}

\subsection{Symbolic Execution: Auxiliary Definitions}

As we did for concrete execution and semiconcrete execution, we introduce 
a few auxiliary definitions for use in the definition of symbolic 
execution. They are as follows. 

In concrete and semiconcrete execution, assuming a boolean expression 
evaluates the expression in the current store and blocks if it evaluates 
to false. In symbolic execution, this is not possible, since evaluation of 
a boolean expression under a symbolic store yields a formula rather than a 
boolean value. Symbolic execution, therefore, asks an \emph{SMT solver}, a 
type of automatic theorem prover, to try to prove that the formula is 
inconsistent with the path condition. If it succeeds, symbolic execution 
blocks. Otherwise, the formula is added to the path condition, in order to 
record that on the remainder of the current symbolic execution path, of 
all possible interpretations of the symbols used in the symbolic state, 
only the ones that satisfy the formula are relevant.

We write $\Phi \vdash_\mathrm{SMT} \varphi$ to denote that the SMT solver 
succeeds in proving that the set of formulae $\Phi$ implies the formula 
$\varphi$.

Similarly, asserting a boolean expression in symbolic execution means 
evaluating it to a formula under the current symbolic store and asking the 
SMT solver to try to prove that the formula follows from the path 
condition. If it succeeds, execution proceeds normally; otherwise, it 
fails.

The set $\mathsf{Used}(\Phi)$ denotes the set of symbols $\varsigma$ for 
which a formula $\varsigma = \varsigma$ appears in the path condition 
$\Phi$. In a well-formed symbolic state, all symbols used in the symbolic 
state are in this set. 

$\mathsf{fresh}(\Phi)$ denotes some symbol that is not in 
$\mathsf{Used}(\Phi)$. It is defined using a \emph{choice function} 
$\epsilon$, which maps each nonempty set to some element of that set.

The mutator $\mathsf{fresh}$ picks some symbol $\varsigma$ that is not yet 
used by the current symbolic state, records that it is now being used by 
adding a formula $\varsigma = \varsigma$ to the path condition, and yields 
the symbol as its answer.

We define the notation $\bigoplus t.\;C(t)$, where $C$ is a mutator 
parameterized by a term, to denote angelic choice over all terms that only 
use symbols that are already being used by the current symbolic state. 
$\mathrm{FS}(t)$ denotes the set of free symbols that appear in term $t$, 
i.e.~the set of symbols used by $t$.\footnote{Since the syntax of terms 
does not include any binding constructs, all symbols that appear in a term 
are free symbols of the term.}

Symbolic consumption $\mathsf{sconsume\_chunks}(\sh)$ of a multiset $\sh$ of 
symbolic terms differs from concrete and semiconcrete consumption in that 
it does not simply look for the exact chunks $\sh$ in the current heap; 
rather, it looks for chunks for which the SMT solver succeeds in proving 
that their argument terms are equal under all relevant interpretations of 
the symbols. For example, suppose the heap contains a chunk 
$\mathsf{list}(\syml)$ and the path condition contains a formula $\syml = 
\syml'$; then consumption of a chunk $\mathsf{list}(\syml')$ succeeds, 
even though the exact chunk $\mathsf{list}(\syml')$ does not appear in the 
symbolic heap. Symbolic production is simpler; as in semiconcrete 
execution, it simply adds the specified chunks to the heap. Symbolic consumption and production of a single symbolic chunk $\hat\alpha$ are defined in the obvious way.

\begin{defi}{Symbolic Execution: Auxiliary Definitions}

$$\begin{array}{l}
\mathsf{sassume}(\varphi) = \lambda (\Phi, \ss, \sh).\; \bigotimes \Phi \not\vdash_\mathrm{SMT} \lnot \varphi.\;\langle (\Phi \cup \{\varphi\}, \ss, \sh)\rangle\\
\mathsf{sassume}(b) = \ss \leftarrow \mathsf{sstore}; \mathsf{sassume}(\llbracket b\rrbracket_{\ss})\\
\mathsf{sassert}(b) = \lambda (\Phi, \ss, \sh).\ \bigoplus \Phi \vdash_\mathrm{SMT} \llbracket b\rrbracket_{\ss}.\;\langle(\Phi, \ss, \sh)\rangle
\\
\mathsf{Used}(\Phi) = \{\varsigma \in \mathit{Symbols}\ |\ (\varsigma = \varsigma) \in \Phi\}\\
\mathsf{fresh}(\Phi) = \epsilon(\{\varsigma \in \mathit{Symbols}\ |\ \varsigma \notin \mathsf{Used}(\Phi)\})\\
\mathsf{fresh} =
  \lambda (\Phi, \ss, \sh).\
  \mathbf{let}\ \varsigma = \mathsf{fresh}(\Phi)\ \mathbf{in}\ 
  \langle (\Phi \cup \{\varsigma = \varsigma\}, \ss, \sh), \varsigma \rangle\\
\\
\bigoplus t.\;C(t) = \Phi \leftarrow \mathsf{pc}; \bigoplus t \in \mathit{Terms}, \mathrm{FS}(t) \subseteq \mathsf{Used}(\Phi).\;C(t)\\
\mathsf{sconsume\_chunks}(\sh') = \lambda (\Phi, \ss, \sh).\;\bigotimes \sh'' \le \sh, \Phi \vdash_\mathrm{SMT} \sh'' = \sh'.\;\langle (\Phi, \ss, \sh - \sh'')\rangle\\
\mathsf{sconsume\_chunk}(\hat\alpha) = \mathsf{sconsume\_chunks}(\llbrace\hat\alpha\rrbrace)\\
\mathsf{sproduce\_chunks}(\sh') = \lambda (\Phi, \ss, \sh).\;\langle(\Phi, \ss, \sh \uplus \sh')\rangle\\
\mathsf{sproduce\_chunk}(\hat\alpha) = \mathsf{sproduce\_chunks}(\llbrace\hat\alpha\rrbrace)\\
\\
\textrm{where}\\
\quad \epsilon(X)\quad=\quad\textrm{some element of $X$}
\end{array}$$

\end{defi}

\subsection{Symbolic Execution: Definition}

The definition of symbolic execution is entirely analogous to that of 
semiconcrete execution, except that symbolic versions of the auxiliary 
mutators are used and that each demonic choice over all values is replaced 
by picking a fresh symbol.

\begin{defi}{Producing Assertions}

$$\begin{array}{l}
\mathsf{sproduce}(b) = \mathsf{sassume}(b)\\[.5em]

\mathsf{sproduce}(p(\overline{e}, \overline{{?}x})) =\\
\quad
  \overline{\sv} \leftarrow \mathsf{seval}(e);
  \overline{\sv}' \leftarrow \mathsf{fresh};
  \mathsf{sproduce\_chunk}(p(\overline{\sv}, \overline{\sv}'));
  \overline{x} := \overline{\sv}'\\[.5em]

\mathsf{sproduce}(a * a') = \mathsf{sproduce}(a); \mathsf{sproduce}(a')\\[.5em]

\mathsf{sproduce}(\mathbf{if}\ b\ \mathbf{then}\ a\ \mathbf{else}\ a') =\\
\quad \mathsf{sassume}(b); \mathsf{sproduce}(a) \otimes \mathsf{sassume}(\lnot b); \mathsf{sproduce}(a')
\end{array}$$

\end{defi}

\begin{defi}{Consuming Assertions}

$$\begin{array}{l}
\mathsf{sconsume}(b) = \mathsf{sassert}(b)\\[.5em]

\mathsf{sconsume}(p(\overline{e}, \overline{{?}x})) =\\
\quad
  \overline{\sv} \leftarrow \mathsf{seval}(e);
  \bigoplus \overline{\sv}'.\;
  \mathsf{sconsume\_chunk}(p(\overline{\sv}, \overline{\sv}'));
  \overline{x} := \overline{\sv}'\\[.5em]

\mathsf{sconsume}(a * a') = \mathsf{sconsume}(a); \mathsf{sconsume}(a')\\[.5em]

\mathsf{sconsume}(\mathbf{if}\ b\ \mathbf{then}\ a\ \mathbf{else}\ a') =\\
\quad \mathsf{sassume}(b); \mathsf{sconsume}(a) \otimes \mathsf{sassume}(\lnot b); \mathsf{sconsume}(a')
\end{array}$$

\end{defi}

\begin{figure}
$$\begin{array}{l}
\mathsf{symexec}(x := e) = \sv \leftarrow \mathsf{seval}(e); x := \sv\\
\\
\mathsf{symexec}(c; c') = \mathsf{symexec}(c); \mathsf{symexec}(c')\\
\\
\mathsf{symexec}(\mathbf{if}\ b\ \mathbf{then}\ a\ \mathbf{else}\ a') =\\
\quad \mathsf{sassume}(b); \mathsf{symexec}(c) \otimes \mathsf{sassume}(\lnot b); \mathsf{symexec}(c')\\
\\
\mathsf{symexec}(\mathbf{while}\ e\ \mathbf{inv}\ a\ \mathbf{do}\ c) = \textrm{See Figure~\ref{fig:symexec-while}}\\
\\
\mathsf{symexec}(r(\overline{e})) =\\
\quad \overline{\sv} \leftarrow \mathsf{eval}(\overline{e}); \mathsf{with}(\mathbf{0}[\overline{x}:=\overline{\sv}], \mathsf{sconsume}(a); \mathsf{sproduce}(a'))\\
\quad \textrm{where $\mathbf{routine}\ r(\overline{x})\ \mathbf{req}\ a\ \mathbf{ens}\ a'$}\\
\\
\mathsf{symexec}(x := \mathbf{malloc}(n)) =\\
\quad \sell,\sv_1,\dots,\sv_n \leftarrow \mathsf{fresh}; \mathsf{sassume}(0 < \sell);\\
\quad \mathsf{sproduce\_chunks}(\llbrace\mathsf{mb}(\sell, n),\sell \mapsto \sv_1,\dots,\sell + n - 1\mapsto \sv_n\rrbrace); x := \sell\\
\\
\mathsf{symexec}(x := [e]) =\\
\quad \sell \leftarrow \mathsf{seval}(e); \bigoplus \sv.\;\mathsf{sconsume\_chunk}(\sell \mapsto \sv); \mathsf{sproduce\_chunk}(\sell \mapsto \sv); x := \sv\\
\\
\mathsf{symexec}([e] := e') =\\
\quad \sell, \sv \leftarrow \mathsf{seval}(e, e'); \bigoplus \sv'.\;\mathsf{sconsume\_chunk}(\sell \mapsto \sv'); \mathsf{sproduce\_chunk}(\sell \mapsto \sv)\\
\\
\mathsf{symexec}(\mathbf{free}(e)) = \sell \leftarrow \mathsf{seval}(e);\\
\quad \bigoplus n, \sv_1, \dots, \sv_n.\;\mathsf{sconsume\_chunks}(\llbrace\mathsf{mb}(\sell, n), \sell_1 \mapsto \sv_1, \dots, \sell_n \mapsto \sv_n\rrbrace)\\
\\
\mathsf{symexec}(\mathbf{open}\ p(\overline{e})) = \overline{\sv} \leftarrow \mathsf{eval}(\overline{e});\\
\quad \mathsf{sconsume\_chunk}(p(\overline{\sv})); \mathsf{with}(\mathbf{0}[\overline{x}:=\overline{\sv}], \mathsf{sproduce}(a))\\
\quad\textrm{where $\mathbf{predicate}\ p(\overline{x}) = a$}\\
\\
\mathsf{symexec}(\mathbf{close}\ p(\overline{e})) = \overline{\sv} \leftarrow \mathsf{eval}(\overline{e});\\
\quad \mathsf{with}(\mathbf{0}[\overline{x}:=\overline{\sv}], \mathsf{sconsume}(a)); \mathsf{sproduce\_chunk}(p(\overline{\sv}))\\
\quad\textrm{where $\mathbf{predicate}\ p(\overline{x}) = a$}\\
\end{array}$$
\caption{Symbolic Execution of Commands}\label{fig:symexec}
\end{figure}

\begin{figure}
$$\begin{array}{l}
\mathsf{shavoc}(\overline{x}) = \overline{\sv} \leftarrow \mathsf{fresh}; \overline{x}:=\overline{\sv}\\
\mathsf{sleakcheck} = \lambda (\Phi, \ss, \sh).\; \bigoplus \sh = \mathbf{0}.\;\top\\
\\
\mathsf{symexec}(\mathbf{while}\ b\ \mathbf{inv}\ a\ \mathbf{do}\ c) =\\
\quad \ss \leftarrow \mathsf{sstore}; \mathsf{with}(\ss, \mathsf{sconsume}(a));\\
\quad \mathsf{shavoc}(\mathsf{targets}(c));\\
\quad (\\
\quad\quad \mathsf{sheap} := \mathbf{0};\\
\quad\quad \ss \leftarrow \mathsf{sstore}; \mathsf{with}(\ss, \mathsf{sproduce}(a));\\
\quad\quad \mathsf{sassume}(b); \mathsf{symexec}(c);\\
\quad\quad \ss \leftarrow \mathsf{sstore}; \mathsf{with}(\ss, \mathsf{sconsume}(a));\\
\quad\quad \mathsf{sleakcheck}\\
\quad \otimes\\
\quad\quad \ss \leftarrow \mathsf{sstore}; \mathsf{with}(\ss, \mathsf{sproduce}(a))\\
\quad\quad \mathsf{sassume}(\lnot b);\\
\quad )
\end{array}$$
\caption{Symbolic Execution of Loops}\label{fig:symexec-while}
\end{figure}

\begin{defi}{Symbolic Execution of Commands}
See Figures~\ref{fig:symexec} and \ref{fig:symexec-while}.
\end{defi}

\begin{defi}{Validity of Routines}

$$\begin{array}{l}
\mathsf{svalid}(r) =\\
\quad (\emptyset, \mathbf{0}, \mathbf{0})\; \triangleright\\
\quad \overline{\sv} \leftarrow \mathsf{fresh};\\
\quad \mathsf{with}(\mathbf{0}[\overline{x}:=\overline{\sv}],\\
\quad\quad \ss' \leftarrow \mathsf{with}(\mathbf{0}[\overline{x}:=\overline{\sv}], \mathsf{sproduce}(a); \mathsf{sstore});\\
\quad\quad \mathsf{symexec}(c);\\
\quad\quad \mathsf{with}(\ss', \mathsf{sconsume}(a'))\\
\quad );\\
\quad \mathsf{sleakcheck}\\
\quad \{\mathrm{true}\}\\
\quad \textrm{where $\mathbf{routine}\ r(\overline{x})\ \mathbf{req}\ a\ \mathbf{ens}\ a' = c$}\\
\end{array}$$

\end{defi}

\begin{defi}{Symbolic Execution: Program Safety}

\[
\mathsf{sym}\textsf{-}\mathsf{safe\_program}(c)\quad=\quad (\forall r.\;\mathsf{svalid}(r)) \land (\emptyset, \mathbf{0}, \mathbf{0}) \triangleright \mathsf{symexec}(c)\;\{\mathrm{true}\}
\]

\end{defi}

\subsection{Soundness}
We now argue the soundness of symbolic execution with respect to 
semiconcrete execution, i.e.~that symbolic execution is a safe 
approximation of semiconcrete execution, and therefore if symbolic 
execution does not fail, then semiconcrete execution does not fail. To do 
so, we need to characterize the relationship between symbolic states and 
semiconcrete states. We do so by means of the concept of an 
\emph{interpretation}. 

\begin{defi}{Soundness of symbolic execution: Definitions}

$$\begin{array}{r c l}
I \in \mathit{Interps} & = & \mathit{Symbols} \rightharpoonup \mathbb{Z} = \mathit{Symbols} \rightarrow \mathbb{Z} \cup \{\mathsf{undef}\}\\
\mathrm{dom}\,I & = & \{\varsigma\;|\;I(\varsigma) \neq \mathsf{undef}\}\\
I \subseteq I' & = & \forall \varsigma.\;I(\varsigma) = \mathsf{undef} \lor I(\varsigma) = I'(\varsigma)\\
I((\Phi, \ss, \sh)) & = & \left\{\begin{array}{l l}
(s, h) & \textrm{if $\mathrm{dom}\,I = \mathsf{Used}(\Phi) \land \llbracket \Phi, \ss, \sh\rrbracket_I = \mathsf{true}, s, h$}\\
\mathsf{undef} & \textrm{otherwise}
\end{array}\right.\\
\rho_I & = & \lambda \hat{\sigma}.\;\bigotimes I' \supseteq I, \sigma, I'(\hat{\sigma}) = \sigma.\;\langle \sigma\rangle\\
C \rightsquigarrow_I C' & = & C{;}, \rho_I \Rrightarrow \rho_I; C'\\
C({-}) \rightsquigarrow_I C'({-}) & = & \forall I' \supseteq I, t, v, \llbracket t\rrbracket_{I'} = v.\;C(t) \rightsquigarrow_{I'} C'(v)
\end{array}$$

\end{defi}

\noindent An interpretation is a partial function from symbols to program values. By 
\emph{partial function}, we mean that it maps each symbol either to a 
program value (an integer) or to the special value $\mathsf{undef}$. By 
this, we reflect that at each point during symbolic execution, only some 
of the symbols are in use and the others may be picked by a future 
execution of mutator $\mathsf{fresh}$. 

We say an interpretation $I'$ \emph{extends} another interpretation $I$, 
denoted $I' \supseteq I$, if for each symbol for which $I$ is defined, 
$I'$ is defined and $I'$ maps it to the same value as $I$. 

We define the evaluation $\llbracket {-}\rrbracket_I$ of a term, a 
formula, a path condition, a symbolic store, or a symbolic heap under an 
interpretation $I$ as the partial function that yields $\mathsf{undef}$ if 
the interpretation yields $\mathsf{undef}$ for any of the symbols that 
appear in the input, and the output obtained by replacing all symbols by 
their value otherwise. 

We also use an interpretation as a partial function from symbolic states 
to semiconcrete states, as follows. For an interpretation $I$ and a 
symbolic state $(\Phi, \ss, \sh)$, if the domain of $I$ is exactly 
$\mathsf{Used}(\Phi)$, and $\Phi$ evaluates to true under $I$, and the 
symbolic store and heap $\ss$ and $\sh$ evaluate to a semiconcrete store 
and heap $s$ and $h$ under $I$, then the value of $(\Phi, \ss, \sh)$ under 
$I$ is $(s, h)$, and otherwise it is undefined. Notice that this means 
that the interpretation of a symbolic state is undefined if the symbolic 
state is not well-formed, i.e.~if it uses symbols $\varsigma$ for which no 
formula $\varsigma = \varsigma$ appears in the path condition. 

We now define the \emph{interpretation mutator} $\rho_I$ that, for a given 
symbolic state $\hat{\sigma}$, demonically chooses an extension $I'$ of 
$I$ for which $I'(\hat{\sigma})$ is defined and sets the resulting 
semiconcrete state as the current state. 

Given this mutator, we define the concept of \emph{safe approximation} $C 
\rightsquigarrow_I C'$ of a semiconcrete mutator $C'$ by a symbolic 
mutator $C$ under an interpretation $I$. This holds if $C{;}, \rho_I$ 
covers $\rho_I; C'$. 

We extend this notion to the case of a symbolic operator $C({-})$ 
parameterized by a term and a semiconcrete operator $C'({-})$ 
parameterized by a value. It holds if for any extension $I'$ of $I$, and 
for any term whose value is defined under $I'$, $C(t)$ safely approximates 
$C'(\llbracket t\rrbracket_{I'})$ under $I'$. 

\begin{defi}{Logical Consequence}
$$\Phi \vDash \varphi\quad\Leftrightarrow\quad \forall I.\;\llbracket \Phi\rrbracket_I = \mathsf{true} \Rightarrow \llbracket \varphi\rrbracket_I = \mathsf{true}$$
\end{defi}

\begin{asm}[SMT Solver Soundness]
$$\Phi \vdash_\mathrm{SMT} \varphi\quad\Rightarrow\quad \Phi \vDash \varphi$$
\end{asm}\medskip

\noindent Soundness of symbolic execution relies on one assumption: that the SMT 
solver is sound. That is, if the SMT solver reports success in proving 
that a formula follows from a path condition, then it must be the case 
that this formula does indeed follow from this path condition. We say a 
formula follows from a path condition if all interpretations that satisfy 
the path condition satisfy the formula.

It is not necessary for soundness of symbolic execution that the SMT 
solver be \emph{complete}, i.e.~that it succeed in proving all true 
facts. In fact, symbolic execution is sound even when using an SMT solver 
that does not even try and always reports failure to prove a fact. 
However, in that case symbolic execution itself is highly incomplete, 
i.e.~it fails even if concrete execution does not fail. Indeed, we do not 
claim completeness of Featherweight VeriFast. 

Given these concepts, we can state the soundness lemmas of symbolic 
execution:
\begin{lem}[Soundness]
\display{\begin{array}{r @{\ } c @{\ } l}
C({-}) \rightsquigarrow_I C'({-}) \Rightarrow \sv \leftarrow \mathsf{fresh}; C(\sv) & \rightsquigarrow_I & \bigotimes v.\;C'(v)\\
C({-}) \rightsquigarrow_I C'({-}) \Rightarrow \bigoplus \sv.\;C(\sv) & \rightsquigarrow_I & \bigoplus v.\;C'(v)\\
\mathsf{sassume}(b), \mathsf{sassert}(b) & \rightsquigarrow_I & \mathsf{assume}(b), \mathsf{assert}(b)\\
\llbracket \sh \rrbracket_I = h \Rightarrow \mathsf{sconsume}(\sh), \mathsf{sproduce}(\sh) & \rightsquigarrow_I & \mathsf{consume}(h), \mathsf{produce}(h)\\
\mathsf{sconsume}(a), \mathsf{sproduce}(a) & \rightsquigarrow_I & \mathsf{consume}(a), \mathsf{produce}(a)\\
\mathsf{symexec}(c) & \rightsquigarrow_I & \mathsf{scexec}(c)\\
\mathsf{svalid}(r) & \Rightarrow & \mathsf{valid}(r)\\
\mathsf{sym}\textsf{-}\mathsf{safe\_program}(c) & \Rightarrow & \mathsf{sc}\textsf{-}\mathsf{safe\_program}(c)
\end{array}}
\end{lem}\medskip

\noindent Mutator $\mathsf{fresh}$ safely approximates demonic choice of 
a value; angelic choice of a term that uses only symbols already being 
used by the current symbolic state safely approximates angelic choice of a 
value; symbolic assumption and assertion safely approximate semiconcrete 
assumption and assertion; symbolic consumption and production of heap 
chunks safely approximate semiconcrete consumption and production of their 
interpretations; and symbolic execution safely approximates semiconcrete 
execution. The soundness theorem follows directly. 

Proving the properties stated above is mostly easy; below we go into some 
detail of two of the more interesting proofs: soundness of 
$\mathsf{fresh}$ and soundness of $\mathsf{sassume}$. 

\begin{lem}[Soundness of $\mathsf{fresh}$]
\display{C({-}) \rightsquigarrow_I C'({-}) \Rightarrow \sv \leftarrow \mathsf{fresh}; C(\sv) \rightsquigarrow_I \bigotimes v.\;C'(v)}
\end{lem}
\begin{proof}
We assume the premise and we 
unfold the definition of safe approximation, of mutator coverage, and of 
outcome coverage. Fix an input symbolic state $(\Phi, \ss, \sh)$ and a 
postcondition $Q$. Unfold the definition of $\mathsf{fresh}$. Let 
$\varsigma$ be the fresh symbol. Assume $(\Phi \cup \{\varsigma = 
\varsigma\}, \ss, \sh) \triangleright C(\varsigma); \rho_I\;\{Q\}$. It is 
sufficient to prove $(\Phi, \ss, \sh) \triangleright \rho_I; \bigotimes 
v.\;C'(v)\;\{Q\}$. Unfolding the definition of $\rho_I$ in the goal, fix 
an interpretation $I' \supseteq I$ and a semiconcrete state $(s, h)$ such 
that $I'((\Phi, \ss, \sh)) = (s, h)$. Further fix a value $v$ picked by 
the demonic choice in the goal. It is sufficient to prove that $(s, h) 
\triangleright C'(v)\;\{Q\}$. We build a new interpretation $I''$ by 
binding the fresh symbol $\varsigma$ to value $v$: $I'' = 
I'[\varsigma:=v]$. It follows that $I''((\Phi \cup \{\varsigma = 
\varsigma), \ss, \sh)) = (s, h)$. Using $I''$, we can rewrite our goal 
into the following form: 
$$(\Phi \cup \{\varsigma = \varsigma\}, \ss, \sh) \triangleright \rho_{I''}; C'(v)\;\{Q\}$$
The goal now matches the consequent of our premise $C({-}) 
\rightsquigarrow_I C'({-})$ after unfolding the definition of safe 
approximation, mutator coverage, and outcome coverage. Finally, the 
antecedent matches our assumption. 
\end{proof}

\begin{lem}[Soundness of $\mathsf{sassume}$]
\display{\mathsf{sassume}(b) \rightsquigarrow_I \mathsf{assume}(b)}
\end{lem}
\begin{proof}
Unfold the definition of safe 
approximation, mutator coverage, outcome coverage, and $\rho_I$. Fix an 
input symbolic state $(\Phi, \ss, \sh)$, a postcondition $Q$, an extension 
$I' \supseteq I$, and a semiconcrete state $(s, h)$ such that $I'((\Phi, 
\ss, \sh)) = (s, h)$. Unfold the definition of $\mathsf{sassume}$. Assume 
$\Phi \not\vdash_\mathrm{SMT} \lnot \llbracket b\rrbracket_{\ss} 
\Rightarrow (\Phi \cup \{\llbracket b\rrbracket_{\ss}\}, \ss, \sh) 
\triangleright \rho_I\;\{Q\}$. Unfold the definition of $\mathsf{assume}$. 
Assume $\llbracket b\rrbracket_s = \mathsf{true}$. Our goal reduces to 
$(s, h) \in Q$. 

Since $\llbracket \Phi\rrbracket_{I'} = \mathsf{true}$ and $\llbracket 
\llbracket b\rrbracket_{\ss}\rrbracket_{I'} = \mathsf{true}$, we have 
$\Phi \not\vDash \lnot \llbracket b\rrbracket_{\ss}$. By soundness of the 
SMT solver, it follows that $\Phi \not\vdash_\mathrm{SMT} \lnot \llbracket 
b\rrbracket_{\ss}$. Hence, by our assumption above, $(\Phi \cup 
\{\llbracket b\rrbracket_{\ss}\}, \ss, \sh) \triangleright \rho_I\;\{Q\}$. 
In this fact, we unfold $\rho_I$ and instantiate the demonic choice with 
$I'$. Since $I'((\Phi \cup \{\llbracket b\rrbracket_{\ss}\}, \ss, \sh)) = 
(s, h)$, we obtain $(s, h) \in Q$.
\end{proof}

\begin{thm}[Soundness of Featherweight VeriFast]
\emph{$$\mathsf{sym}\textsf{-}\mathsf{safe\_program}(c) \Rightarrow \mathsf{safe\_program}(c)\eqno{\qEd}$$}
\end{thm}\medskip

\noindent Combining the soundness of symbolic execution with respect to semiconcrete execution and the
soundness of semiconcrete execution with respect to concrete execution, we obtain the soundness of
Featherweight VeriFast: if symbolic execution does not fail, then concrete execution does not fail.

\section{Mechanisation}\label{sec:mech}

Above we presented a formal definition of Featherweight VeriFast and we 
gave the highlights of a proof of its soundness. We hope that the 
definitions are clear and the proof outline is convincing. However, the 
definition, while formal (in the sense of: consisting of symbols rather 
than natural language), is written in the general language of mathematics 
and not in any particular explicitly defined \emph{formal logic}, with a 
well-defined formal language of \emph{formulae} and a well-defined formal 
language of \emph{proofs} that specifies which formulae are 
\emph{logically true}. Therefore, the precise meaning of the definition 
might not be clear to all readers. A fortiori, the soundness proof is not 
expressed in such a formal language of proofs, and therefore, there is 
always the possibility that some of the inferences made are \emph{invalid} 
and the conclusion is \emph{false}; i.e., it is not an argument that will 
necessarily convince all readers.

To address these limitations, we developed a definition and soundness 
proof of a slight variant of Featherweight VeriFast, called Mechanised 
Featherweight VeriFast, in the machine-readable formal language of the 
interactive proof assistant Coq. Coq is a computer program that takes 
as input a set of files containing definitions and proofs expressed in its 
formal language, and checks that these definitions and proofs are indeed 
well-formed. Since we have successfully checked our development with Coq, 
we can have very high confidence that the theorems that we have proven are 
indeed true, with respect to the given definitions.

Note that it is still possible that Mechanised Featherweight VeriFast 
contains errors: it might still be the case that the stated definitions 
and theorems are not the ones that we \emph{intended}; for example, if we 
made an error in the definition of the concrete execution such that 
concrete execution always blocks, or we made an error in the definition of 
the symbolic execution such that symbolic execution always fails, then the 
soundness theorem holds vacuously and does not really tell us anything 
meaningful. We partially address this issue by including a small 
\emph{test suite} in our development, where we run the concrete execution 
and the symbolic execution on specific example programs, and test that 
concrete execution does indeed sometimes fail as expected, and that 
symbolic execution does indeed sometimes succeed as expected. Still, we 
should remain skeptical, and confidence in the \emph{relevance} of a 
formally proven statement can never be 100\%. It can be improved further 
by enlarging the test suite and/or by proving additional properties of the 
various executions, e.g.~by relating MFVF's concrete execution to another 
programming language semantics found in the literature.

While MFVF follows FVF very closely in most respects, there are a few 
differences, mainly motivated by the fact that we wanted MFVF to be 
\emph{executable} so as to be able to test it easily, whereas for FVF 
\emph{simplicity} is more important. Also, MFVF has a few minor additional 
features, which were left out of FVF, again for the sake of simplicity.

In the remainder of this section, we briefly discuss the main differences 
between MFVF and FVF and the executability of MFVF, we show the soundness 
theorem, and we point the reader to the full Coq sources which are 
available online.

\subsection{Differences between MFVF and FVF: Syntax}

In Figure~\ref{fig:mfvf} we show the syntax of the programming language and the annotations 
accepted by MFVF. The differences with FVF are shown in red; as the reader 
can see, they are very minor.

\begin{figure}
$$\begin{array}{r l}
& z \in \mathbb{Z}, n \in \mathbb{N}\\
& x \in \mathit{Vars}\\
e ::= & z\ |\ x\ |\ e + e\\
b ::= & e = e\ |\ e < e\ |\ \lnot b\\
c ::= & x := e\ |\ (c; c)\ |\ \mathbf{if}\ b\ \mathbf{then}\ c\ \mathbf{else}\ c\ |\ {\color{Red}\mathbf{skip}}\ |\ {\color{Red}\mathbf{message}\ \textit{text}}\\
& |\ {\color{Red}x\;{:=}\;}r(\overline{e})\ |\ x := \mathbf{malloc}(n)\ |\ x := [e]\ |\ [e] := e\ |\ \mathbf{free}(e)\\
& |\ \mathbf{while}\ b\ \mathbf{inv}\ a\ \mathbf{do}\ c\ |\ \mathbf{open}\ q(\overline{e}, {\color{Red}\overline{?\_}})\ |\ \mathbf{close}\ q(\overline{e})\\
\mathit{rdef} ::= & \mathbf{routine}\ r(\overline{x}) = c\\
\\
& q \in \mathit{UserDefinedPredicates}\\
p ::= & \mapsto\ |\ \mathsf{mb}\ |\ q\\
a ::= & b\ |\ p(\overline{e}, \overline{{?}x})\ |\ a * a\ |\ \mathbf{if}\ b\ \mathbf{then}\ a\ \mathbf{else}\ a\\
\mathit{preddef} ::= & \mathbf{predicate}\ q(\overline{x}) = a\\
\mathit{rspec} ::= & \mathbf{routine}\ r(\overline{x})\ \mathbf{req}\ a\ \mathbf{ens}\ a
\end{array}$$
\caption{Syntax of Mechanised Featherweight VeriFast's input language}\label{fig:mfvf}
\end{figure}

The main difference is that MFVF supports \emph{routine return values}; 
when executing a routine call $x := r(\overline{e})$, after execution of 
the routine body ends, the value assigned by the routine body to variable 
$\mathsf{result}$ is assigned to variable $x$ of the caller.

A minor difference is in the syntax of $\mathbf{open}$ commands: MFVF allows the command to leave some of the chunk arguments unspecified.
The command $\mathbf{open}\ q(\overline{e}, \overline{?\_})$ opens some chunk that matches the pattern $q(\overline{e}, \overline{?\_})$.

Two new commands are added. The $\mathbf{skip}$ command does nothing; it is equivalent to $x := x$.
The command $\mathbf{message}\ \mathit{text}$ prints message $\mathit{text}$ to the console.
This command is useful in MFVF for testing the executions.

\subsection{Differences between MFVF and FVF: Executions}

The main difference between the executions (concrete execution, 
semiconcrete execution, and symbolic execution) of MFVF and those of FVF 
is in the definition and use of the auxiliary mutators for the consumption 
of heap chunks ($\mathsf{cconsume\_chunks}(h)$, $\mathsf{consume\_chunks}(h)$, and 
$\mathsf{sconsume\_chunks}(\sh)$ in FVF). In FVF, these mutators take as an 
argument the precise multiset of heap chunks (up to provable equality for 
symbolic execution) to be consumed. However, at a typical use site, only 
part of the argument list of a chunk is fixed, and the remaining arguments 
are to be looked up in the heap. In FVF, this is achieved by angelically 
choosing these remaining arguments.

For example, consider symbolic execution of a heap lookup command:
$$\begin{array}{l}
\textrm{FVF:}\\
\mathsf{symexec}(x := [e]) =\\
\quad \sell \leftarrow \mathsf{seval}(e); \bigoplus \sv.\;\mathsf{sconsume\_chunk}(\sell \mapsto \sv); \mathsf{sproduce\_chunk}(\sell \mapsto \sv); x := \sv\\
\mathsf{sconsume\_chunk} \in \mathit{SHeaps} \rightarrow \mathit{SOutcomes}(\mathsf{unit})\\
\\
\textrm{MFVF:}\\
\mathsf{symexec}(x := [e]) =\\
\quad \sell \leftarrow \mathsf{seval}(e); {\color{Red}[\sv]\;{\leftarrow}}\;\mathsf{sconsume\_chunk}({\color{Red}{\mapsto}, [\sell], 1}); \mathsf{sproduce\_chunk}(\sell \mapsto \sv); x := \sv\\
\mathsf{sconsume\_chunk} \in \mathit{SPredicates} \rightarrow \mathit{Terms}^* \rightarrow \mathbb{N} \rightarrow \mathit{SOutcomes}(\mathit{Terms}^*)
\end{array}$$\medskip

\noindent In FVF, symbolic execution of a command of the form $x := 
[e]$ that reads the memory cell at address $e$ evaluates $e$ to obtain 
term $\sell$, then angelically chooses some term $\sv$, and then attempts 
to consume the points-to chunk that maps address $\sell$ to this 
angelically chosen term $\sv$. This consumption operation succeeds if a 
points-to chunk exists in the symbolic heap such that the SMT solver 
succeeds in proving that its arguments are equal to $\sell$ and $\sv$, 
respectively.

This definition is perfectly fine, except that angelically choosing a term 
from the set of all terms (that use only symbols that are already being 
used by the current symbolic state) is not directly executable, since that 
set is infinite, and even if it was finite, it would be highly 
inefficient. Therefore, in MFVF, a slightly more complex but directly 
executable version of the chunk consumption mutators is used. These 
mutators consume only a single chunk at a time, and they take as arguments 
the predicate name, the list of fixed chunk arguments, and the number of 
non-fixed chunk arguments; they return the values of the non-fixed 
arguments of the chunk that was consumed as their answer. Correspondingly, 
in MFVF, symbolic execution of $x := [e]$, rather than angelically 
choosing a term for the value of the cell, retrieves that term as the 
answer of the $\mathsf{sconsume\_chunk}$ auxiliary mutator.

\subsection{Executability}

This concludes the discussion of the differences between MFVF and FVF. We 
now discuss some specific encoding choices made when defining MFVF to 
obtain \emph{executable} definitions of symbolic execution and concrete 
execution.

\newcommand{\coqindtype}[1]{\textbf{\textsf{\small #1}}}

The most important such choice is in the definition of the type of 
outcomes. The definition of inductive type $\coqindtype{outcome}$ in 
MFVF is shown below.\vspace{-5 pt}
\begin{beamerframe}{Executability: Outcomes}

\begin{changemargin}{-1cm}{-1cm}
$$\begin{array}{l}
\texttt{Inductive}\ \coqindtype{type\_name} := \mathsf{n\_Empty\_set}\ |\ \mathsf{n\_bool}\ |\ \mathsf{n\_Z}\ |\ \mathsf{n\_T}(T: \texttt{Type})\mathsf{.}\\
\\
\texttt{Fixpoint}\ \textsf{Itype\_name}(n: \coqindtype{type\_name}): \texttt{Type} := \texttt{match}\ n\ \texttt{with}\\
\quad|\ \mathsf{n\_Empty\_set} \Rightarrow \coqindtype{Empty\_set}\\
\quad|\ \mathsf{n\_bool} \Rightarrow \coqindtype{bool}\\
\quad|\ \mathsf{n\_Z} \Rightarrow \coqindtype{Z}\\
\quad|\ \mathsf{n\_T}\ T \Rightarrow T\\
\quad\texttt{end}\mathsf{.}\\
\\
\texttt{Inductive}\ \coqindtype{set}(X: \texttt{Type}) := \mathsf{set\_}(n: \coqindtype{type\_name})(f: \mathsf{Itype\_name}\ n \rightarrow X)\mathsf{.}\\
\\
\texttt{Inductive}\ \coqindtype{outcome}(S\ A: \texttt{Type}) :=\\
|\ \mathsf{single}(s: S)(a: A)\\
|\ \mathsf{demonic}(\mathit{os}: \coqindtype{set}\ (\coqindtype{outcome}\ S\ A))\\
|\ \mathsf{angelic}(\mathit{os}: \coqindtype{set}\ (\coqindtype{outcome}\ S\ A))\\
|\ \mathsf{message}(\mathit{msg}: \coqindtype{string})(o: \coqindtype{outcome}\ S\ A)\mathsf{.}
\end{array}$$
\end{changemargin}

\end{beamerframe}\medskip

\noindent It corresponds exactly to the definition of outcomes 
given earlier for FVF (except for the extra case of messages): an outcome 
$\phi$ is either a singleton outcome $\langle \sigma, a\rangle$ with 
output state $\sigma$ and answer $a$, or a demonic choice $\bigotimes 
\Phi$ or angelic choice $\bigoplus \Phi$ over a set of outcomes $\Phi$. 
However, there are two interesting aspects about this definition, and more 
specifically, about the type $\coqindtype{{set}}$ used for the sets of 
outcomes.

First of all, we had to choose this type carefully to obtain a proper 
inductive definition. The simplest approach for defining a type for sets 
of elements of some type $X$ is as follows: 
$$\texttt{Definition}\ \mathsf{set}\ X := X \rightarrow 
\texttt{Prop}\mathsf{.}$$ That is, a set of elements of type $X$ is simply 
a predicate over type $X$. However, using this definition of sets in the 
definition of outcomes would cause Coq to reject the definition of 
outcomes, since it would not be a proper inductive definition. Indeed, it 
would allow us to write $\mathsf{demonic}\ (\lambda \_.\; \mathsf{True})$, 
denoting the demonic choice over \emph{all} outcomes, including 
$\mathsf{demonic}\ (\lambda \_. \mathsf{True})$ itself, defeating the 
crucial notion that each value of an inductive type is built from 
\emph{smaller} values of that type, and thus rendering proof by induction 
unsound. 

Perhaps the simplest possible definition for a type of sets of elements of 
type $X$ that is compatible with inductive definitions is the following: 
$$\texttt{Inductive}\ \coqindtype{{set}}(X: \texttt{Type}) := \mathsf{set\_}(I: \texttt{Type})(f: I \rightarrow 
X)\mathsf{.}$$ This type allows a set to be constructed by providing an 
\emph{index type} $I$ and a function $f$ that maps each value of type $I$ 
to some value of type $X$. For example, the set containing exactly the 
integers 24 and 42 can be constructed as follows:
$$\mathsf{set\_}\ \coqindtype{{bool}}\ (\lambda b.\;\mathsf{if}\ b\ 
\mathsf{then}\ 24\ \mathsf{else}\ 42)$$ Using this type of sets in the definition of outcomes would be accepted by Coq.

However, another problem would still remain: we would like to write a Coq 
function that takes the outcome of symbolically executing some program 
starting from the empty symbolic state, and decides if that outcome 
satisfies postcondition $\mathsf{True}$, i.e., if symbolic execution has 
failed or not. An outcome satisfies postcondition $\mathsf{True}$ iff the 
outcome is a singleton outcome, or it is a demonic choice over some set of 
outcomes, each of which satisfies postcondition $\mathsf{True}$, or it is 
an angelic choice over some set of outcomes, at least one of which 
satisfies postcondition $\mathsf{True}$ (or it is a message outcome and 
its continuation satisfies postcondition $\mathsf{True}$). So, for demonic 
and angelic choice over some set of outcomes, we need to be able to 
enumerate the elements of the set. Given the definition of sets above, it 
would be necessary to enumerate the elements of the index type $I$. 
Unfortunately, however, this is not generally possible: the index type 
might be infinite.

However, MFVF's definition of symbolic execution uses only very restricted 
forms of demonic or angelic choice: it only uses blocking, failure, and 
binary choice, i.e., choices over zero elements or two elements. So, if in 
symbolic execution we use as index types only the type 
$\coqindtype{{Empty\_set}}$ and the type $\coqindtype{{bool}}$, can 
we write our Coq function? Unfortunately, still no, because this would 
require our function to perform a case analysis on a comparison between the 
index type of a set and the types $\coqindtype{{Empty\_set}}$ or 
$\coqindtype{{bool}}$, and Coq's execution engine does not support 
this. Coq's execution engine supports only pattern matching on values of 
inductive types, and types themselves are not values of inductive types.

The solution we adopted for this problem is to not directly allow 
arbitrary types to be specified as the index type when constructing a set, 
but rather to define an inductive type \coqindtype{{type\_name}} of 
\emph{type names}, with names for type $\coqindtype{{Empty\_set}}$, for 
type $\coqindtype{{bool}}$, and for the type $\coqindtype{{Z}}$ of 
integers, and a fallback case $\mathsf{n\_T}$ for arbitrary types. We also 
defined an interpretation function $\mathsf{Itype\_name}$ for these type 
names that maps each type name to its corresponding type. By using type 
names and the interpretation function in the definition of sets, we were 
able to write a Coq function that decides whether an outcome satisfies 
postcondition $\mathsf{True}$. (For the cases $\mathsf{n\_Z}$ and 
$\mathsf{n\_T}$ this function is not executable, but since symbolic 
execution uses only the other two cases, it executes properly for the 
outcomes of symbolic execution.)

\begin{figure}
$$\begin{array}{l}
\mathtt{Definition}\ \mathsf{listDef} :=\\
\label{listing:reverse}
\quad \mathbf{predicate}\ \mathsf{list}(\mathsf{l}) =\\
\quad\quad
\mathbf{if}\ \mathsf{l} = 0\ \mathbf{then}\ 0 = 0\ \mathbf{else}\ 
\mathsf{mb}(\mathsf{l}, 2) * \mathsf{l} \mapsto \_ * \mathsf{l} + 1 \mapsto {?}\mathsf{next} * \mathsf{list}(\mathsf{next})\mathsf{.}\\
\\
\mathtt{Compute}\ \mathsf{svalid\_routine}\ [\mathsf{listDef}]\ []\ [\mathsf{l}]\ \mathsf{list}(\mathsf{l})\ \mathsf{list}(\mathsf{result})\\
\quad(\\
\quad\quad \mathbf{close}\ \mathsf{list}(\mathsf{b});\\
\quad\quad \mathbf{while}\ \lnot(\mathsf{a} = 0)\ \mathbf{inv}\ \mathsf{list}(\mathsf{a}) * \mathsf{list}(\mathsf{b})\ \mathbf{do}\ (\\
\quad\quad\quad \mathbf{open}\ \mathsf{list}(\mathsf{a});\\
\quad\quad\quad \mathsf{n} := [\mathsf{a} + 1]; [\mathsf{a} + 1] := \mathsf{b}; \mathsf{b} := \mathsf{a}; \mathsf{a} := \mathsf{t};\\
\quad\quad\quad \mathbf{close}\ \mathsf{list}(\mathsf{b})\\
\quad\quad );\\
\quad\quad \mathbf{open}\ \mathsf{list}(\mathsf{a});\\
\quad\quad \mathsf{result} := \mathsf{b}\\
\quad)\mathsf{.}\\
\textit{\textsf{ok}}
\end{array}$$
\caption{Running MFVF on an in-place list reversal routine}\label{fig:reverse}
\end{figure}

Figure~\ref{fig:reverse} shows an example where we run symbolic execution to check validity of a 
routine that performs in-place reversal of a linked list.
Function 
$\mathsf{svalid\_routine}$ (the syntax of the example is slightly 
simplified from the actual Coq development) takes as arguments a list of 
predicate definitions (in the example, just the definition $\mathsf{listDef}$ of 
the $\mathsf{list}$ predicate), a list of routine specifications (in the 
example, an empty list, since the list reversal routine does not itself 
perform any routine calls), a list of parameters (in the example, just a 
single parameter $\mathsf{l}$, a pointer to the linked list to be 
reversed), a precondition (in the example, $\mathsf{list}(\mathsf{l})$, 
expressing that the routine expects to find a linked list at address 
$\mathsf{l}$), a postcondition (in the example, 
$\mathsf{list}(\mathsf{result})$, expressing that after the routine 
completes, the routine's result will point to a linked list), and the body 
of the routine to be verified. Coq command $\texttt{Compute}$ evaluates a 
Coq expression and prints the result: in the example, the result is 
\textit{\textsf{ok}}, indicating that the routine was verified 
successfully. 

MFVF includes an executable definition of symbolic execution and a 
\emph{semi-executable} definition of concrete execution. Concrete 
execution is \emph{semi-executable} in the sense that we have been able to 
write Coq functions that compute, for a given input program and a given 
sequence of values for demonic choices over the booleans or the integers, 
if concrete execution, for those choices, ends up in a singleton outcome 
or in failure.

For example, consider the program that allocates a memory cell and then 
accesses the memory cell at address 42. If the newly allocated memory cell 
was allocated at address 42, execution succeeds; otherwise, it fails.

We can easily check that both execution paths do indeed behave as 
expected, using the Coq functions $\mathsf{atZ}$, $\mathsf{isSingle}$, and 
$\mathsf{isFail}$ shown in Figure~\ref{fig:mfvf-cexec}. As shown in the figure, using Coq's $\texttt{Compute}$ command, 
and by first picking value 2 for the depth of concrete execution (any 
greater value would do as well) and then 42 for the address of the newly 
allocated memory block, we can confirm that we end up in a singleton 
outcome, and that by alternatively picking address 43 we end up in 
failure.

\begin{figure}
$$\begin{array}{l}
\mathtt{Definition}\ \mathsf{atZ}\ z\ o := \mathtt{match}\ o\ \mathtt{with}\\
\quad|\ \mathsf{Some}\ (\mathsf{demonic}\ (\mathsf{set\_}\ \mathsf{n\_Z}\ o')) \Rightarrow \mathsf{Some}\ (o'\ z)\\
\quad|\ \_ \Rightarrow \mathsf{None}\ \mathtt{end}\mathsf{.}\\
\mathtt{Definition}\ \mathsf{isSingle}\ o :=\\
\quad\texttt{match}\ o\ \mathtt{with}\ \mathsf{Some}\ (\mathsf{single}\ \_\ \_) \Rightarrow \mathsf{true}\ |\ \_ \Rightarrow \mathsf{false}\ \mathtt{end}\mathsf{.}\\
\mathtt{Definition}\ \mathsf{isFail}\ o := \texttt{match}\ o\ \mathtt{with}\\
\quad|\ \mathsf{Some}\ (\mathsf{angelic}\ (\mathsf{set\_}\ \mathsf{n\_Empty\_set}\ \_)) \Rightarrow \mathsf{true}\\
\quad|\ \_ \Rightarrow \mathsf{false}\ \mathtt{end}\mathsf{.}\\
\\
\mathtt{Definition}\ \mathsf{o} := \mathsf{cstate0}\ \triangleright\ \mathsf{exec}\ []\ (\mathsf{x} := \mathbf{malloc}(1); [42] := 123)\mathsf{.}\\
\\
\mathtt{Compute}\ \mathsf{Some}\ \mathsf\mathsf{o}\ \triangleright\ \mathsf{atZ}\ 2\ \triangleright\ \mathsf{atZ}\ 42\ \triangleright\ \mathsf{isSingle}\mathsf{.}\\
\textit{\textsf{true}}\\
\mathtt{Compute}\ \mathsf{Some}\ \mathsf\mathsf{o}\ \triangleright\ \mathsf{atZ}\ 2\ \triangleright\ \mathsf{atZ}\ 43\ \triangleright\ \mathsf{isFail}\mathsf{.}\\
\textit{\textsf{true}}
\end{array}$$
\caption{Testing MFVF concrete execution}\label{fig:mfvf-cexec}
\end{figure}

The definition of function $\mathsf{atZ}$ exploits the fact 
that there is a separate case for type $\coqindtype{{Z}}$ in type 
$\coqindtype{{type\_name}}$, and that concrete execution of 
$\mathbf{malloc}$ commands uses this type name in its demonic choice. 

\subsection{Soundness}

The Coq statement of the soundness theorem is shown below: if symbolic 
execution of a program does not fail, then concrete execution of that 
program does not fail. The proof is accepted by Coq.

$$\begin{array}{@{} l @{}}
\mathtt{Theorem}\ \mathsf{soundness}\ \mathit{rspecs}\ \mathit{pdefs}\ \mathit{rdefs}\ c:\\
\quad \mathsf{svalid\_program}\ \mathit{rspecs}\ \mathit{pdefs}\ \mathit{rdefs}\ c = \mathsf{ok}\rightarrow\\
\quad \mathsf{cvalid\_program}\ \mathit{rdefs}\ c\mathsf{.}\\
\mathtt{Proof.}\\
\cdots\\
\mathtt{Qed.}\\
\\
\mathtt{Print\ Assumptions}\ \mathsf{soundness.}\\
\textit{\textsf{Coq.Sets.Ensembles.Extensionality\_Ensembles}}\\
\textit{\textsf{Coq.Logic.Classical\_Prop.classic}}\\
\textit{\textsf{Coq.Logic.IndefiniteDescription.constructive\_indefinite\_description}}\\
\textit{\textsf{Coq.Logic.FunctionalExtensionality.functional\_extensionality\_dep}}
\end{array}$$\medskip

\noindent We can use Coq's \texttt{Print Assumptions} command to check which axioms 
are used (directly or indirectly) in the proof of the soundness theorem. 
Only the four listed axioms are used: they are axioms of classical logic, 
offered by the Coq standard library. 

The Coq development can be browsed in HTML and PDF form and the full 
sources can be downloaded at \textsf{http://www.cs.kuleuven.be/\symbol{126}bartj/fvf/}.

\section{Related work}\label{sec:related-work}

\subsection{Hoare logic, separation logic}

A more abstract, higher-level approach for reasoning about imperative 
pointer-manipulating programs is given by separation logic 
\cite{seplogic-ohearn,seplogic,ohearn-marktoberdorf}, which is an 
extension of Hoare logic \cite{hoare-logic}.

Hoare logic deals with \emph{program correctness judgments} (also known as 
\emph{Hoare triples}) of the form $\{b\}\ c\ \{b'\}$, where $b$, the 
\emph{precondition}, and $b'$, the \emph{postcondition}, are boolean 
expressions (as in Definition~\ref{defi:cmd-syntax}, except that they may 
also contain additional logical operators such as conjunction and quantification), and $c$ is 
a command that does not involve the heap (i.e., it does not allocate, 
deallocate, or access heap cells); the judgment means that $c$, when 
started with a store that satisfies precondition $b$, if it terminates, 
terminates with a store that satisfies postcondition $b'$:
$$\forall s.\;\llbracket b\rrbracket_s = \mathsf{true} \Rightarrow s 
\triangleright \mathsf{exec}(c)\ \{s'.\;\llbracket b'\rrbracket_{s'} = \mathsf{true}\}$$

Hoare logic defines a number of \emph{axioms} and \emph{inference rules} 
for deriving correctness judgments; the main ones are shown in 
Figure~\ref{fig:hoarelogic}. Here, $b[e/x]$ denotes the boolean expression 
obtained by substituting expression $e$ for variable $x$ in $b$, and $b 
\Rightarrow b'$ denotes that $b$ implies $b'$ in all stores, \ie{} 
$\forall s.\;\llbracket b\rrbracket_s \Rightarrow \llbracket 
b'\rrbracket_s$. 

\begin{figure}
\begin{mathpar}
\inferrule[Assign]{}{
\{b[e/x]\}\ x := e\ \{b\}
}
\and
\inferrule[If]{
\{b \land b'\}\ c\ \{b''\}\\
\{b \land \lnot b'\}\ c'\ \{b''\}
}{
\{b\}\ \mathbf{if}\ b'\ \mathbf{then}\ c\ \mathbf{else}\ c'\ \{b''\}
}
\and
\inferrule[While]{
\{b \land b'\}\ c\ \{b\}
}{
\{b\}\ \mathbf{while}\ b'\ \mathbf{do}\ c\ \{b \land \lnot b'\}
}
\and
\inferrule[Seq]{
\{b\}\ c\ \{b'\}\\
\{b'\}\ c'\ \{b''\}
}{
\{b\}\ c; c'\ \{b''\}
}
\and
\inferrule[Conseq]{
b \Rightarrow b'\\
\{b'\}\ c\ \{b''\}\\
b'' \Rightarrow b'''
}{
\{b\}\ c\ \{b'''\}
}
\and
\inferrule[Exists]{
\forall v.\;\{b[v/x]\}\ c\ \{b'[v/x]\}
}{
\{\exists x.\;b\}\ c\ \{\exists x.\;b'\}
}
\end{mathpar}
\caption{The main axioms and inference rules of Hoare logic}\label{fig:hoarelogic}
\end{figure}

For example, we can derive the judgment $\{0 \le n\}\ i := 0; 
\mathbf{while}\ i < n\ \mathbf{do}\ i := i + 1\ \{i = n\}$ using the proof 
tree in Figure~\ref{fig:hoare-tree}.

\begin{figure}
\[
\trfrac[\textsc{\small Seq}]{
\trfrac[\textsc{\small Assign}]{}{
(b)
}\quad
\trfrac[\textsc{\small Conseq}]{
(d)\quad
\trfrac[\textsc{\small While}]{
\trfrac[\textsc{\small Conseq}]{
(h)\quad
\trfrac[\textsc{\small Assign}]{}{
(i)
}\quad
(j)
}{
(g)
}
}{
(e)
}\quad
(f)
}{
(c)
}
}{
(a)
}
\]
\[
\begin{array}{l l}
(a) & \{0 \le n\}\ i := 0; \mathbf{while}\ i < n\ \mathbf{do}\ i := i + 1\ \{i = n\}\\
(b) & \{0 \le n\}\ i := 0\ \{i \le n\}\\
(c) & \{i \le n\}\ \mathbf{while}\ i < n\ \mathbf{do}\ i := i + 1\ \{i = n\}\\
(d) & i \le n \Rightarrow i \le n\\
(e) & \{i \le n\}\ \mathbf{while}\ i < n\ \mathbf{do}\ i := i + 1\ \{i \le n \land \lnot (i < n)\}\\
(f) & i \le n \land \lnot (i < n) \Rightarrow i = n\\
(g) & \{i \le n \land i < n\}\ i := i + 1\ \{i \le n\}\\
(h) & i \le n \land i < n \Rightarrow i + 1 \le n\\
(i) & \{i + 1 \le n\}\ i := i + 1\ \{i \le n\}\\
(j) & i \le n \Rightarrow i \le n\\
\end{array}
\]
\caption{Proof tree in Hoare logic for a simple example program}\label{fig:hoare-tree}
\end{figure}

A more convenient representation of this proof tree is in the form of the \emph{proof outline} of Figure~\ref{fig:hoare-outline}, where assertions inserted between components of a
sequential composition indicate applications of the \textsc{Seq} rule, and multiple consecutive assertions indicate applications of the \textsc{Conseq}
rule.
\begin{figure}
\begin{displaymath}
\begin{array}{l}
\{0 \le n\}\\
i := 0;\\
\{i \le n\}\\
\mathbf{while}\ i < n\ \mathbf{do}\\
\quad \{i \le n \land i < n\}\\
\quad \{i + 1 \le n\}\\
\quad i := i + 1\\
\quad \{i \le n\}\\
\{i \le n \land \lnot (i < n)\}\\
\{i \le n\}
\end{array}
\end{displaymath}
\caption{Proof outline in Hoare logic for a simple example program}\label{fig:hoare-outline}
\end{figure}

Separation logic extends Hoare logic with additional assertion logic constructs and proof rules for reasoning conveniently about heap-manipulating programs.
The syntax of separation logic assertions extends the syntax of logical formulae with constructs for specifying the heap: the assertion $\mathbf{emp}$ states
that the heap is empty; the points-to assertion $e \mapsto e'$ states that the heap consists of exactly one heap cell, mapping address $e$ to value $e'$, and
the separating conjunction $P * Q$ states that the heap can be split into two disjoint parts such that $P$ holds for one part and $Q$ holds for the other.
Instead of Featherweight VeriFast's $?x$ syntax, separation logic uses regular existential quantification. Formally:
$$\begin{array}{l l}
s, h \vDash \mathbf{emp} & \Leftrightarrow h = \emptyset\\
s, h \vDash e \mapsto e' & \Leftrightarrow h = \{(\llbracket e\rrbracket_s, \llbracket e'\rrbracket_s)\}\\
s, h \vDash P * Q & \Leftrightarrow \exists h_1, h_2.\;h = h_1 \uplus h_2 \land s, h_1 \vDash P \land s, h_2 \vDash Q\\
s, h \vDash \exists x.\;P & \Leftrightarrow \exists v.\;s[x:=v], h \vDash P\\
s, h \vDash b & \Leftrightarrow \llbracket b\rrbracket_s = \mathsf{true}\\
\end{array}$$
In separation logic, predicates are typically treated like inductive definitions, i.e.~their meaning is taken to be the smallest interpretation
(i.e.~set of heaps) that satisfies the definition;
such an interpretation always exists (by the Knaster-Tarski theorem) provided that predicates are used inside of predicate
definitions only in positive positions, i.e.~not under negations or on the left-hand side of implications \cite{seplogic-abstraction}.
The typical example of such a predicate is the predicate $\mathsf{lseg}(\ell, \ell')$ denoting a linked list segment from a starting node $\ell$ (inclusive)
to a limiting node $\ell'$ (exclusive):
$$\mathsf{lseg}(\ell, \ell') \stackrel{\mathrm{def}}{=} \ell = \ell' \land \mathbf{emp} \lor \exists v, n.\;\ell \mapsto v * \ell + 1 \mapsto n * \mathsf{lseg}(n, \ell')$$
Separation logic's program logic extends Hoare logic's inference system with axioms for the heap manipulation commands and the \emph{frame axiom} (see Figure~\ref{fig:seplogic}).
The former are \emph{small axioms}: they mention only the heap cells required for the command to succeed. The frame axiom allows the small axioms to be lifted to larger
heaps. Similarly, one can write small specifications for routines and use the frame axiom to lift those to the larger heap present in a given calling context.

\begin{figure}
\begin{mathpar}
\inferrule[Cons]{}{
\{\mathbf{emp}\}\ x := \mathbf{cons}(v, v')\ \{x \mapsto v * x + 1 \mapsto v'\}
}
\and
\inferrule[Dispose]{}{
\{\ell \mapsto v\}\ \mathbf{dispose}(\ell)\ \{\mathbf{emp}\}
}
\and
\inferrule[Read]{}{
\{\ell \mapsto v\}\ x := [\ell]\ \{\ell \mapsto v \land x = v\}
}
\and
\inferrule[Write]{}{
\{\ell \mapsto v\}\ [\ell] := v'\ \{\ell \mapsto v'\}
}
\and
\inferrule[Frame]{
\{P\}\ c\ \{Q\}
}{
\{P * R\}\ c\ \{Q * R\}
}
\ \textrm{if}\ \mathsf{targets}(c) \cap \mathsf{freevars}(R) = \emptyset
\end{mathpar}
\caption{The additional proof rules of separation logic}\label{fig:seplogic}
\end{figure}

 Figure~\ref{fig:reverse-seplogic} shows a proof outline in separation logic for a program that performs an in-place reversal of a linked list.\enlargethispage{\baselineskip}

\begin{figure}
$$\begin{array}{l}
\{\mathsf{lseg}(i, 0)\}\\
j := 0;\\
\{\mathsf{lseg}(i, 0) * \mathsf{lseg}(j, 0)\}\\
\mathbf{while}\ i \neq 0\ \mathbf{do}\ (\\
\quad \{\mathsf{lseg}(i, 0) * \mathsf{lseg}(j, 0) \land i \neq 0\}\\
\quad \{\exists v, n.\;i + 1 \mapsto n * i \mapsto v * \mathsf{lseg}(n, 0) * \mathsf{lseg}(j, 0)\}\\
\quad\quad \{i + 1 \mapsto n * i \mapsto v * \mathsf{lseg}(n, 0) * \mathsf{lseg}(j, 0)\}\quad\textrm{Rule \textsc{Exists}. Fix $v, n$.}\\
\quad\quad\quad \{i + 1 \mapsto n\}\quad\textrm{Rule \textsc{Frame}.}\\
\quad\quad\quad k := [i + 1];\\
\quad\quad\quad \{i + 1 \mapsto n \land k = n\}\\
\quad\quad \{(i + 1 \mapsto n \land k = n) * i \mapsto v * \mathsf{lseg}(n, 0) * \mathsf{lseg}(j, 0)\}\\
\quad\quad \{i + 1 \mapsto k * i \mapsto v * \mathsf{lseg}(k, 0) * \mathsf{lseg}(j, 0)\}\\
\quad\quad\quad \{i + 1 \mapsto k\}\quad\textrm{Rule \textsc{Frame}.}\\
\quad\quad\quad [i + 1] := j;\\
\quad\quad\quad \{i + 1 \mapsto j\}\\
\quad\quad \{i + 1 \mapsto j * i \mapsto v * \mathsf{lseg}(k, 0) * \mathsf{lseg}(j, 0)\}\\
\quad\quad \{\mathsf{lseg}(k, 0) * \mathsf{lseg}(i, 0)\}\\
\quad\quad j := i;\\
\quad\quad \{\mathsf{lseg}(k, 0) * \mathsf{lseg}(j, 0)\}\\
\quad\quad i := k\\
\quad\quad \{\mathsf{lseg}(i, 0) * \mathsf{lseg}(j, 0)\}\\
\quad \{\mathsf{lseg}(i, 0) * \mathsf{lseg}(j, 0)\}\\
)\\
\{\mathsf{lseg}(i, 0) * \mathsf{lseg}(j, 0) \land i = 0\}\\
\{\mathsf{lseg}(j, 0)\}
\end{array}$$
\caption{A proof outline in separation logic of a program that performs an in-place reversal of a linked list}\label{fig:reverse-seplogic}
\end{figure}

VeriFast could be considered to be a type of ``separation logic theorem 
prover'', by interpreting the input files as a separation logic Hoare triple that serves as the proof goal and 
the annotations as hints to direct the construction of the proof. From 
this point of view, VeriFast applies the separation logic frame rule when 
verifying loops and routine calls.

\subsection{Separation logic tools}



\subsubsection{Smallfoot}

Smallfoot \cite{smallfoot2006} was a breakthrough in program verification tool development;
it was successful in its goal of showcasing for the first time the power of separation logic for automated program verification and analysis.
Like FVF, it takes as input an annotated program and checks each procedure
against its contract. The programming language is very similar: like FVF's, it is a simple while language with
procedures. The main difference is that it
includes concurrency constructs (resource declarations, parallel procedure calls, and conditional critical regions).
The annotation language is very similar as well: a precondition and postcondition must
be specified for each procedure, and a loop invariant must be specified for each loop; these are not inferred.
(If one of these is omitted, it defaults to $\mathbf{emp}$.)
The main difference is that besides the points-to assertion, Smallfoot has built-in predicates for trees, list segments, doubly-linked lists, and xor lists,
does not support user-defined predicates, and does not require $\mathbf{open}$ or $\mathbf{close}$ commands or any other kinds of proof hints (other than the procedure
and loop annotations mentioned above). Another difference is that it does not support (even FVF's very restricted form of) existential quantification.

The main difference in Smallfoot's functional behavior is that it is automatic: thanks to a complete, decidable proof theory for the supported assertion language,
Smallfoot never requires proof hints. In particular, not only does it automatically fold and unfold the definitions of the inductive predicates
(which in FVF requires $\mathbf{open}$ and $\mathbf{close}$ ghost commands),
it also has sufficient rules built in to reason automatically about inductive properties such as appending two list segments.
In FVF, this would require defining and calling a recursive ``lemma'' routine that establishes the property.

While Smallfoot's algorithm is in many ways more powerful and more interesting than FVF's, FVF's goal is educational,
and we believe its presentation in this article succeeds better at clearly conveying the essence of VeriFast's operation,
especially to an audience that is new to formal methods,
than the presentation of Smallfoot's operation \cite{smallfoot2006,smallfoot-symbex} does.

\subsubsection{Other tools}

Smallfoot's algorithm has been used as a basis for \emph{shape analysis} algorithms that automatically infer loop invariants and postconditions \cite{distefano-tacas2006},
and even preconditions \cite{DBLP:conf/popl/CalcagnoDOY09}.
These algorithms have been implemented in a tool called Infer \cite{infer} that has successfully been exploited commercially.
Another tool based on these ideas, called SLAyer \cite{slayer}, is being used inside Microsoft to verify Windows device drivers.

These techniques have been extended to a concurrent setting, e.g.~to infer invariants for shared resources \cite{distefano-resource-invariants}. Integration of
separation logic and rely-guarantee reasoning \cite{DBLP:conf/ifip/Jones83} has led to tools SmallfootRG \cite{smallfootrg} for verifying safety properties
and Cave \cite{cave} for verifying linearizability of fine-grained concurrent modules.

Extensions of separation logic for dealing with object-oriented programming patterns such as dynamic binding have been implemented in the tool jStar \cite{jstar}
that takes as input a Java program, a precondition and postcondition for each method, and a set of inference and abstraction rules, and attempts to automatically apply
these rules to verify each method body against its specification. jStar does not require (or support) annotations inside method bodies.

The HIP/SLEEK toolstack \cite{DBLP:journals/scp/ChinDNQ12} uses separation logic-based symbolic execution to automatically verify shape, size, and bag properties
of programs. Like VeriFast, it supports user-defined recursive predicates to express the shape of data structures.

\subsubsection{Proof assistant-based approaches}

Like VeriFast, the tools mentioned above take as input annotated programs and then run without further user interaction.
Another approach is to see program verification as a special case of interactive proof development, and to extend proof assistants like Isabelle/HOL and Coq
with theories defining program syntax and semantics and specification formalisms, as well as lemmas and tactics (reusable proof scripts) for aiding users in discharging
proof obligations.

Holfoot \cite{UCAM-CL-TR-799} is an implementation of Smallfoot inside the HOL 4 theorem prover.
In addition to the features supported by Smallfoot it can handle data and supports interactive proofs.
Moreover, it can handle arrays. Simple specifications with data like copying a list can be handled automatically.
More complicated ones like fully functional specifications of filtering a list, mergesort, quicksort or an
implementation of red-black trees require user interaction.
During this interaction all the features of the HOL 4 theorem prover can be used, including the interface to external SMT solvers like Yices.

Ynot \cite{DBLP:conf/icfp/ChlipalaMMSW09} is a library for the Coq proof assistant which turns it into a 
full-fledged environment for writing and verifying imperative programs. In 
the tradition of the Haskell IO monad, Ynot axiomatizes a parameterized 
monad of imperative computations, where the type of a computation specifies not only what type of data it returns, but also what Hoare-logic-style 
precondition and postcondition it satisfies. On top of the simple 
axiomatic base, the library defines a separation logic. Specialized 
automation tactics are able to discharge automatically most proof goals 
about separation-style formulas that describe heaps, meaning that building 
a certified Ynot program is often not much harder than writing that 
program in Haskell. 

Bedrock \cite{DBLP:conf/pldi/Chlipala11} is a Coq library for mostly-automated verification of low-level programs in computational separation logic; a major
difference from Ynot is that it has improved support for reasoning about code pointers.

Charge! \cite{charge} is a set of tactics for working with a shallow embedding of a higher-order separation logic for a subset of Java in Coq.

The Verified Software Toolchain project \cite{appel-verified-c-shakedown,vstbook,verified-software-toolchain-esop-2011} has produced a separation logic for C, called Verified C, in the form of a Coq library, as well as a Smallfoot implementation in Coq, extractable to OCaml, called VeriSmall \cite{verismall}, both proven sound in Coq with respect to the operational semantics of C against with the CompCert project \cite{Leroy-Compcert-CACM} verified the correctness of their C compiler, thus obtaining that the compiled program satisfies the verified properties.

\subsection{Non-separation logic tools}\enlargethispage{\baselineskip}

Another approach for extending Hoare logic to reason about programs with 
pointers (or other kinds of aliasing, such as Java's object references) is 
to simply treat the heap as a program variable whose value is a function 
that maps addresses to values, and to retain regular classical logic as 
the assertion language. The following tools are based on this approach. 

In this approach, the separation logic frame rule and small axioms that 
allow a simple syntactic treatment of heap mutation and procedure effect 
framing are generally not available, but other approaches to procedure 
effect framing may be used. Most alternative approaches are variants of 
\emph{dynamic frames} \cite{dynamic-frames}, where a module uses abstract 
variables of type ``set of memory locations'' to abstractly specify which 
memory locations are modified by a procedure as well as which memory 
locations may influence the value of abstract variables. 

VCC \cite{DBLP:conf/tphol/CohenDHLMSST09} is a verifier for concurrent C 
programs annotated with contracts expressed in classical logic. For each C 
function, VCC generates a set of verification conditions (using a variant 
of \emph{weakest preconditions} \cite{weakest-preconditions}) to be 
discharged by an SMT solver. For modularity, it uses the \emph{admissible 
invariants} approach: a \emph{two-state invariant} may be associated with 
each C struct instance $s$, which may mention the fields of $s$ as well as 
those of other struct instances $s'$, provided it is \emph{admissible}: 
any update of $s'.f$ that satisfies the invariant of $s'$ must preserve 
the invariant of $s$. By encoding an ownership system on top of this 
approach, it can be used both for precise reasoning about fine-grained 
concurrency and for reasoning in a dynamic frames-like style about 
sequential code. VCC has been used to verify a large part of the Microsoft 
Hyper-V hypervisor. 

Other important non-separation logic tools include Chalice \cite{chalice} 
(a verifier for concurrent Java-like programs based on \emph{implicit 
dynamic frames} \cite{DBLP:journals/toplas/SmansJP12}), Dafny 
\cite{DBLP:conf/vstte/Leino12}, KeY \cite{DBLP:conf/fmoods/AhrendtBHS07}, 
and KIV \cite{kiv}.

As in the case of separation logic-based approaches, some non-separation 
logic-based verification efforts have been carried out in a 
general-purpose proof assistant rather than a specialized tool. Notable in 
this category are the L4.verified project \cite{Klein_EHACDEEKNSTW_09}, 
which verified an OS microkernel consisting of 8KLOC of C code in 
Isabelle/HOL, and the Verisoft project \cite{Alkassar:VSTTE2010-71}, which 
performed large parts of the pervasive verification, also in Isabelle/HOL, 
of the complete software stack (plus parts of the hardware), including 
microkernel, kernel, and applications, of a secure e-mail system and an 
embedded automotive system.

\subsection{Semantic framework: Outcomes}

In our formalization, to express and relate the semantics of the 
programming language and the verification algorithm, via the intermediary 
of semiconcrete execution, we developed the semantic framework based on 
\emph{outcomes}, with the important derived concepts of \emph{mutators}, 
\emph{postcondition satisfaction}, and \emph{coverage}. This enabled us to 
deal conveniently with failure, nontermination, and both demonic and 
angelic nondeterminism.

This framework is essentially nothing more than the predicate transformer 
semantics proposed by Dijkstra \cite{weakest-preconditions}:
$$\mathsf{exec}(c)\;\{Q\} \equiv \mathsf{wp}(c, Q)$$

Also, mutators with answers are essentially a combination of a state monad 
and a continuation monad.

Our choice of defining the set of outcomes as an inductive datatype, 
rather than a predicate over postconditions (i.e.~a function from 
postconditions to $\mathsf{bool}$, such that mutators would be predicate 
transformers or functions from predicates to predicates) or, equivalently, 
a state-continuation monad, has two advantages: firstly, we immediately 
have that all outcomes are \emph{monotonic} (postcondition satisfaction is 
preserved by weakening of the postcondition); and secondly, our Coq 
encoding of concrete execution yields not an unexecutable function to 
$\mathsf{bool}$ but an (infinite-branching) execution tree which we can 
explore, as shown in Section~\ref{sec:mech}.

\subsection{Machine-checked tools}

An effort similar to our executable machine-checked encoding into Coq of 
Featherweight VeriFast is the executable machine-checked encoding into Coq 
of Smallfoot, called VeriSmall \cite{verismall}.

Whereas VeriSmall's primary purpose is to serve as the basis for a 
certified program verification tool chain, MFVF's primary purpose is to 
serve as evidence for the correctness of the presentation of FVF and its 
soundness proof in this article. Therefore, MFVF mirrors the presentation 
very closely, and is more optimized for reading than VeriSmall.

\section{Conclusion}\label{sec:conclusion}

We presented a formal definition and outlined a soundness proof of 
Featherweight VeriFast, thus hopefully achieving a clear and precise 
exposition of a core subset of the VeriFast approach for sound modular 
verification of imperative programs. We also described our executable 
definition and machine-checked soundness proof of Mechanised Featherweight 
VeriFast, a slight variant of Featherweight VeriFast, in the Coq proof 
system.

Future work includes: extending Featherweight VeriFast to include 
additional features of VeriFast, such as lemma functions, inductive 
datatypes and fixpoint functions, concurrency, fractional permissions, 
function pointers, lemma function pointers, predicate families, and 
higher-order predicates\footnote{Note that these advanced features
have already been formalized, with machine-checked soundness proofs,
separately \cite{DBLP:conf/fm/JacobsSP11,DBLP:conf/popl/JacobsP11}.};
extending the executable definition of Mechanised 
Featherweight VeriFast so that it can be used as a higher-assurance 
drop-in replacement for VeriFast to verify annotated C source code files; 
and linking the resulting tool to existing formalisations of C semantics, 
such as CompCert \cite{Leroy-Compcert-CACM}.

\subsection*{Acknowledgements}

We would like to thank the anonymous reviewers for their helpful comments. This work was supported in part by the Research Fund KU Leuven and by EU project ADVENT.

\bibliographystyle{plain}
\bibliography{fvf-lmcs}
\newpage

\theendnotes
\end{document}